\documentclass{article}

\title{Short-Circuit Logic}

\author{
	Jan A.\ Bergstra$^1$\thanks{J.A. Bergstra acknowledges
	support from NWO (project Thread Algebra for Strategic
	Interleaving).} \qquad
	Alban Ponse$^1$ \qquad 
	Daan J.C. Staudt$^2$\\[2mm]
  {\small 
  $^1$Section Theory of Computer Science, Informatics Institute}\\
  {\small $^2$Institute for Logic, Language and Computation,}\\[2mm]
 {\small  Faculty of Science, University of Amsterdam}\\
	{\small Url: \url{www.science.uva.nl/~{janb,alban}}\qquad\url{www.daanstaudt.nl}}
}

\date{}

\usepackage{tikz,stmaryrd,amssymb,amsmath,amsthm,url,color}

\newcommand{\see}{\ensuremath{\textit{se}}}
\newcommand{\sef}{\ensuremath{\textit{sef}}}

\newcommand{\ST}{\PS^\SCL}
\newcommand{\SNF}{\textit{SNF}\,}
\newcommand{\SE}{se}
\newcommand{\nfs}{\ensuremath f}
\newcommand{\EqFSCL}{\SCLe}
\newcommand{\FSCLT}{\FreeSCL}
\newcommand{\T}{\NT}
\newcommand{\Tone}{\T_{\Box}}
\newcommand{\sub}[2]{\ensuremath{[#1 \mapsto #2]}}
\newcommand{\tlef}{\unlhd}
\newcommand{\trig}{\unrhd}
\newcommand{\cd}{cd}
\newcommand{\dd}{dd}
\newcommand{\tsd}{tsd}
\newcommand{\invs}{\ensuremath g}
\newcommand{\FSCL}{\FSCLT}

\newcommand{\MSCLe}{\axname{EqMSCL}}
\newcommand{\MSCL}{\axname{MSCL}}
\newcommand{\SCLe}{\axname{EqFSCL}}
\newcommand{\SCL}{\axname{SCL}}
\newcommand{\RPSCL}{\axname{RPSCL}}
\newcommand{\CSCL}{\axname{CSCL}}
\newcommand{\SSCL}{\axname{SSCL}}
\newcommand{\SSCLe}{\axname{EqSSCL}}
\newcommand{\FreeSCL}{\axname{FSCL}}
\newcommand{\export}{~\Box~}

\newcommand{\leftand}{~
     \mathbin{\setlength{\unitlength}{1ex}
     \begin{picture}(1.4,1.8)(-.3,0)
     \put(-.6,0){$\wedge$}
     \put(-.54,-0.2){\textcolor{white}{\circle*{0.6}}}
     \put(-.54,-0.2){\circle{0.6}}
     \end{picture}
     }}
 
\newcommand{\fulland}{~
     \mathbin{\setlength{\unitlength}{1ex}
     \begin{picture}(1.4,1.8)(-.3,0)
     \put(-.6,0){$\wedge$}
     \put(-.54,-0.2){\circle*{0.6}}
     \end{picture}
     }}

\newcommand{\leftor}{~
     \mathbin{\setlength{\unitlength}{1ex}
     \begin{picture}(1.4,1.8)(-.3,0)
     \put(-.6,0){$\vee$}
     \put(-.54,1.54){\textcolor{white}{\circle*{0.6}}}
     \put(-.54,1.54){\circle{0.6}}
     \end{picture}
     }}

\newcommand{\fullor}{~
     \mathbin{\setlength{\unitlength}{1ex}
     \begin{picture}(1.4,1.8)(-.3,0)
     \put(-.6,0){$\vee$}
     \put(-.54,1.54){\circle*{0.6}}
     \end{picture}
     }}

\newcommand{\sleftand}{\leftand}
\newcommand{\sleftor}{\leftor}
\newcommand{\PS}{\ensuremath{\textit{PS}\,}}

\newcommand{\NT}{\ensuremath{\mathcal{T}}}

\newcommand{\fr}{\ensuremath{{fr}}}
\newcommand{\rp}{\ensuremath{{rp}}}
\newcommand{\con}{\ensuremath{{cr}}}

\newcommand{\mem}{\ensuremath{{mem}}}
\newcommand{\stat}{\ensuremath{{stat}}}
\newcommand{\tr}{\ensuremath{{\sf T}}}
\newcommand{\fa}{\ensuremath{{\sf F}}}
\newcommand{\true}{\ensuremath{\textit{true}}}
\newcommand{\false}{\ensuremath{\textit{false}}}

\newcommand{\lef}{\ensuremath{\triangleleft}}
\newcommand{\rig}{\ensuremath{\triangleright}}

\newtheorem{theorem}{Theorem}
\newtheorem{lemma}{Lemma}
\newtheorem{proposition}{Proposition}
\newtheorem{corollary}{Corollary}
\newtheorem{definition}{Definition}  
\theoremstyle{definition}
\newtheorem{example}{Example}
\newtheorem*{citaat}{Quote}

\newcommand{\CP}{\axname{CP}}

\newcommand{\axname}[1]{\textup{\ensuremath{\textrm{#1}}}}

\title{Short-Circuit Logic}
\begin{document}

\maketitle
\begin{abstract}
Short-circuit evaluation denotes the semantics of 
propositional connectives in which the second
argument is evaluated only if the first argument does not suffice 
to determine the value of the expression.
In programming, short-circuit evaluation is widely used, with
sequential conjunction and disjunction as primitive connectives.

A short-circuit logic is a variant of propositional logic
(PL) that can be defined with help of Hoare's 
conditional, a ternary connective
comparable to if-then-else, and that implies all identities
that follow from four basic axioms for the conditional and
can be expressed in PL
(e.g., axioms for
associativity of conjunction and 
double negation shift). 
In the absence of side effects, short-circuit evaluation
characterizes PL. However,
short-circuit evaluation
admits the possibility to model side effects
and gives rise to various different short-circuit logics.
The first extreme case 
is FSCL (free short-circuit logic), which
characterizes the setting in which evaluation of 
each atom (propositional variable) can yield a side effect.
The other extreme case is MSCL (memorizing short-circuit logic), the most
identifying variant we distinguish below PL. 
In MSCL, only a very restricted type of side 
effects can be modelled, 
while sequential conjunction is non-commutative. 
We provide axiomatizations for FSCL and MSCL.

Extending MSCL with one simple axiom yields
SSCL (static short-circuit logic, or sequential
PL), for which we also provide a completeness result.
We
briefly discuss two variants in between FSCL and MSCL,
among which a logic that admits contraction of 
atoms and of their negations. 
\\[5mm]
\emph{Keywords:}
Non-commutative conjunction,
conditional composition,
reactive valuation,
sequential connective,
short-circuit evaluation,
side effect
\end{abstract}

\newpage{\small\tableofcontents}

\section{Introduction}
\label{sec:Intro}
In the setting of imperative programming,
\emph{short-circuit} evaluation of the so-called ``Boolean operators''
is often explained by means of an example.
A typical example is the expression 
\begin{equation}
\label{eq:vb}
\texttt{(b \~{}= 0) \&\& (a/b > 18.5)}
\end{equation}
where \texttt{\&\&} is the short-circuit AND-operator, 
the programming variables \texttt{a} and 
\texttt{b} are assigned to decimal number values, and 
\texttt{(b \~{}= 0)} expresses 
that the value of \texttt{b} is not zero. 
In a state in which
\texttt{(b \~{}= 0)} evaluates to \false, the 
short-circuit evaluation of Example~\eqref{eq:vb} 
yields \false\ and the
expression \texttt{(a/b > 18.5)} is \emph{not} evaluated.
In a state where \texttt{(b \~{}= 0)} evaluates to \true, 
the short-circuit evaluation result of~\eqref{eq:vb}
is the evaluation result of the expression \texttt{(a/b > 18.5)}.
Some comments are in order here:
\begin{enumerate}
\item We view the expression
\texttt{(b \~{}= 0)}
as a \emph{propositional
variable}, or \emph{atom} for short: depending on the execution environment
it evaluates either to \true\ or to \false.
This suggests a strict correspondence 
with negation $\neg$ as used in propositional logic:
$\texttt{(b \~{}= 0)}$ and $\neg\texttt{(b = 0)}$ are equivalent.
So, the expression \texttt{(b \~{}= 0)} can be
viewed as an atom or as the negation of a more 
simple atom. 
We can also view {both}  \texttt{(b \~{}= 0)} and
\texttt{(b = 0)} as atoms and adopt
the \emph{identity} $\texttt{(b \~{}= 0)}=\neg\texttt{(b = 0)}$. 
In this paper we will adopt an equational setting and 
thus we will use notions such as ``identity",  ``equation'' and so on.
\item If the value of \texttt{b} is not equal to zero, 
the expression 
\texttt{(a/b > 18.5)} is also viewed as an atom
(and similar remarks about negation can be made).
The spirit of Example~\eqref{eq:vb} is that short-circuit evaluation
ensures that this expression is only evaluated 
\emph{if} this particular condition on \texttt{b} holds.
So, this example emphasizes that 
the connective  \texttt{\&\&}, i.e., the short-circuit 
AND-operator is  \emph{not} commutative.

\end{enumerate} 

In general, short-circuit evaluation denotes the semantics of 
propositional connectives in which the second
argument is evaluated only if the first argument does not suffice 
to determine the value of the expression.
Short-circuit evaluation is prescribed by the use of
particular connectives such as \texttt{\&\&}
and refers to the setting of propositional 
logic: expressions either yield \true\ or \false.
Hence, a natural question is this one:
\begin{quote}
What are the logical laws that axiomatize short-circuit evaluation?
\end{quote}
In order to answer this question, various settings should be distinguished. 
In the case of propositional logic, short-circuit evaluation is 
nothing more than a specific evaluation strategy (evaluation stops as soon as the 
value of the expression is determined), and conjunction and disjunction 
are commutative,
while in cases such as Example~\eqref{eq:vb}, the effect
of short-circuit evaluation is more distinctive: full evaluation
is not possible without assumptions about division by
zero\footnote{See, for example, our work in \cite{BBP12}.}, 
or about expressions being undefined.
In other words, propositional logic identifies more expressions than 
the logic underlying Example~\eqref{eq:vb}, and different degrees of
identification yield different
\emph{short-circuit logics}.
The expressions we are interested in 
are built up from atoms, constants \tr\ for
\true\ and \fa\ for \false,
and connectives prescribing short-circuit evaluation.

We listed
\emph{Side effect} as a keyword for this paper. 
This is a complex and context-dependent notion, and  not  
the point of Example~\eqref{eq:vb}. 
From now on we will only consider atoms that evaluate to
either \true\ or \false, depending on a state (or execution environment),
and we will argue that under
this restriction side effects are relevant. For example, in the programming language
\textbf{Perl}~\cite{perl}, assignments on a
scalar variable \verb+$x+ take the form \verb+($x=...)+ and can be considered
as atoms in a conditional statement. This is  
illustrated in Figure~\ref{fig:perl} where we depict a run of a small example
program in which also short-circuit disjunction occurs with the familiar 
notation \verb+||+, as well as a test \verb+($x==2)+ that is also used as an atom.
This example clearly demonstrates that the connective \texttt{\&\&} 
is not commutative due to the side effect of an assignment
(in Section~\ref{subsec:SCLe} we return to this example).
\begin{figure}
\hrule
\small
\begin{verbatim}

> perl Not.pl 

 $x=0  (assignment)
 "(($x=$x+1) && not($x=$x+1)) || $x==2" is true 

 $x=0  (assignment)
 "(not($x=$x+1) && ($x=$x+1)) || $x==2" is false 
>
\end{verbatim}
\hrule
\caption{A run of a Perl program \texttt{Not.pl}
exemplifying $\texttt{A
\&\& not(A)}\neq\texttt{not(A) \&\& A}$}
\label{fig:perl}
\end{figure}

As suggested above, short-circuit evaluation of binary connectives
can be very well characterized with help of ``if-then-else'' expressions:
\begin{align*}
x\texttt{ \&\& } y \quad
&\approx\quad \texttt{if }x\texttt{ then }y\texttt{ else }\fa\\\
x\texttt{ || } y \quad
&\approx\quad \texttt{if }x\texttt{ then }\tr\texttt{ else }y.
\end{align*}
In~\cite{BP10} we introduced \emph{Proposition Algebra} as a general setting for the 
study of such if-then-else-expressions. Following Hoare~\cite{Hoa85}, we use the
ternary connective
\[x\lef y \rig z\quad\approx \quad\texttt{if }y\texttt{ then }x\texttt{ else }z\]
and define several so-called \emph{valuation congruences}, and equational 
axiomatizations of each of these congruences. For example, the equation
$x\lef x\rig\fa=x$ (or equivalently,
$x\texttt{ \&\& } x=x$) is derivable
from some of these axiomatizations, and not from others.
In~\cite{Hoa85}, Hoare provides an equational axiomatization of 
propositional logic
using the above characterizations of the binary connectives.
Our paper~\cite{BP10}
provides the set-up to define short-circuit logics and is
therefore briefly discussed in Section~\ref{sec:Hoare}.

The further contents of the paper can be summarized as follows: 
In Section~\ref{sec:SCL} we provide a generic definition
of {short-circuit logic} (SCL) and we define 
\emph{free} SCL (\FreeSCL) as
the least identifying 
short-circuit logic we consider. As an example, 
the equation $x\texttt{ \&\& } x=x$ is not valid in \FreeSCL.
A main result is our equational
axiomatization of \FreeSCL\ and its detailed proof that is based on normal forms.
In Section~\ref{sec:Other} we define various other
short-circuit logics, among which \emph{memorizing}
SCL (\MSCL), the 
most identifying short-circuit logic below propositional logic that we
distinguish and in which $x\texttt{ \&\& } x=x$ is valid.
A second main result is the axiomatization of \MSCL. 
In the last two parts of Section~\ref{sec:Other} we consider some other
short-circuit logics, among which a short-circuited
version of propositional logic, and we propose a definition
of side effects and discuss mixed settings
in which different short-circuit logics can be used.
Section~\ref{sec:Conc} contains some conclusions, 
remarks on related work and
proposals for future work.
In Section~\ref{sec:Digr} we return to proposition algebra
and present some new results.


\section{Short-circuit evaluation and proposition algebra}
\label{sec:Hoare}
In this section we briefly discuss 
\emph{proposition algebra}~\cite{BP10}, which
has short-circuit evaluation as its natural semantics
and provides a set-up to define various different short-circuit logics.

\subsection{Hoare's conditional connective and proposition algebra}
\label{subsec:Hoare}
In 1985, Hoare introduced in the paper~\cite{Hoa85}
the ternary connective
\[x\lef y\rig z,\]
in order to characterize the `propositional calculus' 
and called this connective the 
\emph{conditional}.\footnote{Not to be 
 confused with Hoare's \emph{conditional}
 introduced in his 1985~book on CSP~\cite{Hoa85a} 
 and in his well-known 1987 paper \emph{Laws of Programming}
 \cite{HHH87} 
 for expressions $P\lef b\rig Q$ with $P$ and $Q$ programs and 
 $b$ a Boolean expression;  these sources do
 not refer to~\cite{Hoa85}
 that appeared in 1985.}
A more common expression for the conditional $x\lef y\rig z$
is
\[
\texttt{if }y\texttt{ then }x\texttt{ else }z.
\]
However, in order to reason 
systematically with conditionals, a notation
such as $x\lef y\rig z$ seems indispensable.
In $x\lef y\rig z$, first $y$ is evaluated, and depending 
on that evaluation result, then
either $x$ or $z$ is evaluated (and the other is not), 
which is a typical example of short-circuit evaluation.
In~\cite{Hoa85}, Hoare proves that propositional logic
can be equationally
characterized over the signature
$\{\tr,\fa,\_\lef\_\rig\_\}$ with constants \tr\ and \fa\
for the truth values \true\ and \false, respectively,
and he provides a set of elegant axioms to this end,
including those in Table~\ref{CP}.\footnote{In
 fact Hoare used eleven axioms; in Section~\ref{subsec:pa}
 we provide a simple equational basis for propositional
 logic based on this signature.}

\begin{table}
\centering
\hrule
\begin{align*}
\label{cp1}
\tag{CP1} x \lef \tr \rig y &= x\\
\label{cp2}\tag{CP2}
x \lef \fa \rig y &= y\\
\label{cp3}\tag{CP3}
\tr \lef x \rig \fa  &= x\\
\label{cp4}\tag{CP4}
\qquad
    x \lef (y \lef z \rig u)\rig v &= 
	(x \lef y \rig v) \lef z \rig (x \lef u \rig v)
\end{align*}
\hrule
\caption{The set \CP\ of axioms for proposition algebra}
\label{CP}
\end{table}

Given a countable set $A$ of of atoms, the set \PS\ of closed terms
over $\Sigma_{\CP}$, further called
\emph{propositional statements}, can be defined inductively: 
\[t::= \tr\mid\fa\mid a\mid t\lef t\rig t\]
where $a$ ranges over $A$. 

As stated above, a natural view on propositional statements in \PS\
involves short-circuit evaluation, similar to how we 
consider an
``$\texttt{if }y\texttt{ then }x\texttt{ else }z$'' expression.
We provide a simple form of a ``short-circuit
semantics'' taken from~\cite{Daan} that is sufficient for the \CP-case.

\begin{definition}[Evaluation Trees]
\label{def:trees}
The set \NT\ of evaluation trees over $A$ with leaves in 
$\{\tr, \fa\}$ is defined inductively:
\[\tr\in\NT,\quad\fa\in\NT, \quad (X\unlhd a\unrhd Y)\in\NT 
\text{ for any }X,Y \in \NT \text{  and } a\in A.\]
The operator $-\unlhd a\unrhd-$ is called 
\textbf{post-conditional composition over $a$}.
In the evaluation tree $X \unlhd a \unrhd Y$, 
the root is represented by $a$,
the left branch by $X$ and the right branch by $Y$. 
The depth of an evaluation tree $X$ is defined recursively by 
$d(\tr) = d(\fa) = 0$ and 
$d(Y \unlhd a \unrhd Z) = 1 + \max(d(Y ), d(Z))$. 
\end{definition}
We refer to trees in \NT\ as evaluation trees, or trees for short. 
Evaluation trees will play a crucial role in the proof of one of the
main results of this paper.

In order to define a short-circuit semantics of the conditional 
connective, we first define the \emph{leaf replacement} operator, 
`replacement' for short, on trees in \NT\ as follows. 
Let $X,X',X'',Y,Z \in\NT$ and $a \in A$. 
The replacement of \tr\ with $Y$ and $\fa$ with $Z$ in $X$, denoted
\[X[\tr\mapsto Y, \fa \mapsto Z]\]
is defined recursively by
\begin{align*}
\tr[\tr\mapsto Y,\fa\mapsto Z]&= Y,\\
\fa[\tr\mapsto Y,\fa\mapsto Z]&= Z,\\
(X'\unlhd a\unrhd X'')[\tr\mapsto Y,\fa\mapsto Z]
&=X'[\tr\mapsto Y,\fa\mapsto Z]\unlhd a\unrhd X''[\tr\mapsto Y,\fa\mapsto Z].
\end{align*}
We note that the order in which the replacements of the leaves of 
$X$ is listed inside the brackets
is irrelevant and adopt the convention of not listing any 
identities inside the brackets, e.g., 
$X[\fa\mapsto Z]=X[\tr\mapsto \tr,\fa\mapsto Z]$.
Repeated replacements satisfy the following identity:
\begin{align*}
&X[\tr\mapsto Y_1,\fa\mapsto Z_1][\tr\mapsto Y_2,\fa\mapsto Z_2]\\
&=
X[\tr\mapsto Y_1[\tr\mapsto Y_2,\fa\mapsto Z_2],
\fa\mapsto Z_1[\tr\mapsto Y_2,\fa\mapsto Z_2]].
\end{align*}
We now have the terminology and notation to formally define the 
interpretation of propositional statements in \PS\ (i.e., closed
$\Sigma_\CP$-terms) as evaluation trees
by a function $se$ (abbreviating  short-circuit evaluation).

\begin{definition}
\label{def:se}
The unary short-circuit evaluation function $se : \PS \to\NT$ 
is defined as
follows, where $a\in A$:
\begin{align*}
se(\tr) &= \tr,\\
se(\fa) &= \fa,\\
se(a)&=\tr\unlhd a\unrhd \fa,\\
se(P \lef Q\rig R)&= se(Q)[\tr\mapsto se(P), \fa\mapsto se(R)].
\end{align*}
\end{definition}
As we can see from the definition on atoms, the evaluation 
continues in the left branch if an atom yields \true\ and in 
the right branch if it yields \false, and we use the constants 
\tr\ and \fa\ to denote these truth values.
For an example see the evaluation trees in Fig.~\ref{fig:1}.
An evaluation of a propositional
statement $P$ can be characterized by a complete path in $se(P)$
(from root to leaf), including the 
evaluations of its successive atoms. 

\begin{figure}[t]
\hrule
\[\begin{array}{ccc}
\\[-3mm]
\begin{tikzpicture}[%
      level distance=7.5mm,
      level 1/.style={sibling distance=15mm},
      level 2/.style={sibling distance=7.5mm},
      baseline=(current bounding box.center)]
      \node (a) {$b$}
        child {node (b1) {$a$}
          child {node (d1) {$\tr$}} 
          child {node (d2) {$\fa$}}
         }
        child {node (b2) {$c$}
          child {node (d1) {$\tr$}} 
          child {node (d2) {$\fa$}}
        };
      \end{tikzpicture}
&&
\begin{tikzpicture}[%
      level distance=7.5mm,
      level 1/.style={sibling distance=15mm},
      level 2/.style={sibling distance=7.5mm},
      baseline=(current bounding box.center)]
      \node (a) {$b$}
        child {node (b1) {$\fa$}
        }
        child {node (b2) {$a$}
          child {node (d1) {$\tr$}} 
          child {node (d2) {$\fa$}}
        };
      \end{tikzpicture}\\\\
\text{\footnotesize The evalution tree $se(a\lef b\rig c)$}
&\qquad\quad&
\text{\footnotesize The evalution tree $se(a\lef(\fa\lef b\rig\tr) \rig \fa)$}
\end{array}
\]
\hrule
\caption{Evaluation trees, where
branches descending to the left indicate that the node
is evaluated \true\ and to the right that it yielded 
\false}
\label{fig:1}
\end{figure}


\begin{definition}[Evaluation]
\label{def:eval}
Let $P\in\PS$. An \textbf{evaluation} of $P$ is a pair 
\[(\sigma, B)\]
where $\sigma\in(A\{\tr,\fa\})^*$ and $B\in\{\tr,\fa\}$,
such that if $se(P)\in\{\tr,\fa\}$, then $\sigma=\epsilon$ (the empty string)
and $B=se(P)$, and 
otherwise, 
\[\sigma=a_1B_1a_2B_2...a_nB_n,\]
with $a_1a_2...a_nB$ is a complete path in $se(P)$ 
and
\begin{itemize}
\item
for $i<n$, if $a_{i+1}$ is a left child of $a_i$ then $B_i=\tr$, and otherwise
$B_i=\fa$,
\item
if $B$ is a left child of $a_n$ then $B_n=\tr$, and otherwise $B_i=\fa$.
\end{itemize}
We refer to $\sigma$ as the \textbf{evaluation path} and
to $B$ as the \textbf{evaluation result}. 
\end{definition}

As an example, consider $\fa\lef a\rig (\fa\lef a\rig\tr)$ and its
$se$-image
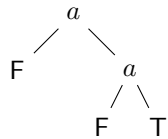
\[
\begin{tikzpicture}[%
      level distance=7.5mm,
      level 1/.style={sibling distance=15mm},
      level 2/.style={sibling distance=7.5mm},
      baseline=(current bounding box.center)]
      \node (a) {$a$}
        child {node (b1) {$\fa$}
        }
        child {node (b2) {$a$}
          child {node (d1) {$\fa$}} 
          child {node (d2) {$\tr$}}
        };
      \end{tikzpicture}
\]      
In this evaluation tree, the evaluation $(a\fa a\tr,\fa)$ expresses 
that the first occurrence
of $a$ is evaluated to \fa, the second occurence of $a$ is then 
evaluated to \tr, and the final evaluation value is \fa. 
In this way, each evaluation tree in turn gives rise to a \emph{unique}
propositional statement:

\begin{definition}
\label{def:basic}
\textbf{Basic terms} are defined by the following grammar
($a\in A$):
\[t::= \tr\mid\fa\mid t\lef a \rig t.\]
\end{definition}

The basic term associated
with the last example  is
$\fa\lef a\rig (\fa\lef a\rig\tr)$, and its $se$-image
is $\fa\unlhd a\unrhd (\fa\unlhd a\unrhd\tr)$. It is easy to
see that for each
basic term, its $se$-image has exactly the same 
syntactic structure. 
For $P,Q\in\PS$, we write
\[P=_\fr Q\quad\text{iff}\quad se(P)=se(Q),\]
and the relation $=_\fr$ is called \emph{free valuation congruence}.
In the sequel we shall
use the notion of a \emph{valuation congruence} for a congruence over \PS\
that can be associated with various forms of short-circuit evaluation. 
So, if $P=_\fr Q$, then each evaluation of $P$ yields the same result on $Q$,
and $=_\fr$ is a congruence relation. The notion `valuation congruence'
stems from~\cite{BP10}.

\begin{theorem}
\label{thm:1}
For all $P,Q\in\PS$, 
$\CP\vdash P=Q\iff P=_\fr Q$.
\end{theorem}
\begin{proof}
It is easy to show that $=_\fr$ is a congruence relation and
that all \CP-axioms are sound. For example, 
the soundness of axiom~\eqref{cp3} follows from
\[
se(\tr\lef P\rig \fa)=se(P)[\tr\mapsto \tr,\fa\mapsto \fa]=se(P),
\]
and the soundness of axiom~\eqref{cp4} from
\begin{align*}
se(P\lef&(Q\lef R\rig S)\rig U)\\
&=se(Q\lef R\rig S)[\tr\mapsto se(P),\fa\mapsto se(U)]\\
&=se(R)[\tr\mapsto se(Q),\fa\mapsto se(S)][\tr\mapsto se(P),\fa\mapsto se(U)]\\
&=se(R)[\tr\mapsto se(Q)[\tr\mapsto se(P),\fa\mapsto se(U)],\\
&\phantom{ =se(R)[}~\fa\mapsto se(S)[\tr\mapsto se(P),\fa\mapsto se(U)]]\\
&=se(R)[\tr\mapsto se(P\lef Q\rig U),\fa\mapsto se(P\lef S\rig U)]\\
&=se((P\lef Q\rig U)\lef R\rig(P\lef S\rig U)).
\end{align*}

As for $\Longleftarrow$, it was proved in~\cite{BP10} that free valuation congruence 
$=_\fr$ as defined in that paper coincides with equality of basic forms, 
and thus with $=_\fr$
as defined above. Also, in~\cite{BP10} it was proved that \CP\
axiomatizes $=_\fr$.
\end{proof}

We note that it was shown in~\cite{Chris} that the axioms of \CP\ are
independent, and also that they are 
$\omega$-complete if the set of 
atoms involved contains at least two elements.
In~\cite{BP10} we define varieties of so-called
\emph{valuation algebras} in order to axiomatize various valuation 
congruences for proposition algebra. 
All varieties discussed 
in~\cite{BP10} 
satisfy the set \CP\ of axioms (see~Table~\ref{CP}) and,
as stated above, the variety that identifies least 
is axiomatized by
exactly these four axioms. We return to valuation algebras
in Section~\ref{sec:Conc}.

\subsection{Definable connectives and their basic properties}
\label{subsec:defop}
With the conditional as a primitive connective,
negation can be defined by
\begin{equation}
\label{eq:neg}
\neg x=\fa\lef x\rig\tr,
\end{equation}
and the following consequences are easily derived 
from the extension of \CP\ with negation:
\begin{align*}
	\fa &= \neg\tr,\\
	\neg \neg x &= x,\\
	\neg (x \lef y \rig z) &= \neg x\lef y \rig \neg z,\\
	x\lef\neg y\rig z&=z\lef y\rig x.
\end{align*}
As an example, we prove the latter identity: 
\begin{align*}
x\lef\neg y\rig z
&=x\lef(\fa\lef y\rig \tr)\rig z
&&\text{by definition}\\
&=(x\lef\fa\rig z)\lef y\rig (x\lef\tr\rig z)
&&\text{by axiom~\eqref{cp4}}\\
&=z\lef y\rig x.
&&\text{by axioms~\eqref{cp2} and \eqref{cp1}}
\end{align*}

Instead of using the programming-oriented notation \texttt{\&\&}
for short-circuit conjunction, we will further use
the notation
\[x\leftand y\]
taken from~\cite{BBR95}, where the small circle at the left 
indicates that the left-argument is evaluated
first, and we shall use the name 
\emph{left-sequential conjunction} for this connective.
Left-sequential conjunction can be defined in \CP\ by
\begin{equation}
\label{eq:leftand}
x\leftand y=y\lef x\rig\fa.
\end{equation}
Left-sequential disjunction $\leftor$ (notation also taken from
\cite{BBR95})
can be defined by a left-sequential form of duality:
\begin{equation}
\label{eq:leftor}
x\leftor y = \neg(\neg x\leftand \neg y).
\end{equation}
A more convenient equation for $\leftor$ is
\begin{equation}
\label{eq:alt}
x\leftor y = \tr\lef x\rig y,
\end{equation}
the correctness of which can be shown as follows:
\begin{align*}
\neg(\neg &x\leftand \neg y)\\
&=\fa\lef((\fa\lef y\rig\tr)\lef(\fa\lef x\rig\tr)\rig\fa)\rig\tr\\
&=\fa\lef(\fa\lef x\rig(\fa\lef y\rig\tr))\rig\tr
&&\text{by \eqref{cp4}, \eqref{cp2} and \eqref{cp1}}\\
&=\tr\lef x\rig(\fa\lef(\fa\lef y\rig\tr)\rig\tr)
&&\text{by \eqref{cp4} and \eqref{cp2}}\\
&=\tr\lef x\rig(\tr\lef y\rig\fa)
&&\text{by \eqref{cp4}, \eqref{cp2} and \eqref{cp1}}\\
&=\tr\lef x\rig y.
&&\text{by \eqref{cp3}}
\end{align*}
We write $\Sigma_{\CP}(\neg,\leftand)$ for the extension of
the signature $\Sigma_\CP$ with $\neg$ and $\leftand$,
and consider $\neg$ as defined by~\eqref{eq:neg},
$\leftand$ as defined by~\eqref{eq:leftand}, and
$\leftor$ as defined by~\eqref{eq:leftor}.

The connectives $\leftand$ and $\leftor$ are associative and 
the dual of each other, where duality refers to a left-sequential
version of De~Morgan's laws.\footnote{
$\neg(x\leftand y)=\neg x\leftor \neg y$ \quad and\quad
$\neg(x\leftor y)=\neg x\leftand \neg y$.}
The associativity of ${\leftand}$ can be derived in \CP\ extended
with $\eqref{eq:neg}-\eqref{eq:leftor}$ in the following way:
\begin{align*}
(x \leftand y)\leftand z
&=z\lef(y \lef x \rig \fa)\rig\fa\\
&=(z\lef y \rig\fa)\lef x \rig (z\lef \fa\rig \fa)
&&\text{by \eqref{cp4}}\\
&=(z\lef y \rig\fa)\lef x \rig \fa
&&\text{by \eqref{cp2}}\\
&=x\leftand(y \leftand z),
\end{align*}
and duality 
immediately follows from the definition of $\leftor$.

\begin{definition}
\label{def:ext}
Extend the definition of the evaluation
function $se$ (Definition~\ref{def:se})
to closed $\Sigma_{\CP}(\neg,\leftand)$-terms 
by the following extra clauses:
\begin{align*}
\SE(\neg P)&=\SE(P)[\tr\mapsto \fa,\fa\mapsto\tr],\\
\SE(P\leftand Q)&=\SE(P)[\tr\mapsto \SE(Q)],\\
\SE(P\leftor Q)&=\SE(P)[\fa\mapsto \SE(Q)].
\end{align*}
An \textbf{evaluation} (Definition~\ref{def:eval})
now also refers to closed $\Sigma_{\CP}(\neg,\leftand)$-terms.
\end{definition}

Under this extension, the function $se$ is well-defined:
$\SE(\neg P)=\SE(\fa\lef P\rig\tr)$,
$\SE(P\leftand Q)=se(Q\lef P\rig\fa)$ and
$\SE(P\leftor Q)=se(\tr\lef P\rig Q)$.

Finally, observe that from \CP\ extended with $\eqref{eq:neg}-\eqref{eq:leftor}$,
the following equations (and their duals) are derivable:
\[\tr\leftand x=x, \quad x\leftand\tr=x,\quad
\tr\leftor x=\tr,\]
in contrast to $x\leftor\tr=\tr$ and $x\leftand\fa=\fa$,
which are not derivable: the evaluation tree of
$se(a\leftor\tr)=\tr\unlhd a\unrhd\tr$ is not equal to
$se(\tr)=\tr$ (cf.~Theorem~\ref{thm:1}).

In~\cite{BP10} we show that not each \PS-term is in \CP\ derivably equal
to one in which only the connectives $\leftand,~\leftor$ and negation
occur (next to the atoms and \tr\ and \fa). For example, $a\lef b\rig c$
cannot be expressed without the conditional connective.

\subsection{Memorizing valuation congruence}
\label{subsec:CPmem}
In~\cite{BP10} we introduced various extensions
of the axiom set
\CP. Such extensions are defined by adding axioms to \CP\
and identify
more propositional statements than those identified
by \CP.
One of these extensions is 
defined by adding this axiom to \CP:
\begin{align*}
\label{CPmem}
\tag{CPmem} 
\qquad
x\lef y\rig(z\lef u\rig(v\lef y\rig w))
&= x\lef y\rig(z\lef u\rig w).
\end{align*}
The axiom~\eqref{CPmem} 
expresses that the first evaluation value of $y$
is memorized.
We use the name ``memorizing \CP'', notation $\CP_\mem$,
for the set $\CP\cup\{\eqref{CPmem}\}$ of axioms. The signature 
of $\CP_\mem$ is $\Sigma_\CP$.

For the sake of completeness, we define in Appendix~\ref{app:Mem} evaluation trees 
that characterize \emph{memorizing evaluations} and a function $mse:\PS\to\T$
that assigns such `memorizing evaluation trees'. 
We write $P=_\mem Q$ (memorizing valuation congruence) if $P$ and $Q$
yield the same memorizing evaluation tree. We here
simply take the result (that is, Theorem~\ref{thm:mem} in Appendix~\ref{app:Mem})
as a point of departure: \emph{For all $P,Q\in\PS$},
\[\CP_\mem\vdash P=Q\iff P=_\mem Q.\]
Below we explain why we need not define memorizing valuation 
congruence ($=_\mem$) in detail at this place and why 
the above theorem is a sufficient point of departure.

In one of the forthcoming completeness proofs we will use the fact that
replacing in axiom~\eqref{CPmem} the variable $y$ by 
$\fa\lef y\rig\tr$ and/or the variable $u$ by $\fa\lef u\rig\tr$
yields various equivalent versions of this axiom, 
in particular,
\begin{align}
\label{CPmem'}
\tag{CPmem$'$}
\qquad
(x\lef y\rig(z\lef u\rig v))\lef u\rig w&=
(x\lef y\rig z)\lef u\rig w,\\
\label{CPmem''}
\tag{CPmem$''$}
x\lef y\rig((z\lef y\rig u)\lef v\rig w)
&= x\lef y\rig(u\lef v\rig w).
\end{align}
If we replace in axiom~\eqref{CPmem} $u$ by $\fa$,
we find the \emph{contraction law}
\begin{equation}
\label{eq:contr}
\qquad
x\lef y\rig(v\lef y\rig w)=x\lef y\rig w,
\end{equation} 
and replacing $y$ by $\fa\lef y\rig\tr$ 
then yields the symmetric contraction law 
\begin{equation}
\label{eq:contr2}
\qquad
(w\lef y\rig v)\lef y\rig x= w\lef y\rig x.
\end{equation}

If we extend memorizing \CP\ with the defining equations
for $\neg$ and $\leftand$ (and thus 
$\Sigma_\CP$ to $\Sigma_\CP(\neg,\leftand)$)
we can easily derive 
with~\eqref{eq:contr2} the \emph{idempotence}
of $\leftand$:
\begin{align*}
x\leftand x&=x\lef x\rig\fa\\
&=(\tr\lef x\rig\fa)\lef x\rig\fa
&&\text{by~\eqref{cp3}}\\
&=\tr\lef x\rig\fa
&&\text{by~\eqref{eq:contr2}}\\
&=x.
&&\text{by~\eqref{cp3}}
\end{align*}

An important property of $\CP_\mem$ extended with $\neg$ and $\leftand$
is that the conditional connective
can be expressed with the binary connectives and negation only:
\begin{align*}
(y\leftand x)\leftor(\neg y\leftand z)
&=\tr\lef(x\lef y\rig\fa)\rig(\fa\lef y\rig z)
&&\text{by~\eqref{eq:leftand} and \eqref{eq:alt}}\\
&=(\tr\lef x\rig(\fa\lef y\rig z))\lef y~\rig
\\
&\phantom{=~}(\fa\lef y\rig z)
&&\text{by \eqref{cp4} and \eqref{cp2}}\\
&=(\tr\lef x\rig\fa)\lef y\rig z
&&\text{by \eqref{CPmem'} and \eqref{eq:contr}}\\
&=x\lef y\rig z.
&&\text{by \eqref{cp3}}
\end{align*}
As a consequence, it is not necessary to define $=_\mem$ in terms of evaluation
trees and we can instead use the axiomatization $\CP_\mem$
for our further purposes.
Another way to express the conditional is $x\lef y\rig z=
(y\leftor z)\leftand(\neg y\leftor x)$.
In Section~\ref{subsec:MSCL} we provide axioms 
over the signature
$\{\tr,\neg,\leftand\}$ that constitute an equational basis 
for $\CP_\mem$.

With $x\leftand x=x$ it easily follows that $=_\mem$ identifies more 
than $=_\fr$ (for example, $se(a)\ne se(a\leftand a)$, so  
$a\ne_\fr a\leftand a$), 
and a typical inequality is
$a\leftand b\ne_\mem b\leftand a$ for different atoms $a$ and $b$: 
an evaluation can be such that 
$a\leftand b$ yields \emph{true} and $b\leftand a$ yields \emph{false}.

\subsection{Static valuation congruence}
\label{subsec:CPstat}
The most identifying axiomatic extension of \CP\ in~\cite{BP10}
is defined by adding to \CP\ both the axiom
\begin{align*}
\label{CPstat}\tag{CPstat} 
\qquad
(x\lef y\rig z)\lef u\rig v
&=(x\lef u\rig v)\lef y\rig (z\lef u\rig v)
\end{align*}
and the contraction law~\eqref{eq:contr2}, that is,
\begin{align*}
(x\lef y\rig z)\lef y\rig u&=x\lef y\rig u.
\end{align*}
We write $\CP_\stat$ for this extension
and we use the name ``static \CP'' for the
axioms of $\CP_\stat$.
In~\cite{BP10}
we prove that $\CP_\stat$ and Hoare's
axiomatization in \cite{Hoa85} are inter-derivable.
Furthermore, $\CP_\stat$ is also an axiomatic extension
of $\CP_\mem$: in Proposition~\ref{prop:3} 
(Section~\ref{subsec:pa}) we give a concise proof
of this fact.

The axiom \eqref{CPstat}
expresses how the order of evaluation of 
$u$ and $y$ can be swapped and thereby excludes any kind of 
side effects. Some simple examples: first, if we take $u=v=\fa$ in axiom~\eqref{CPstat}
we find
\begin{equation}
\label{eq:Hoare}
\fa=\fa\lef y\rig \fa,
\end{equation}
and with this equation we can easily derive
\begin{align*}
y\lef x\rig\fa
&=(\tr\lef y\rig\fa)\lef x\rig\fa
&&\text{by \eqref{cp3}}\\
&=(\tr\lef x\rig\fa)\lef y\rig(\fa\lef x\rig\fa)
&&\text{by \eqref{CPstat}}\\
&=x\lef y\rig\fa.
&&\text{by \eqref{cp3} and \eqref{eq:Hoare}}\
\end{align*}

The valuation congruence that is axiomatized 
by $\CP_\stat$ is called
\emph{static valuation congruence}, notation $=_\stat$, and 
coincides with any standard semantics of propositional logic: $P=_\stat Q$
iff $\overline P\leftrightarrow \overline Q$ is a tautology,
where $\overline P$ refers to Hoare's definition in~\cite{Hoa85}:
\[\overline{P\lef Q\rig R}=
(\overline P\wedge\overline Q)\vee(\neg\overline Q\wedge\overline R).\]

The fact that $=_\stat$ identifies more than $=_\mem$ 
is easily seen 
if we extend $\CP_\stat$ with left-sequential conjunction. 
The commutativity of ${\leftand}$
then immediately follows from the derivation above.
Hence $a\leftand b=_\stat b\leftand a$, while 
$a\leftand b\ne_\mem b\leftand a$ as was argued in the previous 
section.

In Section~\ref{sec:Intro} we stated that the presence 
of side effects refutes the commutativity of ${\leftand}$.
This implies that in static \CP\ and propositional logic,
it is not possible to express
propositions with side effects.
Finally, we note that
short-circuit evaluation and sequential connectives yield an
interesting perspective on static valuation congruence.
In fact, short-circuit logic turned out to be a crucial
tool in finding an axiomatization of static valuation
congruence that is more simple and elegant than $\CP_\stat$
as defined in~\cite{BP10}; we return to this point
in Section~\ref{sec:Digr}.

\section{Free short-circuit logic}
\label{sec:SCL}
In this section we provide a generic definition
of a {short-circuit logic} and a definition
of \FreeSCL\ (Free \SCL), the least identifying short-circuit logic
we consider. In Section~\ref{subsec:SCLe} we present a 
set of axioms that constitutes an equational
axiomatization of \FreeSCL, for which we use an intermediate result that we prove
in  Section~\ref{subsec:cpl}. We define normal forms in Section~\ref{subsec:snf}
and analyze the structure of the associated
$se$-trees in Section~\ref{subsec:tree}.

\subsection{A generic definition of short-circuit logics}
\label{subsec:defSCL}
We define short-circuit logics using notation from \emph{Module
algebra}~\cite{BHK90}. Intuitively, a
short-circuit logic is a logic that implies\footnote{Or, 
  if one prefers the 
  semantical point of view, ``satisfies''.} all
consequences of some \CP-axiomatization that can be expressed in the signature
$\{\tr,\neg,\leftand\}$. 
For example, in \CP\ extended with negation and the axiom that
defines negation in terms of the conditional, that is,
\[\neg x=\fa\lef x\rig\tr,\] 
we can derive $\neg\neg x=x$, as was stated
in Section~\ref{subsec:defop}. 
Our definition
below uses the export-operator $\Box$ of module
algebra to express this state of affairs in a concise way:
in module algebra, $S\export X$  is the operation that 
exports the signature $S$ from module $X$ while declaring 
other signature elements auxiliary. In this case it declares 
conditional composition to be an auxiliary operator.

\begin{definition}
\label{def:SCL}
A \textbf{short-circuit logic}
is a logic that implies the consequences
of the module expression
\begin{align*}
\SCL=\{\tr,\neg,\leftand\}\export(&\CP\\
& + \langle\, \neg x=\fa\lef x\rig\tr\,\rangle\\
& + \langle\, x\leftand y=y\lef x\rig\fa\,\rangle).
\end{align*}
\end{definition}

As a first example, $~\SCL\vdash \neg\neg x=x$~ 
can be formally proved as follows:
\begin{align*}
\neg \neg x&=\fa\lef(\fa\lef x\rig\tr)\rig\tr
&&\text{by }\langle\, \neg x=\fa\lef x\rig\tr\,\rangle\\
&=(\fa\lef\fa\rig\tr)\lef x\rig(\fa\lef\tr\rig\tr)
&&\text{by~\eqref{cp4}}\\
&=\tr\lef x\rig\fa
&&\text{by~\eqref{cp2} and \eqref{cp1}}\\
&=x.
&&\text{by~\eqref{cp3}}
\end{align*}
In Section~\ref{subsec:defop} we already derived some
more  \SCL-identities, such as the
associativity of $\leftand$ and the identities
$\tr\leftand x=x$ and $x\leftand\tr=x$.

We end this section with a few words on the constant \fa\
and the connective ${\leftor}$.
Both are not in the exported signature of \SCL,
but can be easily added to \SCL\ as defined ingredients
to enhance readability:
the constant \fa\ can be added to \SCL\ as a
shorthand for $\neg\tr$ because 
\begin{align*}
(\CP+ \langle\,\neg x=\fa\lef x\rig\tr\,\rangle)\vdash
\neg\tr&=\fa\lef \tr\rig\fa\\
&=\fa,
\end{align*}
and the connective $\leftor$ can 
be added to \SCL\
by its defining equation 
\[x\leftor y=\neg(\neg x\leftand
\neg y).\]

\subsection{Free short-circuit logic: \FreeSCL}
\label{subsec:SCLe}

Following Definition~\ref{def:SCL}, 
we now define the least identifying short-circuit logic.

\begin{definition}
\label{def:FSCL}
\textbf{$\FreeSCL$ (free short-circuit logic)}
is the short-circuit logic that implies no other 
consequences than those of the module expression \SCL.
\end{definition}

\begin{table}
\hrule
\begin{align}
\label{SCL1}
\tag{SCL1}
\fa&=\neg\tr\\[0mm]
\label{SCL2}
\tag{SCL2}
x\leftor y&=\neg(\neg x\leftand\neg y)\\[0mm]
\label{SCL3}
\tag{SCL3}
\neg\neg x&=x\\[0mm]
\label{SCL4}
\tag{SCL4}
\tr\leftand x&=x\\[0mm]
\label{SCL5}
\tag{SCL5}
x\leftand\tr&=x\\[0mm]
\label{SCL6}
\tag{SCL6}
\fa\leftand x&=\fa\\[0mm]
\label{SCL7}
\tag{SCL7}
(x\leftand y)\leftand z&=x\leftand (y\leftand z)
\\[0mm]
\label{SCL8}
\tag{SCL8}
\qquad
x\leftand \fa&= \neg x \leftand\fa
\\[0mm]
\label{SCL9}
\tag{SCL9}
(x\leftand\fa)\leftor y
&=(x\leftor\tr)\leftand y\\
\label{SCL10}
\tag{SCL10}
(x\leftand y)\leftor(z\leftand\fa)&=
(x\leftor (z\leftand\fa))\leftand(y\leftor (z\leftand\fa))
\end{align}
\hrule
\caption{\SCLe, a set of axioms for \FreeSCL}
\label{tab:SCL}
\end{table}

In Table~\ref{tab:SCL} we provide axioms for \FreeSCL\
and we use the name \SCLe\ for this set of axioms.
Some comments: as explained in the previous section,
axiom~\eqref{SCL1} defines $\fa$, and
axiom~\eqref{SCL2} introduces the connective ${\leftor}$.
Both axioms~\eqref{SCL2} and
\eqref{SCL3} imply sequential versions of 
De~Morgan's laws, which allows us to use a left-sequential version of
the duality principle. Axioms~$\eqref{SCL4}-\eqref{SCL7}$
define some standard identities. Axiom~\eqref{SCL8} 
illustrates 
a typical property of a logic that models side effects: 
although it is the case that for each closed \SCL-term $t$,
evaluation of $t\leftand\fa$ yields
\emph{false}, the evaluation of $t$ might also yield
a side effect. 
However, the same side effect and evaluation result
are obtained upon evaluation of $\neg t\leftand\fa$.
Axiom~\eqref{SCL9} characterizes the case that the right-argument
of each of the connectives is ensured to be evaluated.
Finally, observe that axiom~\eqref{SCL10} 
defines a restricted form of 
right-distributivity of $\leftor$ and (by duality) of $\leftand$.

There is a more concise set of axioms as strong as \SCLe:
replacing axioms~\eqref{SCL8} and~\eqref{SCL10} by
\begin{equation}
\label{eq:short}
\tag{SCL8+10}
(x\leftand y)\leftor(z\leftand\fa)=
(x\leftor (z\leftand\fa))\leftand(y\leftor (\neg z\leftand\fa))
\end{equation}
makes both derivable
(for~\eqref{SCL8}, take $x=\tr$ and $y=\fa$). Moreover, with
the axioms $\eqref{SCL1}-\eqref{SCL6}$ we can combine any pair of
equations
in a systematic way: say $L_1=R_1$ and $L_2=R_2$
can be combined with $u$ a fresh variable into
\[(u\leftand L_1)\leftor(\neg u \leftand L_2)=
(u\leftand R_1)\leftor(\neg u\leftand R_2).\]
However, we prefer elegance to conciseness and stick to the
axioms in Table~\ref{tab:SCL}. 

The following lemma is used in our completeness proof for 
\FSCL\ and gives an impression of how cumbersome derivations
in $\SCLe$ can be. We note that the lemma's identity
was used as an \SCLe-axiom in our earlier paper~\cite{BP12a}
and is now replaced by the current axiom~\ref{SCL8}.
\begin{lemma}
\label{lem:seqs}
$\EqFSCL\vdash (x \sleftor y) \sleftand (z \sleftand \fa) = 
(\neg x \sleftor (z
  \sleftand \fa)) \sleftand (y \sleftand (z \sleftand \fa))$.
\end{lemma}
\begin{proof}
We derive:
\begin{align*}
(x&\sleftor y) \sleftand (z \sleftand \fa) \\
&= (x \sleftor y) \sleftand ((z \sleftand \fa) \sleftand \fa)
&&\textrm{by \eqref{SCL6} and \eqref{SCL7}} \\
&= (x \sleftor y) \sleftand (\neg(z \sleftand \fa) \sleftand \fa)
&&\textrm{by \eqref{SCL8}} \\
&= ((x \sleftor y) \sleftand \neg(z \sleftand \fa)) \sleftand \fa
&&\textrm{by \eqref{SCL7}} \\
&= ((\neg x \sleftand \neg y) \sleftor (z \sleftand \fa)) \sleftand \fa
&&\textrm{by \eqref{SCL8}, \eqref{SCL2} and \eqref{SCL3}} \\
&= ((\neg x \sleftor (z \sleftand \fa)) \sleftand (\neg y \sleftor (z
  \sleftand \fa))) \sleftand \fa
&&\textrm{by \eqref{SCL10}} \\
&= (\neg x \sleftor (z \sleftand \fa)) \sleftand ((\neg y \sleftor (z
  \sleftand \fa)) \sleftand \fa)
&&\textrm{by \eqref{SCL7}} \\
&= (\neg x \sleftor (z \sleftand \fa)) \sleftand ((y \sleftand \neg(z
  \sleftand \fa)) \sleftand \fa)
&&\textrm{by \eqref{SCL8}, \eqref{SCL2} and \eqref{SCL3}} \\
&= ((\neg x \sleftor (z \sleftand \fa)) \sleftand y) \sleftand (\neg(z
  \sleftand \fa) \sleftand \fa)
&&\textrm{by \eqref{SCL7}} \\
&= ((\neg x \sleftor (z \sleftand \fa)) \sleftand y) \sleftand ((z
  \sleftand \fa) \sleftand \fa)
&&\textrm{by \eqref{SCL8}} \\
&= ((\neg x \sleftor (z \sleftand \fa)) \sleftand y) \sleftand (z
  \sleftand \fa)
&&\textrm{by \eqref{SCL7} and \eqref{SCL6}} \\
&= (\neg x \sleftor (z \sleftand \fa)) \sleftand (y \sleftand (z
  \sleftand \fa)).
&&\textrm{by \eqref{SCL7}} 
\end{align*}
\end{proof}

Below we  argue that \SCLe\ is ``sound'' and ``complete''
with respect to \FSCL.
Although this use of terminology is not fully standard, we feel it is adequate:
\emph{soundness} here means that each derivable consequence of \SCLe\ is 
valid in \FreeSCL,
while \emph{completeness} states that each valid consequence in \FreeSCL\ 
can be derived from \SCLe.

\begin{proposition}[Soundness]
\label{prop:sound}
For all \SCL-terms $t$ and $t'$,
\[\SCLe\vdash t=t'\quad\Longrightarrow\quad\FreeSCL\vdash t=t'.\]
\end{proposition}

\begin{proof}
Trivial. As an example we prove the soundness of axiom~\eqref{SCL10}, 
where we use that ${\leftor}$ can be defined
in \FreeSCL\ in exactly the same way as is done in
Table~\ref{tab:SCL} and that $x\leftor y=
\neg(\neg x\leftand\neg y)
=\tr\lef x\rig y$, where the latter identity follows in
\CP\ extended with defining equations for $\neg$ and ${\leftand}$
 (cf.\ Section~\ref{subsec:Hoare}), and where we repeatedly use 
axioms \eqref{cp1}, \eqref{cp2} and \eqref{cp4}:
\begin{align*}
(x&\leftand y)\leftor (z\leftand \fa)\\
&=\tr\lef(y\lef x\rig \fa)\rig (\fa\lef z\rig\fa)
&&\text{by definition}\\
&=(\tr\lef y\rig (\fa\lef z\rig\fa))\lef x\rig (\fa\lef z\rig \fa)
&&\text{by \eqref{cp4} and \eqref{cp2}}\\
&=(\tr\lef y\rig (\fa\lef z\rig\fa))\\
&\phantom{=~}\lef x\rig~\\
&\phantom{=~} ((\tr\lef y\rig (\fa\lef z\rig\fa))\lef(\fa\lef z\rig \fa)\rig\fa)
&&\text{by \eqref{cp4} and \eqref{cp2}}\\
&=(\tr\lef y\rig (\fa\lef z\rig\fa))\lef(\tr\lef x\rig (\fa\lef z\rig\fa))\rig\fa
&&\text{by \eqref{cp4}}\\
&=(x\leftor(z\leftand \fa))\leftand (y\leftor(z\leftand \fa)).
&&\text{by definition}
\end{align*}
\end{proof}

\begin{figure}[htb]
\hrule
{\small
\begin{verbatim}

my $x = 0;
print "\n \$x=$x  (assignment)\n";

if ( (($x=$x+1) && not($x=$x+1)) || $x==2 )
  {print " \"((\$x=\$x+1) && not(\$x=\$x+1)) || \$x==2\" is true \n\n";}
else
  {print " \"((\$x=\$x+1) && not(\$x=\$x+1)) || \$x==2\" is FALSE \n\n";}

$x = 0;
print " \$x=$x  (assignment)\n";

if ( (not($x=$x+1) && ($x=$x+1)) || $x==2 ) 
  {print " \"(not(\$x=\$x+1) && (\$x=\$x+1)) || \$x==2\" is TRUE \n";}
else
  {print " \"(not(\$x=\$x+1) && (\$x=\$x+1)) || \$x==2\" is false \n";}
\end{verbatim}
\hrule
\begin{verbatim}
> perl Not.pl 

 $x=0  (assignment)
 "(($x=$x+1) && not($x=$x+1)) || $x==2" is true 

 $x=0  (assignment)
 "(not($x=$x+1) && ($x=$x+1)) || $x==2" is false 
>
\end{verbatim}}
\hrule
\caption{The code of a Perl program \texttt{Not.pl}, followed by the display of 
an execution of \texttt{Not.pl}, which demonstrates that $\texttt{A
\&\& not(A)}\neq\texttt{not(A) \&\& A}$ }
\label{fig:Perl}
\end{figure}

\begin{example}
\label{ex:Perl}
 The programming language \textbf{Perl}~\cite{perl} can be used
to illustrate our claim that \FreeSCL\ defines a reasonable
logic because Perl's 
language definition is rather liberal 
with respect to conditionals and satisfies all
consequences of \FreeSCL.
In Perl, the simple assignment operator is written 
\texttt{=} 
and there is also an equality operator \texttt{==} 
that tests equality and
returns either \emph{true} or \emph{false}. An assignment is comparable
to a procedure that is evaluated for the side effect of 
modifying a variable and 
regardless of which kind of assignment operator is used, 
the final value of the variable on the left is returned as 
the value of the assignment as a whole.
This implies that in Perl assignments can
occur in if-then-else statements and then
the final value of the variable on the left is interpreted
as a Boolean. In particular, any number is evaluated
\emph{true} except for \texttt{0}.
For this reason, Perl can be used to demonstrate 
that certain
axioms that perhaps seem reasonable,
should \emph{not} be added to \FreeSCL, as for example
$x\leftand x=x $ and $x\leftand \neg x=\neg x\leftand x$.
It is not hard to write Perl programs that demonstrate 
the non-validity
of these identities. As an example, consider
the run \texttt{perl Not.pl} of the Perl program
\texttt{Not.pl}
depicted in Figure~\ref{fig:Perl} that shows
that 
\[\verb+A && not(A)+ = \verb+not(A) && A+\]
does not hold in Perl, not even if \verb+A+ is an 
``atom'', as in the code of \verb+Not.pl+
(in this code, \verb+\n+ prescribes a new-line).
Here we consider \verb-($x=$x+1)- and \verb+($x==2)+ as
atoms (cf.~Example~\eqref{eq:vb} in the Introduction). 
\\[1mm]
It can be argued that Perl fragments that constitute conditions provide
an implementation of \FSCL. The purpose
of this example is
to stress that all \FSCL-identities model valid equivalences
for conditions in a programming language such as \textbf{Perl}
in which the evaluation of 
expressions as boolean values (based on a standard
interpretation of built-in data types) is used to interpret the constituents 
of conditions, and that we can only 
expect a better modeling if we partition the occurring
atoms into those that might have a side effect and those that 
are tests without side effects. We return to this issue
in Section~\ref{subsec:varia}.
\end{example}
Our main result is that \SCLe\ is also complete:

\begin{theorem}[Completeness]
\label{thm:daan}
For all closed \SCL-terms $P$ and $Q$,
\[\FreeSCL\vdash P=Q\quad\Longrightarrow\quad \SCLe\vdash P=Q.\]
\end{theorem}

Before we prove this theorem, we briefly introduce the intermediate result
that underlies our proof. The theorem restricts to \emph{closed}
SCL-terms because our proof is based on properties of the evaluation trees 
of such terms.
From now on we will write 
\[\ST\]
for the set of closed \SCL-terms.
The text in the forthcoming three sections is largely taken from~\cite{Daan}:
in Section~\ref{subsec:snf} we define normal forms for $\ST$
and in Section~\ref{subsec:tree} we analyze the $se$-images 
of $\ST$-terms and provide some results on unique decompositions of such trees. In
Section~\ref{subsec:cpl}
we define an inverse function of $se$ (on the appropriate domain) with which we 
can prove our final result in that section, that is,
\[\text{(Theorem~\ref{thm:sclcpl}.) }
\textit{For all $P, Q \in \ST$, if $se(P) = se(Q)$ then $\EqFSCL \vdash P = Q$.}
\]
With this intermediate result, the proof of Theorem~\ref{thm:daan} is trivial.
\begin{proof}[Proof of Theorem~\ref{thm:daan}]
If $\FreeSCL\vdash P=Q$,
then by Definitions~\ref{def:ext} and \ref{def:SCL} and Theorem~\ref{thm:1}, 
$P=_\fr Q$, that is,  $se(P)=se(Q)$. 
By Theorem~\ref{thm:sclcpl} it follows 
that $\EqFSCL \vdash P = Q$.
\end{proof}

\subsection{Normal forms}
\label{subsec:snf}
To aid in the forthcoming proof of Theorem~\ref{thm:sclcpl}
we define normal forms for $\ST$-terms.
When considering trees in $\SE[\ST]$ (the image of $\SE$ for $\ST$-terms),
we note that some trees only have
$\tr$-leaves, some only $\fa$-leaves and some both $\tr$-leaves and
$\fa$-leaves. For any $\ST$-term $P$, \[\SE(P \sleftor \tr)\] is a tree
with only $\tr$-leaves, as can easily be seen from the definition of $\SE$.
Similarly, for any $\ST$-term $P$, $\SE(P \sleftand
\fa)$ is a tree with only $\fa$-leaves. 
The simplest trees in the
image of $\SE$ that have both types of leaves are $\SE(a)$ for $a \in A$. 

We define the grammar for our normal form
before we motivate it.

\begin{definition}
\label{def:snf}
A term $P \in \ST$ is said to be in \textbf{$\SCL$ Normal Form $(\SNF)$} if it
is generated by the following grammar.
\begin{align*}
P &::= P^\tr ~\mid~ P^\fa ~\mid~ P^\tr \sleftand P^* 
&&(\SNF\text{-terms})\\
P^\tr &::= \tr ~\mid~ (a \sleftand P^\tr) \sleftor P^\tr &&(\tr\text{-terms})\\
P^\fa &::= \fa ~\mid~ (a \sleftor P^\fa) \sleftand P^\fa &&(\fa\text{-terms})\\
P^* &::= P^c ~\mid~ P^d &&(*\text{-terms})\\[2mm]
P^\ell &::= (a \sleftand P^\tr ) \sleftor P^\fa
  ~\mid~ (\neg a \sleftand P^\tr ) \sleftor P^\fa&&(\ell\text{-terms})\\
P^c &::= P^\ell ~\mid~ P^* \sleftand P^d\\
P^d &::= P^\ell ~\mid~ P^* \sleftor P^c
\end{align*}
where $a \in A$. We refer to $P^\tr$-forms as $\tr$-terms, to $P^\fa$-forms as
$\fa$-terms,
to $P^\ell$-forms as
$\ell$-terms (the name refers to literal terms), and to $P^*$-forms as $*$-terms.
Finally, a term of the form $P^\tr \sleftand P^*$ is referred to as a
$\tr$-$*$-term.
\end{definition}

For each $\tr$-term $P$, $\SE(P)$ is a tree with only $\tr$-leaves.  
$\ST$-terms that have in their $se$-image only $\tr$-leaves will be rewritten to  
$\tr$-terms. Similarly, terms that have in their $se$-image only $\fa$-leaves 
will be rewritten to $\fa$-terms. 
Note that $\leftor$ is right-associative in $\tr$-terms, e.g.,
\[(a \leftand \tr) \leftor ((b \leftand \tr) \leftor \tr)
\quad\text{is a \tr-term, but 
$((a \leftand \tr) \leftor (b \leftand \tr)) \leftor \tr$ is not,}
\]
and that $\leftand$ is right-associative in $\fa$-terms.
Furthermore, the $se$-images of $\tr$-terms and $\fa$-terms follow a 
simple pattern: observe that for $P,Q\in P^\tr$,
$se((a\leftand P)\leftor Q)$
is of the form
\[
\begin{tikzpicture}[%
      level distance=7.5mm,
      level 1/.style={sibling distance=15mm},
      level 2/.style={sibling distance=7.5mm},
      baseline=(current bounding box.center)]
      \node (a) {$a$}
        child {node (b1) {$se(P)$}
        }
        child {node (b2) {$se(Q)$}
        };
      \end{tikzpicture}
\]    
Indeed, an alternative characterization for $P^\tr$-terms is
\[P^\tr::=\tr\mid P^\tr\lef a\rig P^\tr,\]
which also clearly demonstrates that basic forms (see Definition~\ref{def:basic})
without occurrences of \fa\
can be expressed with $\leftand$ and $\leftor$ as the only
connectives. Of course, a similar result holds for $P^\fa$-terms.

Before we discuss the $\tr$-$*$-terms | the third type of our $\SNF$-normal 
forms | we consider the $*$-terms, which are 
$\leftand$-$\leftor$-combinations of $\ell$-terms with the restriction 
that $\leftand$ and $\leftor$
associate to 
the left. This restriction is defined with help of the syntactical categories 
$P^c$ and $P^d$. 
We will sometimes
use the notation $\ell,\ell_1,\ell_2$ for $\ell$-terms (literal-terms)
to enhance readability. As an example, 
\[(\ell_1\leftand \ell_2)\leftand\ell_3\]
is  a $*$-term 
(it is in 
$P^c$-form), while
$\ell_1\leftand(\ell_2\leftand\ell_3)$ is not a $*$-term. 
We consider $\ell$-terms to be ``basic''
in $*$-terms in the sense that they are the smallest grammatical unit that 
generate
$se$-images in which both \tr\ and \fa\ occur. More precisely, the $se$-image
of an $\ell$-term
has precisely one node (its root) that has paths to both \tr\ and \fa. 

$\ST$-terms that have both \tr\ and \fa\ in their $se$-image
will be rewritten to
$\tr$-$*$-terms. A $\tr$-$*$-term is the conjunction of a
$\tr$-term 
and a $*$-term. The first conjunct is necessary to encode a term such as 
\[[a\leftor(b\leftor\tr)]\leftand c\]
where the evaluation values of $a$ and $b$ are not relevant, but where 
their side effects may influence the evaluation value of $c$, as can be clearly
seen from its $se$-image that has three different nodes 
that model the evaluation of $c$:
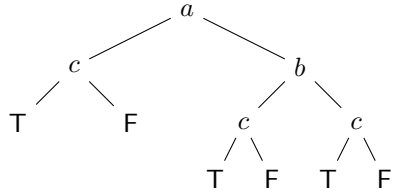
\begin{center}
\begin{tikzpicture}[%
level distance=7.5mm,
level 1/.style={sibling distance=30mm},
level 2/.style={sibling distance=15mm},
level 3/.style={sibling distance=7.5mm}
]
\node (a) {$a$}
  child {node (b1) {$c$}
    child {node (c1) {$\tr$}
    }
    child {node (c2) {$\fa$}
    }
  }
  child {node (b2) {$b$}
    child {node (c3) {$c$}
      child {node (d5) {$\tr$}} 
      child {node (d6) {$\fa$}}
    }
    child {node (c4) {$c$}
      child {node (d7) {$\tr$}} 
      child {node (d8) {$\fa$}}
    }
  };
\end{tikzpicture}
\end{center}
From this example it can be easily seen that the above 
\tr-$*$-term can be also represented
as the disjunction of a $\fa$-term and a $*$-term, namely of the \fa-term that
encodes $a\leftand(b\leftand\fa)$ and the $*$-term that encodes $c$, thus as
\[[(a\leftor\fa)\leftand((b\leftor\fa)\leftand\fa)]\leftor [(c\leftand\tr)\leftor\fa].\]
However, we chose to use a \tr-term and a conjunction for this purpose.

\bigskip

From now on we shall use $P^\tr$, $P^*$, etc.~both to denote grammatical 
categories and as
variables for terms in those categories. The remainder of this section is
concerned with defining and proving correct the normalization function 
\[\nfs:\ST \to \SNF. 
\]
We will define $\nfs$ recursively using the functions
\begin{equation*}
\nfs^n: \SNF \to \SNF \quad\text{and}\quad
\nfs^c: \SNF \times \SNF \to \SNF.
\end{equation*}
The first of these will be used to rewrite negated $\SNF$-terms to $\SNF$-terms
and the second to rewrite the conjunction of two $\SNF$-terms to an
$\SNF$-term. By \eqref{SCL2} we have no need for a dedicated function that
rewrites the disjunction of two $\SNF$-terms to an $\SNF$-term.
The normalization function $\nfs: \ST \to \SNF$ is defined
recursively, using $\nfs^n$ and $\nfs^c$, as follows.
\begin{align}
\nfs(a) &= \tr \sleftand ((a \sleftand \tr) \sleftor \fa)
  \label{eq:nfs1} \\
\nfs(\tr) &= \tr
  \label{eq:nfs2} \\
\nfs(\fa) &= \fa
  \label{eq:nfs3} \\
\nfs(\neg P) &= \nfs^n(\nfs(P))
  \label{eq:nfs4} \\
\nfs(P \sleftand Q) &= \nfs^c(\nfs(P), \nfs(Q))
  \label{eq:nfs5} \\
\nfs(P \sleftor Q) &= \nfs^n(\nfs^c(\nfs^n(\nfs(P)), \nfs^n(\nfs(Q)))).
  \label{eq:nfs6}
\end{align}
Observe that $\nfs(a)$ is indeed the unique \tr-$*$-term 
with the property that $se(a)=se(\nfs(a))$, and also that
$se(\tr)=se(\nfs(\tr))$ and $se(\fa)=se(\nfs(\fa))$
(cf.~Theorem~\ref{thm:nfs}).
 
We proceed by defining $\nfs^n$. Analyzing the semantics of $\tr$-terms and
$\fa$-terms together with the definition of $\SE$ on negations, it becomes
clear that $\nfs^n$ must turn $\tr$-terms into $\fa$-terms and vice versa.
We also remark that $\nfs^n$ must preserve the left-associativity of the
$*$-terms in $\tr$-$*$-terms, modulo the associativity within $\ell$-terms.
We define $\nfs^n: \SNF \to \SNF$ as follows, using the auxiliary function
$\nfs^n_1: P^* \to P^*$ to `push down' or `push in' the negation symbols when
negating a $\tr$-$*$-term. We note that there is no ambiguity between the
different grammatical categories present in an $\SNF$-term, i.e., any
$\SNF$-term is in exactly one of the grammatical categories identified in
Definition~\ref{def:snf}, and that all right-hand sides are of the intended 
grammatical category.

\begin{align}
\nfs^n(\tr) &= \fa
  \label{eq:nfsn1} \\
\nfs^n((a \sleftand P^\tr) \sleftor Q^\tr) &= (a \sleftor
  \nfs^n(Q^\tr)) \sleftand \nfs^n(P^\tr)
  \label{eq:nfsn2} \\[2mm]
\nfs^n(\fa) &= \tr
  \label{eq:nfsn3} \\
\nfs^n((a \sleftor P^\fa) \sleftand Q^\fa) &= (a \sleftand
  \nfs^n(Q^\fa)) \sleftor \nfs^n(P^\fa)
  \label{eq:nfsn4} \\[2mm]
\nfs^n(P^\tr \sleftand Q^*) &= P^\tr \sleftand \nfs^n_1(Q^*)
  \label{eq:nfsn5} \\[2mm]
\nfs^n_1((a \sleftand P^\tr) \sleftor Q^\fa) &= (\neg a \sleftand
  \nfs^n(Q^\fa)) \sleftor \nfs^n(P^\tr)
  \label{eq:nfsn6} \\
\nfs^n_1((\neg a \sleftand P^\tr) \sleftor Q^\fa) &= (a \sleftand
  \nfs^n(Q^\fa)) \sleftor \nfs^n(P^\tr)
  \label{eq:nfsn7} \\
\nfs^n_1(P^* \sleftand Q^d) &= \nfs^n_1(P^*) \sleftor \nfs^n_1(Q^d)
  \label{eq:nfsn8} \\
\nfs^n_1(P^* \sleftor Q^c) &= \nfs^n_1(P^*) \sleftand \nfs^n_1(Q^c).
  \label{eq:nfsn9}
\end{align}

Now we turn to defining $\nfs^c$. We distinguish the following cases:
\begin{enumerate}
\item $\nfs^c(P^\tr, Q)$
\item $\nfs^c(P^\fa, Q)$
\item $\nfs^c(P^\tr\leftand P^*, Q)$
\end{enumerate}
In case 1, it is apparent that the conjunction of a $\tr$-term with
another terms always yields a term of the same grammatical category as the
second conjunct. We define $\nfs^c$ recursively by a 
case distinction on its first argument, and in the second case by a further 
case distinction on its second argument.  
\begin{align}
\nfs^c(\tr, P) &= P
  \label{eq:nfsc1} \\
\nfs^c((a \sleftand P^\tr) \sleftor Q^\tr, R^\tr) &= (a \sleftand
  \nfs^c(P^\tr, R^\tr)) \sleftor \nfs^c(Q^\tr, R^\tr)
  \label{eq:nfsc2} \\
\nfs^c((a \sleftand P^\tr) \sleftor Q^\tr, R^\fa) &= (a \sleftor
  \nfs^c(Q^\tr, R^\fa)) \sleftand \nfs^c(P^\tr, R^\fa)
  \label{eq:nfsc3} \\
\nfs^c((a \sleftand P^\tr) \sleftor Q^\tr, R^\tr \sleftand S^*) &=
  \nfs^c((a \sleftand P^\tr) \sleftor Q^\tr, R^\tr) \sleftand S^*.
  \label{eq:nfsc4}
\end{align}

For case 2 (the first argument is an $\fa$-term) we make use
of \eqref{SCL6}. This immediately implies that the conjunction of an
$\fa$-term with another term is itself an $\fa$-term.
\begin{align}
\nfs^c(P^\fa, Q) &= P^\fa
  \label{eq:nfsc5}
\end{align}

For the remaining case 3 (the first argument is an \tr-$*$-term)
we distinguish three sub-cases:
\begin{enumerate}
\item[3.1.] The second argument is a $\tr$-term,
\item[3.2.] The second argument is a $\fa$-term, and
\item[3.3.] The second argument is a \tr-$*$-term.
\end{enumerate}
For case 3.1 we will use an auxiliary function
$\nfs^c_1: P^* \times P^\tr \to P^*$ to turn conjunctions of a $*$-term with
a $\tr$-term into $*$-terms. We define $\nfs^c_1$ recursively by a 
case distinction on its first argument. 
Together with \eqref{SCL7} (associativity) this allows us to
define $\nfs^c$ for this case. Observe that the right-hand
sides of the clauses defining $\nfs^c_1$ are indeed $*$-terms. 
\begin{align}
\nfs^c(P^\tr \sleftand Q^*, R^\tr) &= P^\tr \sleftand
  \nfs^c_1(Q^*, R^\tr) 
  \label{eq:nfsc6} \\[2mm]
\nfs^c_1((a \sleftand P^\tr) \sleftor Q^\fa, R^\tr) &= (a \sleftand
  \nfs^c(P^\tr, R^\tr)) \sleftor Q^\fa
  \label{eq:nfsc7} \\
\nfs^c_1((\neg a \sleftand P^\tr) \sleftor Q^\fa, R^\tr) &= (\neg a
  \sleftand \nfs^c(P^\tr, R^\tr)) \sleftor Q^\fa
  \label{eq:nfsc8} \\
\nfs^c_1(P^* \sleftand Q^d, R^\tr) &= P^* \sleftand \nfs^c_1(Q^d, R^\tr)
  \label{eq:nfsc9} \\
\nfs^c_1(P^* \sleftor Q^c, R^\tr) &= \nfs^c_1(P^*, R^\tr) \sleftor
  \nfs^c_1(Q^c, R^\tr).
  \label{eq:nfsc10}
\end{align}
For case 3.2 we need to define 
$\nfs^c(P^\tr \sleftand Q^*, R^\fa)$, which will 
be an
$\fa$-term. Using \eqref{SCL7} we reduce this problem to converting
$Q^*$ to an $\fa$-term, for which we use the auxiliary function
$\nfs^c_2: P^* \times P^\fa \to P^\fa$ that we define recursively by a 
case distinction on its first argument. Observe that the right-hand
sides of the clauses defining $\nfs^c_2$ are all $\fa$-terms. 
\begin{align}
\nfs^c(P^\tr \sleftand Q^*, R^\fa) &= \nfs^c(P^\tr, \nfs^c_2(Q^*,
  R^\fa))
  \label{eq:nfsc11} \\[2mm]
\nfs^c_2((a \sleftand P^\tr) \sleftor Q^\fa, R^\fa) &= (a \sleftor
  Q^\fa) \sleftand \nfs^c(P^\tr, R^\fa)
  \label{eq:nfsc12} \\
\nfs^c_2((\neg a \sleftand P^\tr) \sleftor Q^\fa, R^\fa) &= (a
  \sleftor \nfs^c(P^\tr, R^\fa)) \sleftand Q^\fa
  \label{eq:nfsc13} \\
\nfs^c_2(P^* \sleftand Q^d, R^\fa) &= \nfs^c_2(P^*, \nfs^c_2(Q^d,
  R^\fa))
  \label{eq:nfsc14} \\
\nfs^c_2(P^* \sleftor Q^c, R^\fa) &= \nfs^c_2(\nfs^n(\nfs^c_1(P^*,
  \nfs^n(R^\fa))), \nfs^c_2(Q^c, R^\fa)).
  \label{eq:nfsc15}
\end{align}
For case 3.3 we need to define $\nfs^c(P^\tr \sleftand Q^*, R^\tr \sleftand S^*)$.
We use the auxiliary function $\nfs^c_3: P^*
\times (P^\tr \sleftand P^*) \to P^*$ to ensure that the result is a
$\tr$-$*$-term, and we define $\nfs^c_3$ by a case distinction on its second argument. 
Observe that the right-hand
sides of the clauses defining $\nfs^c_3$ are all $*$-terms. 
\begin{align}
\nfs^c(P^\tr \sleftand Q^*, R^\tr \sleftand S^*) &= P^\tr \sleftand 
  \nfs^c_3(Q^*, R^\tr \sleftand S^*)
  \label{eq:nfsc16} \\[2mm]
\nfs^c_3(P^*, Q^\tr \sleftand R^\ell) &= \nfs^c_1(P^*, Q^\tr) \sleftand
  R^\ell
  \label{eq:nfsc17} \\
\nfs^c_3(P^*, Q^\tr \sleftand (R^* \sleftand S^d)) &= \nfs^c_3(P^*, Q^\tr
  \sleftand R^*) \sleftand S^d
  \label{eq:nfsc18} \\
\nfs^c_3(P^*, Q^\tr \sleftand (R^* \sleftor S^c)) &= \nfs^c_1(P^*, Q^\tr)
  \sleftand (R^* \sleftor S^c).
  \label{eq:nfsc19} 
\end{align}

\begin{theorem}[Normal forms]
\label{thm:nfs}
For any $P \in \ST$, $\nfs(P)$ terminates, $\nfs(P) \in \SNF$ and 
\[\EqFSCL\vdash \nfs(P) = P.\]
\end{theorem}

In Appendix~\ref{app:nf} we first prove a number of lemmas showing that the
definitions $\nfs^n$ and $\nfs^c$ are correct and use those to prove the above
theorem. We have chosen to use a function rather than a rewriting system to
prove the correctness of the normal form, because 
this relieves us of the
task of proving confluence for the underlying rewriting system.


%
%
%
\subsection{Tree structure and decompositions}
\label{subsec:tree}
In Section~\ref{subsec:cpl} we will prove that on $\SNF$ we can invert the function $\SE$. 
To do this we need to
prove several structural properties of the trees in $se[\SNF]$,
the image of $\SE$. In the
definition of $\SE$ we can see how $\SE(P \sleftand Q)$ is assembled from
$\SE(P)$ and $\SE(Q)$ and similarly for $\SE(P \sleftor Q)$. To decompose 
trees in $se[\SNF]$ we will introduce some notation. 
The trees in the image of $\SE$ are all
finite binary trees over $A$ with leaves in $\{\tr, \fa\}$, i.e.,
$\SE[\ST] \subseteq \T$. We will now also consider the set $\Tone$ of binary
trees over $A$ with leaves in $\{\tr, \fa, \Box\}$, where $\Box$ is 
called ``box". 
The box will be used as a placeholder when composing or
decomposing trees. Replacement of the leaves of trees in $\Tone$ by trees in
$\T$ or $\Tone$ is defined analogous to replacement for trees in $\T$, adopting
the same notational conventions.
As a first example, we have by definition of $\SE$ that 
$\SE(P \sleftand Q)$ can be
decomposed as
\begin{equation*}
\SE(P)\sub{\tr}{\Box}\sub{\Box}{\SE(Q)},
\end{equation*}
where $\SE(P)\sub{\tr}{\Box} \in \Tone$ and $\SE(Q) \in \T$. We note that
this only works because the trees in the image of $\SE$, or in $\T$ in general,
do not contain any boxes. 
Of course, each tree $X\in\T$ has the \emph{trivial 
decomposition} that involves a replacement of the form $\sub{\Box}Y$, namely 
\[\Box\sub{\Box}X.~\footnote{Also, for each $X\in\T$ it follows
 that $X=X\sub{\Box}Y$ for any $Y\in\T$, but
 we do not consider $X\sub{\Box}Y$ to be a `decomposition' of $X$ in
 this case.}
\]

We start with some simple properties
of the $\SE$-images of \tr-terms, \fa-terms, and
$*$-terms.

\begin{lemma}[Leaf occurrences]
\label{lem:TF}~
\begin{enumerate}
\item
For any
\tr-term $P$, $\SE(P)$ contains  $\tr$, but not $\fa$,
\item
For any
\fa-term $P$, $\SE(P)$ contains  $\fa$, but not $\tr$,
\item
For any
$*$-term $P$, $\SE(P)$ contains both $\tr$ and $\fa$.
\end{enumerate}
\end{lemma}

\begin{proof}
By induction on the structure of $P$. A proof of the first two statements is trivial.
For the third statement, if $P$ is an $\ell$-term, we find that by definition
of the grammar of $P$ that
one branch from the root of $\SE(P)$ will only contain $\tr$ and not $\fa$,
and for the other branch this is the other way around. 

For the induction we have to consider both 
$\SE(P_1 \sleftand P_2)$ and $\SE(P_1\sleftor P_2)$.
Consider $\SE(P_1 \sleftand P_2)$, which equals by definition $se(P_1)\sub{\tr}{se(P_2)}$.
By induction, both $se(P_1)$ and $se(P_2)$ contain both \tr\ and \fa,
so $\SE(P_1 \sleftand P_2)$ contains both \tr\ and \fa.
The case $\SE(P_1\sleftor P_2)$ can be dealt with in a similar way.
\end{proof}

Decompositions of the $se$-image
of $*$-terms turn out to be crucial in our approach. As an example, the $se$-image
of the $*$-term 
\[(\ell_1\leftor\ell_2)\leftand\ell_3\quad\text{with}\quad 
\ell_i=((a_i\leftand\tr)\leftor\fa)\]
can be decomposed as $X_1\sub{\Box}Y$ with $X_1\in\Tone$ as follows:
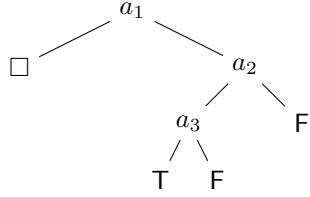
\[
\begin{tikzpicture}[%
level distance=7.5mm,
level 1/.style={sibling distance=30mm},
level 2/.style={sibling distance=15mm},
level 3/.style={sibling distance=7.5mm}
]
\node (a) {$a_1$}
  child {node (b1) {$\Box$}
  }
  child {node (b2) {$a_2$}
    child {node (c3) {$a_3$}
      child {node (d5) {$\tr$}} 
      child {node (d6) {$\fa$}}
    }
    child {node (c4) {$\fa$}
    }
  };
\end{tikzpicture}
\]
and with $Y=se(\ell_3)$, thus $Y=\tr\unlhd a_3\unrhd\fa$, and
two other decompositions are $X_2\sub{\Box}Y=
X_3\sub{\Box}Y$ with $X_2,X_3\in\Tone$ as follows:
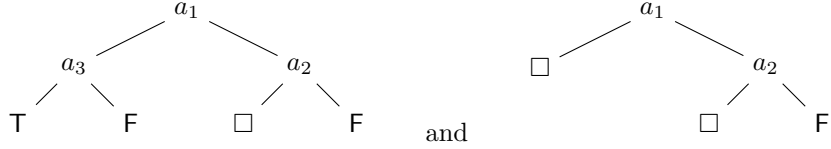
\[
\begin{tikzpicture}[%
level distance=7.5mm,
level 1/.style={sibling distance=30mm},
level 2/.style={sibling distance=15mm}
]
\node (a) {$a_1$}
  child {node (b1) {$a_3$}
    child {node (c1) {$\tr$}
    }
    child {node (c2) {$\fa$}
    }
  }
  child {node (b2) {$a_2$}
    child {node (c3) {$\Box$}
    }
    child {node (c4) {$\fa$}
    }
  };
\end{tikzpicture}
\qquad\text{and}\qquad
\begin{tikzpicture}[%
level distance=7.5mm,
level 1/.style={sibling distance=30mm},
level 2/.style={sibling distance=15mm}
]
\node (a) {$a_1$}
  child {node (b1) {$\Box$}
  }
  child {node (b2) {$a_2$}
    child {node (c3) {$\Box$}
    }
    child {node (c4) {$\fa$}
    }
  };
\end{tikzpicture}
\]
Observe that the first two decompositions have the property that $Y$ is a 
subtree of $X_1$ and $X_2$, respectively.
Furthermore, observe that $X_3=se(\ell_1\leftor\ell_2)\sub{\tr}{\Box}$, and 
hence that
this decomposition agrees with the definition of
the function $se$.
When we want to express that a certain decomposition $X\sub{\Box}Y$ has the property
that $Y$ is not a subtree of $X$, we say that $X\sub{\Box}Y$ 
is a \emph{strict decomposition}. Finally observe that each of these
decompositions satisfy the property that $X_i$ contains \tr\ or \fa, which
is a general property of decompositions of $*$-terms and a consequence of 
Lemma~\ref{lem:snondectf} (see below). The following lemma
provides the $\SE$-image of the rightmost $\ell$-term in a $*$-term as a witness.

\begin{lemma}[Witness decomposition]
\label{lem:sperttf}
For all $*$-terms $P$, $\SE(P)$ can be decomposed as $X\sub{\Box}{Y}$ with $X
\in \Tone$ and $Y \in \T$ such that $X$ contains $\Box$ and $Y = \SE(R)$ for
the rightmost $\ell$-term $R$ in $P$. 
Note that $X$ may be $\Box$. 
\\[1mm]
We will refer to $Y$ as \textbf{the witness} for this lemma for $P$.
\end{lemma}

\begin{proof}
By induction on the number of $\ell$-terms in $P$. 
In the base case $P$ is an $\ell$-term and $\SE(P) =
\Box\sub{\Box}{\SE(P)}$ is the desired decomposition. 
For the induction we have to consider both $\SE(P \sleftand
Q)$ and $\SE(P \sleftor Q)$.

We start with $\SE(P \sleftand Q)$ and let $X\sub{\Box}{Y}$ be the
decomposition for $\SE(Q)$ which we have by induction hypothesis, so
$Y$ is the witness for this lemma for $Q$ and the $se$-image of its
rightmost $\ell$-term, say $R$. Since by
definition of $\SE$ on ${\sleftand}$ we have
\begin{equation*}
\SE(P \sleftand Q) = \SE(P)\sub{\tr}{\SE(Q)}
\end{equation*}
we also have
\begin{equation*}
\SE(P \sleftand Q) = \SE(P)\sub{\tr}{X\sub{\Box}{Y}} =
\SE(P)\sub{\tr}{X}\sub{\Box}{Y}.
\end{equation*}
The last equality is due to the fact that $\SE(P)$ does not contain any boxes.
This gives our desired decomposition: $\SE(P)\sub{\tr}{X}$ contains $\Box$ because
$\SE(P)$ contains \tr\ (Lemma~\ref{lem:TF}) and $X$ contains $\Box$,
and $Y$ is the $se$-image of the 
rightmost $\ell$-term $R$ in $P\leftand Q$.

The case for $\SE(P \sleftor Q)$ is
analogous.
\end{proof}

The following lemma illustrates another structural property of trees in the
image of $*$-terms under $\SE$, namely that each non-trivial decomposition
$X\sub{\Box}{Y}$ of a $*$-term has the property that at least one of \tr\ and \fa\
occurs in $X$.

\begin{lemma}[Non-decomposition]
\label{lem:snondectf}
There is no $*$-term $P$ such that $\SE(P)$ can be decomposed as
$X\sub{\Box}{Y}$ with $X \in \Tone$ and $Y \in \T$, where $X \neq \Box$ and $X$
contains $\Box$, but not $\tr$ or $\fa$.
\end{lemma}

\begin{proof}
We prove the lemma's statement by induction on the number of $\ell$-terms in $P$.
Let $P$ be a single $\ell$-term. When we analyze the grammar of $P$ we find that
one branch from the root of $\SE(P)$ only contains $\tr$ and not $\fa$,
and the other way around for the other branch. 
Hence if $\SE(P) = X\sub{\Box}{Y}$ and $X$ does not contain \tr\ or \fa,
then $Y$ contains occurrences of both $\tr$ and $\fa$. Hence, $Y$ must contain the
root and $X = \Box$.

For the induction we
assume that the lemma holds for all $*$-terms that contain less $\ell$-terms than $P
\sleftand Q$ and $P \sleftor Q$.
We start with the case for $\SE(P \sleftand Q)$. Towards a contradiction, suppose 
that for some $*$-terms $P$ and $Q$,
\begin{equation}
\label{eq:sub}
\SE(P\sleftand Q) = X\sub{\Box}{Y}
\end{equation}
with $X \neq \Box$ and $X$ not containing
any occurrences of $\tr$ or $\fa$. 
Let $Z$ be the witness of Lemma
\ref{lem:sperttf} for $P$
(so one branch of the root of $Z$ contains only \fa-leaves,
and the other only \tr-leaves). Observe that $\SE(P \sleftand Q)$ has one or
more occurrences of the subtree
\[Z\sub{\tr}{\SE(Q)}.\]
The interest of this observation is that one branch of the root of this 
subtree contains only \fa, and the other branch contains both \tr\ and \fa\
(because $se(Q)$ does by Lemma~\ref{lem:TF}). It follows that all occurrences of 
$Z\sub{\tr}{\SE(Q)}$
in $se(P\leftand Q)$ are subtrees in $Y$ after being substituted in $X$:
\begin{itemize}
\item
Because $X$ does not contain \tr\ and \fa, Lemma~\ref{lem:TF} and \eqref{eq:sub}
imply that $Y$ contains both \tr\ and \fa.
\item
Assume there is an occurrence of 
$Z\sub{\tr}{\SE(Q)}$ in $X\sub{\Box}{Y}$ that has its root in $X$.
Hence the parts of the two branches from this root node that are in $X$
must have $\Box$ as their leaves. 
For the branch that only has \fa-leaves this implies that 
$Y$ does not contain \tr, which is a contradiction. 
\end{itemize}
So, $Y$ contains at least one occurrence of $Z\sub{\tr}{\SE(Q)}$, hence
\begin{equation}
\label{JAB}
\text{$se(Q)$ is a \emph{proper} subtree of $Y$.}
\end{equation}
This implies that \emph{each} occurrence of $se(Q)$ in $se(P\leftand Q)$
is an occurrence in $Y$ (after being substituted): if this were
not the case, the root of $se(Q)$ occurs also in $X$
and the parts of the two branches from this node that are in $X$
must have $\Box$ as their leaves, which implies that $Y$ after being substituted
in $X$ is a proper subtree of $se(Q)$. 
By~\eqref{JAB} this implies 
that $se(Q)$ is a proper subtree of itself, which is a contradiction.

Because each occurrence of $se(Q)$ in $se(P\leftand Q)= X\sub{\Box}{Y}$
is an occurrence in $Y$ (after being substituted) and the fact that 
$se(P\leftand Q)=se(P)\sub{\tr}{se(Q)}$, 
it follows that $\SE(P) = X\sub{\Box}{V}$ where $V$ is obtained from $Y$ by
replacing all occurrences of the subtree ${\SE(Q)}$ by \tr. But this
violates the induction hypothesis. 
This concludes the induction step for the case of $se(P\leftand Q)$. 

A proof for the case $\SE(P \sleftor Q)$ is symmetric.
\end{proof}

We now arrive at two crucial definitions concerning decompositions. When
considering $*$-terms, we already know that $\SE(P \sleftand Q)$ can be
decomposed as
\begin{equation*}
\SE(P)\sub{\tr}{\Box}\sub{\Box}{\SE(Q)}.
\end{equation*}
Our goal now is to give a definition for a kind of decomposition so that this
is the only such decomposition for $\SE(P \sleftand Q)$. We also ensure that
$\SE(P \sleftor Q)$ does not have a decomposition of that kind, so that we can
distinguish $\SE(P \sleftand Q)$ from $\SE(P \sleftor Q)$. Similarly, we need
to define another kind of decomposition so that $\SE(P \sleftor Q)$ can only be
decomposed as 
\begin{equation*}
\SE(P)\sub{\fa}{\Box}\sub{\Box}{\SE(Q)}
\end{equation*}
and that $\SE(P \sleftand Q)$ does not have a decomposition of that kind.

\begin{definition}
\label{def:candidate}
The pair $(Y, Z) \in \Tone \times \T$ is a \textbf{candidate conjunction
decomposition (ccd)} of $X \in \T$, if
\begin{itemize}
\item $X = Y\sub{\Box}{Z}$,
\item $Y$ contains $\Box$,
\item $Y$ contains $\fa$, but not $\tr$, and
\item $Z$ contains both $\tr$ and $\fa$.
\end{itemize}
Similarly, $(Y, Z)$ is a \textbf{candidate disjunction decomposition (cdd)} of
$X$, if
\begin{itemize}
\item $X = Y\sub{\Box}{Z}$,
\item $Y$ contains $\Box$,
\item $Y$ contains $\tr$, but not $\fa$, and
\item $Z$ contains both $\tr$ and $\fa$.
\end{itemize}
\end{definition}

Observe that any ccd or cdd $(Y, Z)$ is \emph{strict} because $Z$
contains both
$\tr$ and \fa, and thus cannot be a subtree of $Y$.
A first, crucial property of ccd's and cdd's is the following connection
with $se$-images of $*$-terms.

\begin{lemma}
\label{lem:nocdd}
For any $*$-term $P \sleftand Q$,
$\SE(P \sleftand Q)$ has no cdd. 
Similarly, for any $*$-term $P \sleftor Q$, 
$\SE(P \sleftor Q)$ has no ccd.
\end{lemma}

\begin{proof}
We first treat the case for $P \sleftand Q$,
so $P \in P^*$ and $Q \in P^d$.
Towards a contradiction, suppose that $(Y, Z)$ is a cdd of $\SE(P \sleftand Q)$. 
Let $Z'$ be the witness of Lemma~\ref{lem:sperttf} for $P$.
Observe that $\SE(P \sleftand Q)$ has one or
more occurrences of the subtree
\[Z'\sub{\tr}{\SE(Q)}.\]
It follows that all occurrences
of $Z'\sub{\tr}{se(Q)}$ in $se(P\leftand Q)$ are subtrees in $Z$ after being 
substituted in $Y$, which
can be argued in a similar way as in the proof of Lemma~\ref{lem:snondectf}:
\begin{itemize}
\item Assume there is an occurrence of
$Z'\sub{\tr}{se(Q)}$ in $Y\sub{\Box}{Z}$ that has its root in $Y$. Following
the branch from this node that only has \fa-leaves and that leads in $Y$
to one or more $\Box$-leaves, this implies 
that $Z$ does not contain \tr, which is a contradiction by definition of a cdd.
\end{itemize}
So, $Z$ contains at least one occurrence of $Z'\sub{\tr}{se(Q)}$. This implies 
that \emph{each} occurrence of $se(Q)$ in $se(P\leftand Q)$ is an occurrence in 
$Z$ (after being substituted): if this were
not the case, the root of $se(Q)$ occurs in $Y$ and this
implies that $se(Q)$ is a proper subtree of itself, which is a contradiction.
By definition of $se$, all the occurrences of $\tr$ in $\SE(P \sleftand Q)$
are in occurrences of the subtree $\SE(Q)$. 
Because $Y$ does not contain the root of an $se(Q)$-occurrence,
$Y$ does not contain any occurrences of
$\tr$, so $(Y,Z)$ is not a cdd of $\SE(P \sleftand Q)$.  
A proof
for the case $\SE(P \sleftor Q)$ is symmetric.
\end{proof}

However, the ccd and cdd are not necessarily the decompositions we are
looking for, because, for example, $\SE((P \sleftand Q) \sleftand R)$ has a ccd
\[(\SE(P)\sub{\tr}{\Box}, \SE(Q \sleftand R)),\] 
while the decomposition we
need to reconstruct the constituents of a $*$-term is
\[(\SE(P \sleftand Q)\sub{\tr}{\Box}, \SE(R)).\]
A more intricate example of a ccd $(Y, Z)$ 
that does not produce the constituents of a $*$-term
is this pair of trees $Y$ and $Z$: 
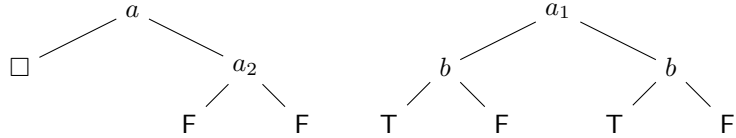
\[
\begin{tikzpicture}[%
level distance=7.5mm,
level 1/.style={sibling distance=30mm},
level 2/.style={sibling distance=15mm},
level 3/.style={sibling distance=7.5mm}
]
\node (a) {$a$}
  child {node (b1) {$\Box$}
  }
  child {node (b2) {$a_2$}
    child {node (c3) {$\fa$}
    }
    child {node (c4) {$\fa$}
          }
  };
\end{tikzpicture}
\qquad
\begin{tikzpicture}[%
level distance=7.5mm,
level 1/.style={sibling distance=30mm},
level 2/.style={sibling distance=15mm},
level 3/.style={sibling distance=7.5mm}
]
\node (a) {$a_1$}
  child {node (b1) {$b$}
    child {node (c1) {$\tr$}
    }
    child {node (c2) {$\fa$}
    }
  }
  child {node (b2) {$b$}
    child {node (c3) {$\tr$}
    }
    child {node (c4) {$\fa$}
          }
  };
\end{tikzpicture}
\]
It is clear that $(Y,Z)$ is a ccd of $se(\ell_1\leftand \ell_2)$ with $\ell_1$ and $\ell_2$ these $\ell$-terms:
\[
\ell_1=(a\leftand ((a_1\leftand\tr)\leftor\tr))\leftor((a_2\leftor\fa)\leftand\fa),
\quad
\ell_2=(b\leftand \tr)\leftor\fa.
\]
Therefore we refine
Definition~\ref{def:candidate} to obtain the decompositions we seek.

\begin{definition}
The pair $(Y, Z) \in \Tone \times \T$ is a \textbf{conjunction decomposition
(cd)} of $X \in \T$, if it is a ccd of $X$ and there is no other ccd $(Y', Z')$
of $X$ where the depth of $Z'$ is smaller than that of $Z$. 

Similarly, the pair $(Y, Z) \in \Tone \times \T$
is a \textbf{disjunction decomposition (dd)} of $X$, if it is a cdd of $X$ and
there is no other cdd $(Y', Z')$ of $X$ where the depth of $Z'$ is smaller than
that of $Z$.
\end{definition}

\begin{theorem}
\label{thm:scddd}
For any $*$-term $P \sleftand Q$, i.e., with $P \in P^*$ and $Q \in P^d$,
$\SE(P \sleftand Q)$ has the unique cd
\begin{equation*}
(\SE(P)\sub{\tr}{\Box}, \SE(Q))
\end{equation*}
and no dd. For any $*$-term $P \sleftor Q$, i.e., with $P \in P^*$ and $Q \in
P^c$, $\SE(P \sleftor Q)$ has no cd and its unique dd is
\begin{equation*}
(\SE(P)\sub{\fa}{\Box}, \SE(Q)).
\end{equation*}
\end{theorem}

\begin{proof}
By simultaneous induction on the number of $\ell$-terms
in $P \sleftand Q$ and $P \sleftor Q$. 

In the basis we have to consider, for $\ell$-terms
$\ell_1$ and $\ell_2$, the terms $\ell_1 \sleftand \ell_2$ and $\ell_1 \sleftor 
\ell_2$. By symmetry, it is sufficient to consider the first case. 
By definition of a ccd and Lemma~\ref{lem:TF},
$(\SE(\ell_1)\sub{\tr}{\Box}, \SE(\ell_2))$ is a ccd of
$\SE(\ell_1 \sleftand \ell_2)$. Furthermore observe that the smallest subtree
in $se(\ell_1\leftand \ell_2)$ that contains both \tr\ and \fa\ is
$se(\ell_2)$. 
Therefore $(\SE(\ell_1)\sub{\tr}{\Box}, \SE(\ell_2))$ is the \emph{unique} cd of
$\SE(\ell_1\sleftand \ell_2)$.
Now for the dd. It suffices to show that there is no cdd of $\SE(\ell_1 \sleftand
\ell_2)$ and this follows from Lemma~\ref{lem:nocdd}. 

For the induction we assume that the
theorem holds for all $*$-terms with less $\ell$-terms than $P \sleftand Q$ and
$P \sleftor Q$. 
We will first treat the case for $P \sleftand Q$ and show that
$(\SE(P)\sub{\tr}{\Box}, \SE(Q))$ is the unique cd of $\SE(P \sleftand Q)$. 
In this case, observe that for any other ccd $(Y, Z)$ either $Z$ 
is a proper subtree of
$\SE(Q)$, or vice versa: if this were not the case, then there are occurrences of
$Z$ and $se(Q)$ in $Y\sub{\Box}Z=se(P\leftand Q)$ 
that are disjoint and at least one of the following cases applies:
\begin{itemize}
\item $Y$ contains an occurrence of $se(Q)$, and hence of \tr, which is a 
contradiction.
\item $se(P)\sub{\tr}{\Box}$ contains an occurrence of $Z$, and hence of \tr,
which is a contradiction.
\end{itemize}
Hence, by definition of a cd
it suffices to show that there is no ccd $(Y, Z)$ where $Z$ is
a proper subtree of $\SE(Q)$. 
Towards a contradiction, suppose that such a ccd $(Y,
Z)$ does exist.
By definition of $*$-terms $Q$ is either an $\ell$-term or a disjunction. 
\begin{itemize}
\item
If
$Q$ is an $\ell$-term and $Z$ a proper subtree of $\SE(Q)$, then $Z$
does not contain both $\tr$ and $\fa$ because one branch from the root of 
$\SE(Q)$ will only contain $\tr$ and not $\fa$,
and the other branch vice versa. 
Therefore $(\SE(P)\sub{\tr}{\Box}, \SE(Q))$ is the \emph{unique} cd of
$\SE(P \sleftand Q)$.
\item
If $Q$ is a disjunction and $Z$ a proper subtree of $\SE(Q)$, then we can
decompose $\SE(Q)$ as $\SE(Q) = U\sub{\Box}{Z}$ for some $U \in \Tone$ that
contains but is not equal to $\Box$ and such that $U\sub{\Box}{Z}$ is strict, 
i.e., $Z$ is not a subtree of $U$. By Lemma~\ref{lem:snondectf} this implies
that $U$ contains either $\tr$ or $\fa$. 
\begin{itemize}
\item
If $U$ contains $\tr$, then so
does $Y$, because 
$Y = \SE(P)\sub{\tr}{U}$, which is the case
because 
\begin{align*}
Y\sub{\Box}Z&=se(P\leftand Q)\\
&=se(P)\sub{\tr}{U\sub{\Box}Z}\\
&=se(P)\sub{\tr}{U}\sub{\Box}Z,
\end{align*}
and the only way in which $Y\ne\SE(P)\sub{\tr}{U}$ is possible is that 
$U$ contains an occurrence of $Z$, which is excluded because
$U\sub{\Box}{Z}$ is strict. Because $Y$ contains an occurrence of \tr,
$(Y, Z)$ is not a ccd of $\SE(P \sleftand Q)$.
\item
If $U$ only contains
$\fa$ then $(U, Z)$ is a ccd of $\SE(Q)$ which violates the induction
hypothesis. 
\end{itemize}
Therefore $(\SE(P)\sub{\tr}{\Box}, \SE(Q))$ is the \emph{unique}
cd of $\SE(P \sleftand Q)$.
\end{itemize}
Now for the dd. By Lemma~\ref{lem:nocdd} there is no cdd of $\SE(P \sleftand
Q)$, so there is neither a dd of $\SE(P \sleftand Q)$. 
A proof
for the case $\SE(P \sleftor Q)$ is symmetric.
\end{proof}

At this point we have the tools necessary to invert $\SE$ on $*$-terms, at
least down to the level of $\ell$-terms. We can easily detect if a tree in the
image of $\SE$ is in the image of $P^\ell$, because all leaves to the left of
the root are one truth value, while all the leaves to the right are the other.
To invert $\SE$ on $\tr$-$*$-terms we still need to be able to reconstruct
$\SE(P^\tr)$ and $\SE(Q^*)$ from $\SE(P^\tr \sleftand Q^*)$. To this end we
define a $\tr$-$*$-decomposition, and as with cd's and dd's we first define a
candidate $\tr$-$*$-decomposition.

\begin{definition}
The pair $(Y, Z) \in \Tone \times \T$ is a \textbf{candidate
$\tr$-$*$-decomposition (ctsd)} of $X \in \T$, if 
\begin{itemize}
\item
$X = Y\sub{\Box}{Z}$, 
\item 
$Y$
does not contain $\tr$ or $\fa$,
\item $Z$ contains both \tr\ and \fa,
\end{itemize}
and there is no decomposition $(U, V) \in
\Tone \times \T$ of $Z$ such that
\begin{itemize}
\item $Z = U\sub{\Box}{V}$,
\item $U$ contains $\Box$,
\item $U \neq \Box$, and
\item $U$ contains neither $\tr$ nor $\fa$.
\end{itemize}
\end{definition}

However, this is not necessarily the decomposition we seek in this case.
Consider for example the $\tr$-term $P^\tr$ with the following semantics:
\begin{center}
\begin{tikzpicture}[%
level distance=7.5mm,
level 1/.style={sibling distance=30mm},
level 2/.style={sibling distance=15mm},
level 3/.style={sibling distance=7.5mm}
]
\node (a) {$a$}
  child {node (b1) {$b$}
    child {node (c1) {$c$}
      child {node (d1) {$\tr$}} 
      child {node (d2) {$\tr$}}
    }
    child {node (c2) {$d$}
      child {node (d3) {$\tr$}} 
      child {node (d4) {$\tr$}}
    }
  }
  child {node (b2) {$b$}
    child {node (c3) {$c$}
      child {node (d5) {$\tr$}} 
      child {node (d6) {$\tr$}}
    }
    child {node (c4) {$d$}
      child {node (d7) {$\tr$}} 
      child {node (d8) {$\tr$}}
    }
  };
\end{tikzpicture}
\end{center}
and observe that
$\SE(P^\tr \sleftand Q^*)$ has a ctsd
\begin{equation*}
(\Box \tlef a \trig \Box, (\SE(Q^*) \tlef c \trig \SE(Q^*)) \tlef b \trig
(\SE(Q^*) \tlef d \trig \SE(Q^*))).
\end{equation*}
But the decomposition we seek is $(\SE(P^\tr)\sub{\tr}{\Box}, \SE(Q^*))$.
Hence we will refine the above definition to aid in the theorem below.

\begin{definition}
The pair $(Y, Z) \in \Tone \times \T$ is a \textbf{$\tr$-$*$-decomposition
(tsd)} of $X \in \T$, if it is a ctsd of $X$ and there is no other ctsd $(Y',
Z')$ of $X$ where the depth of $Z'$ is smaller than that of $Z$.
\end{definition}

\begin{theorem}
\label{thm:stsd}
For any $\tr$-term $P$ and $*$-term $Q$ the unique tsd of $\SE(P \sleftand
Q)$ is
\begin{equation*}
(\SE(P)\sub{\tr}{\Box}, \SE(Q)).
\end{equation*}
\end{theorem}
\begin{proof}
First observe that $(\SE(P)\sub{\tr}{\Box}, \SE(Q))$ is a ctsd because by
definition $\SE(P)\sub{\tr}{\SE(Q)} = \SE(P
\sleftand Q)$ and $\SE(Q)$ is non-decomposable by Lemma~\ref{lem:snondectf}.

Towards a contradiction, suppose there exists a ctsd $(Y, Z)$ such that the depth of
$Z$ is smaller than that of $\SE(Q)$. Now either $Z$ is a proper subtree of
$\SE(Q)$, or vice versa, for otherwise there would be occurrences of
$Z$ and $se(Q)$ in $Y\sub{\Box}Z=\SE(P)\sub{\tr}{\SE(Q)}$ 
that are disjoint and at least one of the following cases applies:
\begin{itemize}
\item $Y$ contains an occurrence of $se(Q)$, and hence of \tr\ and \fa, which is a 
contradiction.
\item $se(P)\sub{\tr}{\Box}$ contains an occurrence of $Z$, and hence of \tr\ and \fa,
which is a contradiction.
\end{itemize}
By definition of a tsd if suffices to only consider the case that 
$Z$ is a proper subtree of $se(Q)$. If this is the case, then
$\SE(Q) = U\sub{\Box}{Z}$ for some $U \in \Tone$ that is not equal
to $\Box$ and does not contain $\tr$ or $\fa$ (because then $Y$ would
too). But this violates Lemma~\ref{lem:snondectf}, which states that no such
decomposition exists. 
Hence, $(\SE(P)\sub{\tr}{\Box}, \SE(Q))$ is the \emph{unique} tsd of 
$\SE(P \sleftand Q)$.
\end{proof}

\subsection{Defining an inverse}
\label{subsec:cpl}
The two decomposition theorems from the previous section
enable us to prove the the intermediate result that we used in our completeness proof
for $\FSCLT$ (Theorem~\ref{thm:daan}). We define three auxiliary functions to aid in our definition of the
inverse of $\SE$ on $\SNF$. Let 
\[\cd : \T \to \Tone \times \T\]
be the function
that returns the conjunction decomposition of its argument, $\dd$ of the same
type its disjunction decomposition and $\tsd$, also of the same type, its
$\tr$-$*$-decomposition. Naturally, these functions are undefined when their
argument does not have a decomposition of the specified type. Each of these
functions returns a pair and we will use $\cd_1$ ($\dd_1$, $\tsd_1$) to denote
the first element of this pair and $\cd_2$ ($\dd_2$, $\tsd_2$) to denote the
second element.

We define $\invs: \T \to \ST$ using the functions $\invs^\tr: \T \to \ST$ for
inverting trees in the image of $\tr$-terms and $\invs^\fa$, $\invs^\ell$
and $\invs^*$ of the same type for inverting trees in the image of
$\fa$-terms, $\ell$-terms and $*$-terms, respectively. These functions are
defined as follows.
\begin{align}
\label{eq:g1}
\invs^\tr(X) &=
  \begin{cases}
    \tr
      &\textrm{if $X = \tr$,} \\
    (a \sleftand \invs^\tr(Y)) \sleftor \invs^\tr(Z) 
      &\textrm{if $X = Y \tlef a \trig Z$.}
  \end{cases} \\ \displaybreak[0]
\label{eq:g2}
\invs^\fa(X) &=
  \begin{cases}
    \fa
      &\textrm{if $X = \fa$,} \\
    (a \sleftor \invs^\fa(Z)) \sleftand \invs^\fa(Y)
      &\textrm{if $X = Y \tlef a \trig Z$.}
  \end{cases} \\ \displaybreak[0]
\label{eq:g3}
\invs^\ell(X) &=
  \begin{cases}
    (a \sleftand \invs^\tr(Y)) \sleftor \invs^\fa(Z) 
      &\textrm{if $X = Y \tlef a \trig Z$ for some $a \in A$} \\
      &\textrm{and $Y$ only has $\tr$-leaves,} \\
    (\neg a \sleftand \invs^\tr(Z)) \sleftor \invs^\fa(Y)
      &\textrm{if $X = Y \tlef a \trig Z$ for some $a \in A$} \\
      &\textrm{and $Z$ only has $\tr$-leaves.}
  \end{cases} \\ \displaybreak[0]
\label{eq:g4}
\invs^*(X) &=
  \begin{cases}
    \invs^*(\cd_1(X)\sub{\Box}{\tr}) \sleftand \invs^*(\cd_2(X))
      &\textrm{if $X$ has a cd,} \\
    \invs^*(\dd_1(X)\sub{\Box}{\fa}) \sleftor \invs^*(\dd_2(X))
      &\textrm{if $X$ has a dd,} \\
    \invs^\ell(X)
      &\textrm{otherwise.}
  \end{cases} \\ \displaybreak[0]
\label{eq:g5}
\invs(X) &=
  \begin{cases}
    \invs^\tr(X)
      &\textrm{if $X$ has only $\tr$-leaves,} \\
    \invs^\fa(X)
      &\textrm{if $X$ has only $\fa$-leaves,} \\
    \invs^\tr(\tsd_1(X)\sub{\Box}{\tr}) \sleftand \invs^*(\tsd_2(X))
      &\textrm{otherwise.}
  \end{cases}
\end{align}

We use the symbol $\equiv$ to denote `syntactic equivalence' and we have
the following result on our normal forms.
\begin{theorem}
\label{thm:sclinv}
For all $P \in \SNF$, $\invs(\SE(P)) \equiv P$.
\end{theorem}

The proof of this theorem is provided in Appendix~\ref{app:sclinvs}. 
Theorem~\ref{thm:sclinv}
immediately implies the intermediate result that we used
in our proof of Theorem~\ref{thm:daan}, which is the next theorem.

\begin{theorem}
\label{thm:sclcpl}
For all $P, Q \in \ST$, if $se(P) = se(Q)$ then $\EqFSCL \vdash P = Q$.
\end{theorem}

\begin{proof}
Suppose that $P$ and $Q$ are two $\ST$-terms with $\SE(P) = \SE(Q)$. 
By Theorem~\ref{thm:nfs}, $P$ is derivably equal to an $\SNF$-term
$P'$, i.e., $\EqFSCL \vdash P = P'$, and  $Q$ is derivably equal to an
$\SNF$-term $Q'$, i.e., $\EqFSCL \vdash Q = Q'$. 
By Proposition~\ref{prop:sound}, $\SE(P) = \SE(P')$ and $\SE(Q) = \SE(Q')$, 
so $\invs(se(P'))\equiv\invs(se(Q'))$. By 
Theorem~\ref{thm:sclinv} it follows that 
$P'\equiv Q'$ and hence $\EqFSCL \vdash P' = Q'$, and thus $\EqFSCL \vdash P = Q$.
\end{proof}

\section{Other short-circuit logics}
\label{sec:Other}
In this section we consider some other variants of 
short-circuit logic. 
In Section~\ref{subsec:MSCL} we define a short-circuit
logic that is based on $\CP_\mem$ (see Section~\ref{subsec:CPmem}), 
and in Section~\ref{subsec:complete} we prove a completeness result.
In Section~\ref{subsec:SSCL} we discuss a short-circuit logic
that constitutes a 
variant of propositional logic. 
In Section~\ref{subsec:rpcon} we consider
two other short-circuit logics that
stem from proposition algebra axiomatizations 
in between \CP\ and $\CP_\mem$ and both
make certain identifications on successive 
occurrences of atoms. In Section~\ref{subsec:varia} we discuss
full left-sequential connectives and a setting in which different 
short-circuit logics can be useful.

\subsection{Memorizing short-circuit logic: \MSCL}
\label{subsec:MSCL}
We define a short-circuit
logic that is based on $\CP_\mem$.
This logic identifies 
much more than \FreeSCL, but the connective $\leftand$ is not commutative.

\begin{definition}
\label{def:MSCL}
\textbf{$\MSCL$ (memorizing short-circuit logic)}
is the short-circuit logic that implies no other 
consequences than those of the module expression
\begin{align*}
\{\tr,\neg,\leftand\}\export(&\CP_\mem\\
& + \langle\, \neg x=\fa\lef x\rig\tr\,\rangle\\
& + \langle\, x\leftand y=y\lef x\rig\fa\,\rangle).
\end{align*}
\end{definition}

According to Definition~\ref{def:SCL}, \MSCL\ is a 
short-circuit logic because $\CP_\mem$ is an 
axiomatic extension
of \CP\ ($\CP_\mem=\CP+\eqref{CPmem}$, 
see Section~\ref{subsec:CPmem}).
In Section~\ref{subsec:MSCL} we provide axioms for \MSCL\
and in Section~\ref{subsec:complete} we prove their completeness.


\begin{table}
\hrule
\begin{align}
\fa&=\neg\tr
\tag{\ref{SCL1}}
\\
x\leftor y&=\neg(\neg x\leftand\neg y)
\tag{\ref{SCL2}}
\\
\neg\neg x&=x
\tag{\ref{SCL3}}
\\
\tr\leftand x&=x
\tag{\ref{SCL4}}
\\ 
\nonumber
x\leftand\tr&=x
\tag{\ref{SCL5}}
\\ 
\fa\leftand x&=\fa
\tag{\ref{SCL6}}
\\
(x\leftand y)\leftand z&=x\leftand (y\leftand z)
\tag{\ref{SCL7}}
\\
x\leftand \fa
&= \neg x \leftand \fa
\tag{\ref{SCL8}}
\\[0mm]
\label{MSCL1}
\tag{MSCL1}
x\leftand(x\leftor y)&=x\\
\label{MSCL2}
\tag{MSCL2}
x\leftand(y\leftor z)&=(x\leftand y)\leftor(x\leftand z)\\
\label{MSCL3}
\tag{MSCL3}
(x\leftor y)\leftand(\neg x\leftor z)&=(\neg x\leftor z)
\leftand(x\leftor y)\\
\label{eq:rightdis}
\tag{MSCL4}
\qquad
((x\leftand y)\leftor(\neg x\leftand z))\leftand u
&=(x\leftor (z\leftand u))\leftand
(\neg x\leftor(y\leftand u))
\end{align}
\hrule
\caption{$\MSCLe$, a set of axioms for $\MSCL$.}
\label{tab:ML}
\end{table}

In Table~\ref{tab:ML} we present
a set of axioms for \MSCL\ and we call this set \MSCLe.
Axioms $\eqref{SCL1}-\eqref{SCL8}$ stem from
in $\SCLe$ (see Table~\ref{tab:SCL}).
None of the axioms $\eqref{MSCL1}-\eqref{eq:rightdis}$ is a consequence
of \SCLe, it is not hard to find closed instances that yield different $se$-images. 
Some comments on
the axioms $\eqref{MSCL1}-\eqref{eq:rightdis}$, 
where we use the notation
$(n)'$ for the dual version of axiom~$(n)$:
\begin{itemize}
\item
Axiom~\eqref{MSCL1} defines a sequential
form of absorption that implies the idempotence of 
$\leftand$ (with~$\eqref{SCL5}'$ and $y=\fa$)
and $\leftor$.

\item
Axiom~\eqref{MSCL2}
defines the left-distributivity of $\leftand$, and 
that of $\leftor$ follows by duality.
\item
Axiom~\eqref{MSCL3} and its dual define 
a restricted form of
commutativity of ${\leftand}$ and $\leftor$, reminiscent
of the identity $y\lef x\rig z=z\lef \neg x\rig y$.
We will sometimes use this identity with 
$y$ and/or $z$ equal to $\fa$,
as in 
\[x\leftand(\neg x\leftor z)=(\neg x\leftor z)\leftand x.\]

\item
Axiom~\eqref{eq:rightdis} is a combination of two more
comprehensible equations: first, with $u=\tr$ it yields
\begin{equation}
\label{eq:switch}
(x\leftand y)\leftor(\neg x\leftand z)=
(x\leftor z)\leftand
(\neg x\leftor y),
\end{equation}
which introduces another defining equation for $y\lef x\rig z$
(cf.\ the identity $y\lef x\rig z=
(x\leftand y)\leftor(\neg x\leftand z)$ discussed 
in Section~\ref{subsec:CPmem}).
Application of~\eqref{eq:switch} 
to the right-hand side of
equation~\eqref{eq:rightdis} reveals a restricted form
of right-distributivity of $\leftand$:
\begin{equation}
\label{eq:AA}
((x\leftand y)\leftor(\neg x\leftand z))\leftand u
=(x\leftand (y\leftand u))\leftor
(\neg x\leftand (z\leftand u)),
\end{equation}
and with $y=x$ and $z=\neg x$ this yields
\begin{equation}
\label{eq:BB}
(x\leftor \neg x)\leftand u=
(x\leftand u)\leftor(\neg x\leftand u).
\end{equation}
Right-distributivity is restricted in~\eqref{eq:AA}
in the sense that the ``guards'' $x\leftand ..$
and $\neg x\leftand ..$ must be present.

In fact, axiom~\eqref{eq:rightdis} can in \MSCLe\ be replaced by
equations $\eqref{eq:switch}$ and $\eqref{eq:AA}$
if one prefers elegance to conciseness. 
\end{itemize}

\begin{proposition}[Soundness]
\label{prop:sound2}
For all \SCL-terms $t$ and $t'$,
\[\MSCLe\vdash t=t'\quad\Longrightarrow\quad \MSCL\vdash t=t'.\]
\end{proposition}

\begin{proof}
We use that $\leftor$ can be defined
in exactly the same way in \MSCL\ as in \MSCLe\ (cf.\
the proof
of Proposition~\ref{prop:sound}). 
With respect to the soundness of \MSCLe, axiom
\eqref{eq:rightdis}, i.e.,
\[((x\leftand y)\leftor(\neg x\leftand z))\leftand u
=(x\leftor (z\leftand u))\leftand(\neg x\leftor(y\leftand u)),\]
is the only non-trivial case. Write $L=R$ 
for axiom~\eqref{eq:rightdis}, then
\begin{align*}
L&=u\lef(\tr\lef(y\lef x\rig\fa)\rig (\fa\lef x\rig z))\rig \fa
\\
&=u\lef(y\lef x\rig\fa)\rig(u\lef(\fa\lef x\rig z)\rig\fa)
&&\text{by \eqref{cp4} and \eqref{cp1}}\\
&=u\lef(y\lef x\rig\fa)\rig(\fa\lef x\rig (u\lef z\rig\fa))
&&\text{by \eqref{cp4} and \eqref{cp2}}\\
&=[u\lef y\rig(\fa\lef x\rig(u\lef z\rig\fa))]\lef x\rig~
\\
&\phantom{=~}~~
[\fa\lef x\rig(u\lef z\rig\fa)]
&&\text{by \eqref{cp4}}\\
&=(u\lef y\rig\fa)\lef x\rig(u\lef z\rig\fa)
&&\text{by \eqref{CPmem'} and \eqref{eq:contr}}\\
&=(u\lef y\rig \fa)\lef x\rig(\tr\lef (u\lef z\rig\fa)\rig\fa)
&&\text{by \eqref{cp3}}\\
&=[(u\lef y\rig \fa)\lef x\rig\tr]\lef x\rig~\\
&\phantom{=~}~~
[((u\lef y\rig\fa)\lef x\rig\tr)\lef(u\lef z\rig\fa)\rig\fa]
&&\text{by \eqref{eq:contr2} and \eqref{CPmem''}}\\
&=[(u\lef y\rig \fa)\lef x\rig\tr]\lef(\tr\lef x\rig(u\lef z\rig\fa))\rig\fa
&&\text{by \eqref{cp4} and \eqref{cp1}}\\
&=R.
\end{align*}
\end{proof}

As stated above, $\MSCL$ does not imply commutativity 
of $x\leftand y$,
but otherwise has some familiar consequences such as
the absorption law~\eqref{MSCL1}
and its dual (and these are the \emph{only} variants of absorption
that are valid in $\MSCL$). 
Some typical consequences of \MSCLe\ that express the flavor 
of \MSCL\ are these:
\begin{align}
\label{eq:4}
x\leftand \neg x&=x\leftand \fa,\\
\label{eq:5}
x\leftand y&=x\leftand (\neg x\leftor y),\\
\label{eq:7}
(x\leftand y)\leftand x&=x\leftand y,\\
\label{eq:ss}
(x\leftand y)\leftand \neg x&=(x\leftand y)\leftand\fa
\end{align}
(derivations are given below).
Equations~$\eqref{eq:4} - \eqref{eq:ss}$ will be used in 
the next sections and each of these expresses a typical property of $\MSCL$:
the first evaluation result of $x$ is memorized. 

\begin{itemize}
\item
Equation~\eqref{eq:4} can be derived as follows:
\begin{align*}
x\leftand \neg x&=
(x\leftor\fa)\leftand (\neg x\leftor\fa)\\
&=
(x\leftand\fa)\leftor (\neg x\leftand\fa)&& \text{by~\eqref{eq:switch}}\\
&=
(x\leftand\fa)\leftor (x\leftand\fa)&& \text{by~\eqref{SCL8}}\\
&=
x\leftand \fa,
\end{align*}
and hence, $\neg x\leftand x=\neg x\leftand \neg\neg x
=\neg x\leftand\fa=x\leftand \neg x$.
Note that the dual of~\eqref{eq:4},  thus
\[x\leftor \neg x=x\leftor\tr,\]
can be seen as a weak version of the law
of the excluded middle.

\item
Equation~\eqref{eq:5}
can be derived as follows:
\begin{align*}
x\leftand y&=x\leftand(\fa\leftor y)&&\text{by~\eqref{SCL4}}'\\
&=(x\leftand\fa)\leftor(x\leftand y)&&\text{by~\eqref{MSCL2}}\\
&=(x\leftand\neg x)\leftor (x\leftand y)
&&\text{by~\eqref{eq:4}}\\
&=x\leftand (\neg x\leftor y).&&\text{by~\eqref{MSCL2}}
\end{align*}
Two immediate consequences of this identity are
\begin{equation}
\label{eq:10}
x=x\leftand (\neg x\leftor \tr)
\quad\text{and its dual}\quad
x=x\leftor(\neg x\leftand\fa).
\end{equation}

\item
Equation~\eqref{eq:7}, i.e.,
$(x\leftand y)\leftand x=x\leftand y$,
is an immediate consequence of the more general equation 
$(x\leftand y)\leftand (x\leftor u)=x\leftand y$ (take $u=\fa$),
which can be derived as follows:
\begin{align*}
(x\leftand y)\leftand (x\leftor u)
&=(x\leftand (\neg x\leftor y))\leftand (x\leftor u)
&&\text{by~\eqref{eq:5}}\\
&=x\leftand ((\neg x\leftor y)\leftand (x\leftor u))
&&\text{by~\eqref{SCL7}}\\
&=x\leftand ((x\leftor u)\leftand (\neg x\leftor y))
&&\text{by~\eqref{MSCL3}}\\
&=(x\leftand (x\leftor u))\leftand (\neg x\leftor y)
&&\text{by~\eqref{SCL7}}\\
&=x\leftand (\neg x\leftor y)
&&\text{by~\eqref{MSCL1}}\\
&=x\leftand y.
&&\text{by~\eqref{eq:5}}
\end{align*}

\item Equation~\eqref{eq:ss} can be proven as follows:
\begin{align*}
(x\leftand y)\leftand \neg x
&=((x\leftand y)\leftand x)\leftand \neg x
&&\text{by~\eqref{eq:7}}\\
&=(x\leftand y)\leftand (x\leftand \neg x)
&&\text{by~\eqref{SCL7}}\\
&=(x\leftand y)\leftand(x\leftand\fa)
&&\text{by~\eqref{eq:4}}\\
&=((x\leftand y)\leftand x)\leftand\fa
&&\text{by~\eqref{SCL7}}\\
&=(x\leftand y)\leftand\fa.
&&\text{by~\eqref{eq:7}}
\end{align*}
\end{itemize}

So, \MSCL\ embodies the typical property that upon
the evaluation of a closed \SCL-term
$P$, once an atom has been evaluated, all
subsequent evaluations of that same atom will yield the same truth value. 
More precisely, if $(\sigma,B)$ is an evaluation of $P$
(see Definitions~\ref{def:eval}
and~\ref{def:ext}) and $aB_i\in\sigma$, then for all occurrences of $aB_j$ in
$\sigma$, $B_j=B_i$. 

\begin{example}
\label{ex:mem}
An example of a `memorizing' atom in a test expression in a conditional
program fragment is a call to a
\emph{memoizing function} with a fixed argument: a memoizing function is a 
function which maintains
a cache of function values for arguments it has previously been called with
(see~\cite{ABH} for a recent and detailed account of memoization and 
\url{http://en.wikipedia.org/wiki/Memoization} for some general information).
Another example that might support \MSCL\ evolves when
we consider programs (say, in a Perl-like language)
that allow in conditions only comparative tests on scalar 
variables, and also special tests on whether a 
program-variable has been
evaluated before, 
say \verb+eval$x+ for scalar variable \verb+x+.
This combines well with the consequences of \MSCL, but
refutes identities that are typically not in \MSCL,
such as for example
\[(\verb+$x==1+)\leftand (\verb+eval$x+)=(\verb+eval$x+)\leftand (\verb+$x==1+)\]
where the left-hand side always evaluates to \emph{true},
while the right-hand side can yield \emph{false}.
\end{example}

\subsection{An axiomatization of \MSCL}
\label{subsec:complete}
In this section we prove that $\MSCLe$ is a complete
axiomatization of \MSCL. We start with an intermediate result.

\begin{lemma}
\label{lem:AA}
$\MSCLe\vdash
x\leftand(y\leftor ((x\leftand z)\leftor(\neg x\leftand u)))
=x\leftand(y\leftor z)$.
\end{lemma}

\begin{proof}
We derive
\begin{align*}
x&\leftand(y\leftor((x\leftand z)\leftor(\neg x\leftand u)))\\
&=(x\leftand y)\leftor (x\leftand[(x\leftand z)\leftor (x\leftand
(\neg x\leftand u))])
&&\text{by \eqref{MSCL2}}\\
&=(x\leftand y)\leftor ((x\leftand (x\leftand z))\leftor (x\leftand
(\neg x\leftand u)))
&&\text{by \eqref{MSCL2}}\\
&=(x\leftand y)\leftor (((x\leftand x)\leftand z)\leftor ((x\leftand
\neg x)\leftand u))
&&\text{by \eqref{SCL7}}\\
&=(x\leftand y)\leftor ((x\leftand z)\leftor (x\leftand\fa))
\\
&=(x\leftand y)\leftor (x\leftand (z\leftor \fa))
&&\text{by \eqref{MSCL2}}\\
&=(x\leftand y)\leftor (x\leftand z)
\\
&=x\leftand (y\leftor z).
&& \text{by~\eqref{MSCL2}}
\end{align*}
\end{proof}

\begin{theorem}[Completeness]
\label{thm:MSCL}
For all \SCL-terms $t$ and $t'$,
\[\MSCL\vdash t=t'\quad\Longrightarrow\quad\MSCLe\vdash t=t'.\]
\end{theorem}

\begin{proof}
It suffices to prove that the axioms of $\CP_\mem$ 
are derivable from \MSCLe.
In this proof we use $\fa$ and $\leftor$ in the familiar
way and also the fact that with the axioms of $\CP_\mem$, the conditional 
$x\lef y\rig z$ can be expressed by
$x\lef y\rig z=(y\leftand x)\leftor(\neg y\leftand z)$ 
(see Section~\ref{subsec:CPmem}). Furthermore, we use
equation~\eqref{eq:switch}, i.e., 
\[
(y\leftand x)\leftor(\neg y\leftand z)
=(y\leftor z)\leftand (\neg y\leftor x).\]
\eqref{cp1}: $x\lef\tr\rig y=(\tr\leftand x)\leftor(\fa\leftand y)=
x\leftor \fa=x$.
\\[2mm]
\eqref{cp2}: $x\lef\fa\rig y=(\fa\leftand x)\leftor (\tr\leftand y)=
\fa\leftor y=y$.
\\[2mm]
\eqref{cp3}: $\tr\lef x\rig\fa= (x\leftand \tr)\leftor(\neg x\leftand\fa)=
x\leftor (\neg x\leftand\fa)=x$ by~\eqref{eq:10}.
\\[2mm]
\eqref{cp4}: To derive \eqref{cp4}, i.e., ~
$x\lef (y\lef z\rig u)\rig v=
(x\lef y\rig v)\lef z\rig(x\lef u\rig v)$,~
to which we further refer to by $L=R$, 
we use the identity
\begin{equation}
\label{eq:version}
(x\leftor(y\leftand z))\leftand(\neg x\leftor(u\leftand z))
=((x\leftor y)\leftand(\neg x\leftor u))\leftand z,
\end{equation}
which can be easily derived from equations~\eqref{eq:switch}
and \eqref{eq:rightdis}.
Then 
\[L=(X\leftand x)\leftor(\neg X\leftand v),\]
with $X=(z\leftand y)\leftor(\neg z\leftand u)$.
Hence, $X=(z\leftor u)\leftand(\neg z\leftor y)$
is derivable from \MSCLe, and so is
$\neg X=(z\leftor\neg u)\leftand(\neg z\leftor \neg y)$.
We derive
\begin{align*}
L&=
([z\leftor(u\leftand x)]\leftand
[\neg z\leftor (y\leftand x)])~\leftor\\
&
\phantom{L=\;}
([z\leftor(\neg u\leftand v)]\leftand
[\neg z\leftor (\neg y\leftand v)])
&&\text{by \eqref{eq:version}}\\
&=
(z\leftand (y\leftand x))\leftor
(\neg z\leftand (u\leftand x))
~\leftor
&&\text{by \eqref{eq:switch} }\\
&
\phantom{=~}
(z\leftand (\neg y\leftand v))\leftor
(\neg z\leftand (\neg u\leftand v))
&&\text{and \smash{$\eqref{SCL7}'$}}\\
&=(z\leftand (y\leftand x))\leftor(z\leftand 
(\neg y\leftand v))
~\leftor\\
&
\phantom{=~}
(\neg z\leftand (u\leftand x))
\leftor
(\neg z\leftand (\neg u\leftand v))
&&\text{by \eqref{MSCL3}}\\
&=
(z\leftand~((y\leftand x)\leftor(\neg y\leftand v)))
~\leftor
\\
&
\phantom{=~}
(\neg z\leftand~((u\leftand x)\leftor
(\neg u\leftand v)))
&&\text{by \eqref{MSCL2}}\\
&=R.
\end{align*}
\eqref{CPmem}: As argued in Section~\ref{subsec:CPmem}
it is sufficient to
derive axiom~\eqref{CPmem'}, that is,
\[(w\lef y\rig(z\lef x\rig u))\lef 
x\rig v=
(w\lef y\rig z)\lef x\rig v,\]
say $L=R$.
We derive
\begin{align*}
L&=(x\leftand(w\lef y\rig (z\lef x\rig u)))\leftor(\neg x\leftand v)\\
&=(x\leftand[(y\leftor [z\lef x\rig u])\leftand
(\neg y\leftor w)])\leftor(\neg x\leftand v)
\\
&=(x\leftand[(y\leftor [(x\leftand z)\leftor(\neg x\leftand u)])
\leftand(\neg y\leftor w)])\leftor(\neg x\leftand v)\\
&=([x\leftand(y\leftor [(x\leftand z)\leftor(\neg x\leftand u)])]
\leftand(\neg y\leftor w))\leftor(\neg x\leftand v)
&&\text{by~\eqref{SCL7}}
\\
&=([x\leftand(y\leftor z)]\leftand(\neg y\leftor w))
\leftor(\neg x\leftand v)&& \text{by~Lemma~\ref{lem:AA}}\\
&=(x\leftand[(y\leftor z)\leftand(\neg y\leftor w)])
\leftor(\neg x\leftand v)&& \text{by~\eqref{SCL7}}\\
&=R.
\end{align*}
\end{proof}

We end this section with a proof of the
derivability of the \SCLe-axioms \eqref{SCL9}~and~\eqref{SCL10} 
from \MSCLe. 
Of course, derivability of all
closed instances of these axioms follows
from Theorem~\ref{thm:MSCL} and the fact that 
\MSCL\ identifies more than \FreeSCL.

\smallskip

\noindent
$\MSCLe\vdash\eqref{SCL9}$.
This can be derived as
follows:
\begin{align*}
(x\leftand\fa)\leftor y
&=(x\leftand\neg x)\leftor y
&& \text{by~\eqref{eq:4}}\\
&=(x\leftor y)\leftand (\neg x\leftor y)
&& \text{by~\eqref{eq:BB}}'\\
&=(x\leftor \neg x)\leftand y
&& \text{by~\eqref{eq:AA}}\\
&=(x\leftor\tr)\leftand y.
&& \text{by~\eqref{eq:4}}'
\end{align*}

\noindent
$\MSCLe\vdash\eqref{SCL10}$.
First, we prove this auxiliary result:
\begin{align}
\nonumber
(x\leftand y)\leftor z
&=(x\leftand (\neg x\leftor y))\leftor z
&&\text{by~\eqref{eq:5}}\\
\nonumber
&=((x\leftor \fa)\leftand (\neg x\leftor y))\leftor z\\
\nonumber
&=
(x\leftand(y\leftor z))\leftor (\neg x\leftand(\fa\leftor z))
&& \text{by~\eqref{eq:rightdis}}'\\
\label{eq:8}
&=(x\leftand(y\leftor z))\leftor (\neg x\leftand z).
\end{align}
With~\eqref{eq:8} and the identity 
$\underline{(z\leftand\fa)}=
\underline{(z\leftand\fa)\leftand u}$, 
the axiom~\eqref{SCL10} can be easily derived:
\begin{align*}
(x\leftand y)&\leftor(z\leftand\fa)\\
&=
(x\leftand(y\leftor(z\leftand\fa)))\leftor(\neg x\leftand(z\leftand\fa))
&& \text{by~\eqref{eq:8}}\\
&=(x\leftor\underline{(z\leftand\fa)})\leftand
(\neg x\leftor (y\leftor(z\leftand\fa)))
&& \text{by~\eqref{eq:switch}}'\\
&=(x\leftor[\underline{(z\leftand\fa)\leftand(y\leftor(z\leftand\fa))}])~
\leftand
\\
&\phantom{=~}
(\neg x\leftor (y\leftor(z\leftand\fa)))\\
&=(x\leftor (z\leftand\fa))\leftand(y\leftor (z\leftand\fa)).
&& \text{by~\eqref{eq:8}}'
\end{align*}

\subsection{Static short-circuit logic: \SSCL}
\label{subsec:SSCL}
In this section we prove that the equation
$x\leftand \fa=\fa$ marks a distinguishing feature between
\MSCL\ and propositional logic: adding this equation to \MSCLe\ 
yields an equational characterization of propositional logic (be
it in sequential notation and defined with short-circuit 
evaluation).

\begin{definition}
\label{def:SSCL}
\textbf{$\SSCL$ (static
 short-circuit logic)}
is the short-circuit logic that implies no other 
consequences than those of the module expression 
\begin{align*}
\{\tr,\neg,\leftand\}\export(&\CP_\mem \\
& + \langle\, \fa\lef x\rig\fa=\fa\,\rangle\\
& + \langle\, \neg x=\fa\lef x\rig\tr\,\rangle\\
& + \langle\, x\leftand y=y\lef x\rig\fa\,\rangle).
\end{align*}
\end{definition}

\begin{definition}
The set \SSCLe\ is defined as the extension of \MSCLe\
(see~Table~\ref{tab:ML})
with the axiom 
\[x\leftand \fa=\fa.\]
\end{definition}
Our first result
is a very simple corollary
of Theorem~\ref{thm:MSCL}.
\begin{theorem}[Soundness and completeness]
\label{thm:SSCL}
For all \SCL-terms $t$ and $t'$,
\[\SSCLe\vdash t=t'\iff \SSCL\vdash t=t'.\]
\end{theorem}
\begin{proof}
Soundness, i.e., $\Longrightarrow$, follows
trivially from Proposition~\ref{prop:sound2} and the
fact that $\SSCL\vdash x\leftand\fa=\fa\lef x\rig\fa=\fa$. 

In order to show $\Longleftarrow$ it suffices to prove
that the axioms of $\CP_\mem$ and $\fa\lef x\rig \fa=\fa$
are derivable from \SSCLe.
By Theorem~\ref{thm:MSCL} and using
the expressibility of $x\lef y\rig z$ as we did in the proof of
Theorem~\ref{thm:MSCL}, it suffices to show that 
$\SSCLe\vdash(x\leftand\fa)\leftor(\neg x\leftand\fa)=\fa$, which is trivial:
\begin{align*}
(x\leftand\fa)\leftor(\neg x\leftand\fa)
&=\fa\leftor\fa
&&\text{by~\eqref{SCL8} and the axiom $x\leftand \fa=\fa$}\\
&=\fa.
&&\text{by $\eqref{SCL4}'$}
\end{align*}
\end{proof}

Combining identity~\eqref{eq:4} (that is, $x\leftand \neg x
=x\leftand\fa$) and $x\leftand\fa=\fa$ yields
\[x\leftand \neg x=\fa\quad
\text{and thus also}\quad x\leftor \neg x=\tr.
\]

\begin{lemma}
\label{lem:comm}
$\SSCLe\vdash x\leftand y=y\leftand x$.
\end{lemma}

\begin{proof}
We derive
\begin{align*}
x\leftand y&=(y\leftor\neg y)\leftand(x\leftand y)\\
&=(y\leftand(x\leftand y))\leftor
(\neg y\leftand(x\leftand y))
&& \text{by~\eqref{eq:BB}}\\
&=(y\leftand x)\leftor((\neg y\leftand x)\leftand\fa)
&& \text{by~\eqref{eq:7} and~\eqref{eq:ss}}\\
&=y\leftand x.
&& \text{by~$x\leftand\fa=\fa$}
\end{align*}
\end{proof}

As a consequence, all equations defining absorption (among which~\eqref{MSCL1})
and distributivity (among which~\eqref{MSCL2}) 
follow from \SSCLe, and it is not difficult
to see that \SSCLe\ defines the mentioned
variant of ``sequential propositional logic'': this follows for example immediately
from~\cite{Sioson} in which equational bases for Boolean 
algebra are provided, and each of these bases
can be easily derived from \SSCLe\ (below we return to this point).

The attentive reader may wonder why we did not define
\SSCL\ using the axiom set $\CP_\stat$ defined in~\cite{BP10} 
(and briefly discussed in Section~\ref{subsec:CPstat}). 
Recall that $\CP_\stat$ is defined as the
extension of \CP\ with the axiom \eqref{CPstat}, that is,
\[(x\lef y\rig z)\lef u\rig v
=(x\lef u\rig v)\lef y\rig (z\lef u\rig v)\]
and the contraction law~\eqref{eq:contr2}, that is, 
\[
(x\lef y\rig z)\lef y\rig u=x\lef y\rig u.
\]
In Section~\ref{subsec:pa} (Proposition~\ref{prop:3}) 
we show that $\CP_\stat$ and 
$(\CP_\mem+\langle\, \fa\lef x\rig\fa=\fa\,\rangle)$ 
are equally strong. Hence we have the following corollary.
\begin{corollary}
\SSCL\ equals 
\begin{align*}
\{\tr,\neg,\leftand\}\export(&\CP_\stat \\
 &+ \langle\, \neg x=\fa\lef x\rig\tr\,\rangle\\
 &+ \langle\, x\leftand y=y\lef x\rig\fa\,\rangle).
\end{align*}
\end{corollary}

Hoare proved in~\cite{Hoa85} that 
each tautology in propositional logic can be (expressed and)
proved with his axioms for the conditional. 
According to~\cite{BP10},
this also holds for $\CP_\stat$, and
thus also for $\SSCLe$
if we identify the symmetric connectives with their
left-sequential counterparts.

\subsection{Contractive and Repetition-Proof short-circuit logic}
\label{subsec:rpcon}
We briefly discuss two other variants of short-circuit logics 
which both involve explicit reference to the set $A$ of atoms
(propositional variables).
Both these variants are located in between \FreeSCL\ and \MSCL.

In~\cite{BP10} we introduced $\CP_\con$ (contractive
\CP)  which is defined by
the extension of \CP\ with these axiom schemes 
($a\in A$):
\begin{align*}
\label{CPcr1}
\tag{CPcr1}
\qquad
(x\lef a\rig y)\lef a\rig z&=x\lef a\rig z,\\
\label{CPcr2}
\tag{CPcr2}
\qquad
x\lef a\rig (y\lef a\rig z)&=x\lef a\rig z.
\end{align*}
These schemes contract for each atom $a$
respectively the \emph{true}-case and the 
\emph{false}-case.
We write $\CP_\con(A)$ to make explicit that 
these axiom schemes refer to the set $A$ of atoms.

\begin{definition}
\label{def:CSCL}
\textbf{$\CSCL$ (contractive
 short-circuit logic)}
is the short-circuit logic that implies no other 
consequences than those of the module expression 
\begin{align*}
\{\tr,\neg,\leftand,a\mid a\in A\}\export(&\CP_\con(A)\\
& + \langle\, \neg x=\fa\lef x\rig\tr\,\rangle\\
& + \langle\, x\leftand y=y\lef x\rig\fa\,\rangle).
\end{align*}
\end{definition}

\begin{table}
\hrule
\begin{align}
\label{con1}
a\leftand(a\leftor x)&=a
\\
\label{con2}
a\leftor(a\leftand x)&=a
\\
\label{con3}
a\leftor\neg a&=a\leftor \tr\\
\label{con4}
a\leftand \neg a&= a\leftand\fa
\end{align}
\hrule
\caption{Axiom schemes ($a\in A$) 
for \CSCL}
\label{tab:con}
\end{table}

The identities defined by \CSCL\ 
include those 
derivable from \SCLe\ 
(see Table~\ref{tab:SCL})
and the axiom schemes in Table~\ref{tab:con}.
The axiom schemes \eqref{con1} and \eqref{con2}
are the counterparts of the axiom schemes \eqref{CPcr1}
and \eqref{CPcr2} (for $a\in A$).
Observe that from \SCLe\ and axiom
schemes~\eqref{con1} and \eqref{con2}
the following equations can be derived
($a\in A$): 
\begin{align*}
a\leftand a&=a,&
a\leftor a&=a,\\
\neg a\leftand(\neg a\leftor x)&=\neg a,
&
\neg a\leftor(\neg a\leftand x)&=\neg a,\\
\neg a\leftand\neg a&=\neg a,& 
\neg a\leftor \neg a &=\neg a.
\end{align*}
Furthermore, it is not hard to prove that 
the axiom schemes~\eqref{con3} and~\eqref{con4}
are also valid in \CSCL, and imply with \SCLe\ 
these consequences
($a\in A$): 
\begin{align*}
\neg a\leftand a&=a\leftand\fa,
&\neg a\leftor a&=a\leftor\tr.
\end{align*}
The question whether the extension of \SCLe\ 
with the axiom schemes in Table~\ref{tab:con}
provides for closed terms an axiomatization of \CSCL\
is left open.

\begin{example}
\label{ex:con}
An example that illustrates
the use of
\CSCL\ concerns atoms that define
manipulation of Boolean registers:
\begin{itemize}
\item 
Consider
atoms \texttt{set:$i$:$j$} and \texttt{eq:$i$:$j$}
with $i\in\{1,...,n\}$ (the number of registers)
and $j\in\{\tr,\fa\}$ (the value of registers).
\item
An atom \texttt{set:$i$:$j$} can
have a side effect
(it sets register $i$ to value $j$) and yields
upon evaluation always \emph{true}.
\item An atom
\texttt{eq:$i$:$j$} has no side effect but 
yields upon evaluation only \emph{true} if register $i$
has value $j$. 
\end{itemize}
Clearly, the consequences mentioned above
are derivable in $\CSCL$, but $x\leftand x=x$ is not:
assume register 1 has value \fa\ and let $t=
\texttt{eq:1:\fa}\leftand \texttt{set:1:\tr}$.
Then $t$ yields \emph{true} upon evaluation in this state, while
$t\leftand t$ yields \emph{false}.
\end{example}

In~\cite{BP10} we also introduced $\CP_\rp$ (repetition-proof
\CP) for the axioms 
in \CP\ extended with these axiom schemes 
($a\in A$):
\begin{align*}
\label{CPrp1}
\tag{CPrp1}
\qquad
(x\lef a\rig y)\lef a\rig z&=(x\lef a\rig x)\lef a\rig z,\\
\label{CPrp2}
\tag{CPrp2}
\qquad
x\lef a\rig (y\lef a\rig z)&=x\lef a\rig (z\lef a \rig z).
\end{align*}
It is easily seen that the axiom schemes~\eqref{CPrp1}
and~\eqref{CPrp2} are derivable in $\CP_\con$ 
(so $\CP_\con$ is also an axiomatic extension
of $\CP_\rp$). 
We write $\CP_\rp(A)$ to make explicit that
the axioms schemes
refer to the set $A$ of atoms.

\begin{definition}
\label{def:RPSCL}
\textbf{$\RPSCL$ (repetition-proof
 short-circuit logic)}
is the short-circuit logic that implies no other 
consequences than those of the module expression
\begin{align*}
\{\tr,\neg,\leftand,a\mid a\in A\}\export(&\CP_\rp(A)\\
& + \langle\, \neg x=\fa\lef x\rig\tr\,\rangle\\
& + \langle\, x\leftand y=y\lef x\rig\fa\,\rangle).
\end{align*}
\end{definition}

\begin{table}
\hrule
\begin{align}
\label{rp1}
a\leftand(a\leftor x)&=a\leftand a
&&(\text{cf.~\eqref{CPrp1})}\\[0mm]
\label{rp2}
a\leftor(a\leftand x)&=a\leftor a
&&(\text{cf.~\eqref{CPrp2})}\\[2mm]
\label{rp3}
(a\leftor\neg a)\leftand x&=(\neg a\leftand a)\leftor x
&&\text{(cf.~\eqref{SCL4}}\\
\label{rp4}
(\neg a\leftor a)\leftand x&=(a\leftand \neg a)\leftor x
&&\text{~and \eqref{SCL9})}\\[2mm]
\label{rp5}
(a\leftand\neg a)\leftand x &=a\leftand\neg a
&&(\text{cf.~\eqref{SCL6}})\\
\label{rp6}
(\neg a\leftand a)\leftand x &=\neg a\leftand a
\\[2mm]
\label{rp9}
(x\leftand y)\leftor(a\leftand\neg a)&=
(x\leftor (a\leftand\neg a))\leftand(y\leftor (a\leftand\neg a))
&&(\text{cf.~\eqref{SCL10}})\\
\label{rp10}
(x\leftand y)\leftor(\neg a\leftand a)&=
(x\leftor (\neg a\leftand a))\leftand(y\leftor (\neg a\leftand a))
\end{align}
\hrule
\caption{Axiom schemes ($a\in A$) 
for \RPSCL}
\label{tab:rp}
\end{table}

The identities defined by \RPSCL\
include those that are 
derivable from \SCLe\ (see Table~\ref{tab:SCL})
and Table~\ref{tab:rp}.
Axiom schemes~\eqref{rp1} and \eqref{rp2} are
the  counterparts of the axiom schemes 
\eqref{CPrp1}
and \eqref{CPrp2}, and axioms schemes~\eqref{rp3} 
and \eqref{rp4} are the counterparts of the identity
$\tr\leftand x=\fa\leftor x$ and also of axiom~\eqref{SCL9}.
Axiom schemes~$\eqref{rp5}-\eqref{rp10}$ are 
the counterparts of the remaining
\SCLe-axioms that involve \tr\ or \fa.
We do not know whether the extension of \SCLe\ 
with the axiom schemes in Table~\ref{tab:rp}
provides for closed terms an axiomatization
of \RPSCL. It can be easily shown that
all axiom schemes in Table~\ref{tab:rp}
follow from \SCLe\ extended with the 
axiom schemes in Table~\ref{tab:con} for \CSCL.

\begin{example}
\label{ex:rp}
An example that illustrates 
the use of \RPSCL\ is
a combination of Example~\ref{ex:Perl} (on \FreeSCL\ and Perl)
and the above example on \CSCL.
Consider simple arithmetic expressions over the natural 
numbers (or the integers) and a program notation for imperative
programs or algorithms in which each atom is either 
a test or an assignment.
Assume that assignments
when used as conditions always evaluate to \emph{true}
(next to having their intended effect).
Then, these atoms
satisfy the axioms in Tables~\ref{tab:SCL} and~\ref{tab:rp}.
However, the assignment \texttt{(n=n+1)} clearly does not
satisfy the contraction law $a\leftand a=a$ because
$(\texttt{(n=n+1)}\leftand\texttt{(n=n+1)})\leftand
\texttt{(n==2)}$ and
$\texttt{(n=n+1)}\leftand
\texttt{(n==2)}$ can yield different evaluation results.
Hence we have a clear example of the repetition-proof
characteristic of \RPSCL.
\footnote{This is related to the work of Lars Wortel~\cite{Wortel},
in which a comparable extension of \emph{Dynamic 
Logic}~\cite{DynLog1,DynLog2} is 
defined. The difference with the ``semantics'' of 
conditions in \textbf{Perl} as described in 
Section~\ref{subsec:SCLe} is that in such a simple extension of
Dynamic Logic it is natural to assume that assignments (as atoms) 
always evaluate to \emph{true} (because they always succeed).}
\end{example}

\subsection{Side effects, full evaluation and combining short-circuit logics}
\label{subsec:varia}
In this section we consider the
role of the constants \tr\ and \fa\ in short-circuit logic, and we propose a 
formal definition of an atom having a side effect 
(given some execution environment) that is 
based on these constants. Then, we briefly
discuss connectives that prescribe \emph{full left-sequential evaluation}
and a setting in which different short-circuit logics can be useful.

\bigskip

A perhaps interesting variant of the module
\SCL\ that generically defines short-circuit logics (Definition~\ref{def:SCL})
is obtained by leaving out the constant \tr\ in the exported
signature (and thus also leaving out
$\fa$ as a definable constant).
Such a variant could be motivated by the fact that these constants are
usually absent in conditions in imperative programming (although 
they may be always used; also $\tr$ can be mimicked by a void equality 
test such as \texttt{(1==1)}, or simply by \texttt{(1)} in an expression-evaluated
language such as \textbf{Perl}). 
Observe that in ``\SCL\ without \tr'' only the
\SCLe-axioms expressing
duality, double negation shift,
and associativity (axioms~\eqref{SCL2}, 
\eqref{SCL3} and~\eqref{SCL7}, respectively) 
remain. Moreover,
these axioms then yield a complete
axiomatization of this restricted form of free valuation congruence.
However, we think that ``\SCL\ without \tr'' does not yield an appropriate 
point of view: in a logic about truth and falsity (up to and including \MSCL)
one should be able to express \emph{true} itself as a value.


Although side effects are well understood in programming, see 
e.g., \cite{BW96,Nor97}, they are often explained without a general formal 
definition. 
 We first 
quote from~\cite{Daan} a discussion about side effects in which a formal
definition is proposed:

\begin{citaat}[from~\cite{Daan}]
``If an
atom $a \in A$ does not have a side effect, then it always behaves either as
the constant $\tr$ or as the constant $\fa$ depending on the atoms that
were evaluated before it and the state of the execution environment. For $a \in
A$ and $P$ a closed \SCL-term let $[\tr / a]P$ denote the term which results from
replacing each occurrence of $a$ in $P$ by $\tr$.  Similarly, let $[\fa /
a]P$ be the term that results from replacing each occurrence of $a$ in $P$ with
$\fa$. Let $y_e$ be the function that returns the evaluation result of an
$\SCL$-term when it is evaluated in execution environment $e$. An atom $a \in
A$ has a side effect if there is some execution environment $e$ and there are
closed \SCL-terms $P$ and $Q$ with $y_e(P) = y_e(Q)$ such that either
\begin{equation*}
y_e([\tr / a]P) \neq y_e([\tr / a]Q)
\quad\text{or}\quad
y_e([\fa / a]P) \neq y_e([\fa / a]Q).
\end{equation*}


As an example consider atoms $a$ and $b$ and suppose that a side effect of $a$
is that any evaluation of $b$ that follows it will yield $\true$. Also suppose
that if $b$ were not preceded by $a$ it would yield $\false$. To make this
concrete we could imagine $a$ being a method that sets some global variable in
the execution environment and always yields $\true$. We could then see $b$ as
being a method that checks whether that variable has been set, in which case it
yields $\true$, or not, in which case it yields $\false$. Letting $e$ be some
execution environment where the global variable is not set, or alternatively
the empty execution environment, we see that $y_e(a \leftand b) = \tr =
y_e(\neg b)$ and that $y_e(\tr \leftand b) = \fa \neq y_e(\neg b)$. Hence
$a$ has a side effect by our definition.
{ [...] }
With this definition of a side effect we see that $\SCLe$, 
unlike propositional logic, preserves
side effects in the sense that 
\[\SCLe \vdash P = Q\]
implies
$\SCLe \vdash
[\tr / a]P = [\tr / a]Q\text{ and }\SCLe \vdash [\fa / a]P = [\fa / a]Q$
for all closed \SCL-terms $P$ and $Q$ and all $a \in A$. Thus, if
we adopt our proposed definition of side effects,  $\SCLe$ can be
used to reason about propositional expressions with atoms that may have side
effects.''
\end{citaat}
Of course, the definition proposed above also preserves side effects
with respect to each of the other short-circuit logics discussed 
(and implies that  \SSCL\ excludes the presence
of atoms having side effects).


We continue with a brief discussion of two
other ``Boolean operators" that are common in imperative
programming, namely those that prescribe full evaluation of their arguments.
First we consider the connective  \texttt{\&}, which in~\cite{Daan}
is called \emph{full left-sequential conjunction}, with
notation
\[x\fulland y.\]
In~\cite{Daan}, an interpretation function to evaluation trees
is defined that can be
seen as an extension of $se$:
\[se(P\fulland Q)=se(P)[\tr\mapsto se(Q),\fa\mapsto se(Q)[\tr\mapsto\fa]].\]
For example, the evaluation tree of $a\fulland b$ is the following perfect binary tree:
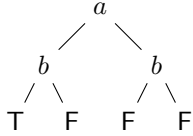
\[
\begin{tikzpicture}[%
      level distance=7.5mm,
      level 1/.style={sibling distance=15mm},
      level 2/.style={sibling distance=7.5mm},
      baseline=(current bounding box.center)]
      \node (a) {$a$}
        child {node (b1) {$b$}
          child {node (d1) {$\tr$}} 
          child {node (d2) {$\fa$}}
        }
        child {node (b2) {$b$}
          child {node (d5) {$\fa$}} 
          child {node (d6) {$\fa$}}
        };
      \end{tikzpicture}
\]
In a mixed setting with negation, $\leftand$ and $\leftor$, 
and the constants $\tr$
and $\fa$, the connective $\fulland$ and its dual $\fullor$
are definable:
\begin{align*}
x\fulland y&=(x\leftor(y\leftand\fa))\leftand y,\\
x\fullor y&=(x\leftand(y\leftor\tr))\leftor y.
\end{align*}
For the setting in which next to atoms only negation, $\fulland$, $\fullor$ and the 
constants $\tr$ and $\fa$ occur (and thus no short-circuit connectives), 
an equational axiomatization of $se$-equality is 
provided in~\cite{Daan}.

\bigskip

In~\cite{Par10}, Parnas writes that
\emph{Most mainline methods disparage side effects as a bad programming practice. 
Yet even in well-structured, reliable software, many components do have side 
effects; side effects are very useful in practice. It is time to investigate 
methods that deal with side effects as the normal case.}
In the following we
argue that side effects occurring from short-circuit evaluation of propositional
statements as illustrated by  
Examples~\ref{ex:Perl}-\ref{ex:rp} can be analyzed with help of a partition of
the set $A$ of atoms (or the subset of $A$ that is relevant for analysis). 
Let 
\[A=A_{\sef}\cup A_{\see},\] 
where $A_{\sef}$
contains the atoms that under no circumstance can have a side effect in 
some condition
to be analyzed (for example, a simple test as \verb+($x==2)+ in a Perl program), say \sef-atoms, and $A_\see$ is the
set of those atoms that can have a side effect.
It is obvious that terms over $A_\sef$ are subject to \SSCL, and can be
simplified or rewritten to certain standard forms using the axioms of
\SSCLe\ (or propositional logic).
The set $A_\see$ can have
various explicit subsets
\begin{align*}
A_\see~\supseteq~ A_\rp ~\supseteq~ A_\con ~\supseteq~ A_\mem,
\end{align*}
of respectively repetition-proof atoms (\rp-atoms),  contractive atoms (\con-atoms),
and memorizing atoms (\mem-atoms). Then (sub)terms containing only \rp-atoms
are subject to \RPSCL, (sub)terms containing only \con-atoms to \CSCL, and
(sub)terms containing only \mem-atoms are subject to \MSCL.
In this way propositional statements can be rewritten or simplified while preserving 
all possible side effects in a setting
with ``mixed atoms''. Furthermore, taking into account that we can translate
left-sequential connectives that prescribe \emph{full evaluation} as explained above, 
we obtain a framework that is suitable for a systematic and atom-based analysis of 
``mixed propositional statements''.

\section{Conclusions}
\label{sec:Conc}

In~\cite{BP10} we introduced proposition 
algebra using Hoare's conditional $x\lef y\rig z$
and the constants $\tr$ and $\fa$.
We distinguished various valuation
congruences that are defined by means of short-circuit 
evaluation,
and provided axiomatizations of these congruences: 
$\CP$ (four axioms) characterizes the least identifying valuation congruence
we consider, and the extension $\CP_\mem$
(one extra axiom) characterizes the
most identifying valuation congruence below propositional
logic. 
In~\cite{HMA,BP12a} we provide an alternative valuation semantics
for proposition algebra in the form of \emph{Hoare-McCarthy algebras} (HMAs)
that is more elegant than the semantical framework
introduced in~\cite{BP10}: HMA-based semantics
has the advantage that one can define a valuation congruence
without first defining the 
valuation \emph{equivalence} it is contained in. 

\bigskip

This paper arose by an attempt to answer the question
whether the extension of $\CP_\mem$ with $\neg$ and
${\leftand}$ characterizes a reasonable logic
if one restricts to axioms defined over the signature
$\{\tr,\neg,\leftand\}$ (and with $\fa$
and $\leftor$ being definable).
After having found an axiomatization of \MSCL\ 
(memorizing short-circuit logic),
we defined
\FreeSCL\ (free short-circuit logic) as the 
most basic (least identifying) short-circuit logic,
where we took \CP\ as a 
point of departure. We used the module expression
\[
\SCL=\{\tr,\neg,\leftand\}\export(\CP
+ \langle\, \neg x=\fa\lef x\rig\tr\,\rangle
+ \langle\, x\leftand y=y\lef x\rig\fa\,\rangle)
\] 
in our generic definition of a short-circuit logic
(Definition~\ref{def:SCL}) and 
proved that \SCLe\ is an axiomatization of
\FreeSCL, and that \SSCLe\ axiomatizes \SSCL\ (static short-circuit logic).
The first proof is based on normal forms for \FreeSCL. 

Furthermore, we defined \CSCL\ (contractive short-circuit logic) and
\RPSCL\ (repetition-proof short-circuit logic) and 
listed some obvious axioms for the associated
extensions of \SCLe, but we left their
completeness as open questions. We provided 
two application-oriented examples (Examples~\ref{ex:con} and~\ref{ex:rp}), 
of which the latter is interesting
because it shows a typical difference with Example~\ref{ex:Perl}
(about Perl), or, more general, with examples taken from programming
languages that use the evaluation
of expressions as their standard semantics: in Example~\ref{ex:rp} 
the (Boolean) evaluation of
an assignment yields \true\ and is independent of its possible side 
effect, and therefore this example is a ``pure example''
for the case of repetition-proof short-circuit logic. 

Finally, we argued in Section~\ref{subsec:varia}
that different short-circuit logics can be used to rewrite or 
simplify propositional statements built from 
short-circuit connectives, full left-sequential connectives, and
atoms with different kinds of side effects. 
The question whether it would be fruitful and feasible 
to apply such rewriting to real-life examples in programming
and to develop tool-support for this purpose
is of course very interesting, and we think that this paper can 
serve as a basis for such future work.
In the remainder of this section
we mention some related
work, and we conclude with some more suggestions for future work.

\paragraph{Related work}
Short-circuit conjunction \texttt{\&\&} is often defined in imperative
programming as an associative operator. However, we did not succeed in finding 
work that provides a systematic answer (or discussion) to the question 
\emph{What are the logical laws that axiomatize short-circuit evaluation}? 
In the following we mention some work that addresses side effects and evaluation
strategies in a wider sense.
In~\cite{BW96,BW98}, Black and Windley have proposed a framework that extends
Hoare axiomatic semantics which removes side effects
from expressions and treats them as separate statements. 
In~\cite{PM}, Papaspyrou and Macos describe a study of evaluation order semantics 
in expressions with side effects and provide in their Section~5 a concise
overview of related work, which includes the above-mentioned references to work on 
side effects (this overview can also be found in~\cite{Pap}).

\paragraph{Future work}
We mention some issues that were not resolved or dealt with in the current paper 
and that suggest future work.
First, we propose an investigation to the applicability of 
$\omega$-completeness of \SCLe\footnote{If
$A=\emptyset$, $\SCLe$ is certainly not $\omega$-complete
because $x\leftand\fa=\fa$ is not derivable, but all
its closed instances are.}
 and the independence of its axioms.
(In~\cite{Chris} it is proved that
\CP\ is an $\omega$-complete axiomatization of free valuation
congruence if $A$ contains at least two atoms.)
Furthermore, we recall the two open questions raised
in Section~\ref{subsec:rpcon}:
\begin{enumerate}
\item
Is the extension of \SCLe\ 
with the axiom schemes in Table~\ref{tab:con}
an axiomatization of \CSCL?
\item Is the extension of \SCLe\ 
with the axiom schemes in Table~\ref{tab:rp}
an axiomatization of \RPSCL?
\end{enumerate}
It is questionable whether the approach we use to prove the
completeness of \FreeSCL\ 
can be lifted to these cases.
This would require a congruence that identifies different
evaluation trees (cf.~the approach in Appendix~\ref{app:Mem}), while 
at the same time unique decomposition of such trees is at stake.
The different types of basic forms for contractive and
repetition-proof valuation congruence
defined in~\cite{BP10} provide a first idea for this case and 
may suggest appropriate normal forms.

Finally, we mention that in~\cite{Ber10} a connection
is proposed between short-circuit logic and instruction
sequences as studied in \emph{Program algebra}~\cite{BL,BM}.
In particular, the relation between non-atomic test
instructions (i.e.,
test instructions involving left-sequential connectives)
and their decomposition in atomic tests and jump 
instructions is 
considered, evolving into a discussion about the length
of instruction sequences and their minimization. Also, 
the paper~\cite{Ber10} contains a 
discussion about a classification of side effects, derived
from
a classification of atoms (thus partitioning $A$, but in a different way
as is proposed in Section~\ref{subsec:varia}). 
We expect that these matters might lead to future work; they
certainly formed an inspiration for \cite{Wortel}.

\section{Digression: Program algebra revisited}
\label{sec:Digr}
The focus on left-sequential conjunction 
that is typical for this paper led to some
new results on proposition algebra. 
In Section~\ref{subsec:pa} we show that the valuation congruences
that we considered can be axiomatized in a purely incremental way, and
in Section~\ref{subsec:ele} we show that static valuation congruence has a 
very concise and elegant axiomatization.

\subsection{Incremental axiomatizations}
\label{subsec:pa}
The valuation congruences
that we considered in Section~\ref{sec:Hoare} can be
axiomatized in a purely incremental way:
the axiom systems $\CP_\rp$ up to and including  
$\CP_\stat$ as defined in~\cite{BP10}
all share the axioms
of \CP\ and each next system can be defined by 
the addition of
either one or two axioms, in most cases making 
previously added axiom(s) redundant. 
For the case of $\CP_\stat$ (see Section~\ref{subsec:CPstat})
we still have to prove
the following result.

\begin{proposition}
\label{prop:3}
The axiom sets
$\CP_\stat$ and 
$(\CP_\mem+\langle\, \fa\lef x\rig\fa=\fa\,\rangle)$ 
are equally strong.
\end{proposition}

\begin{proof} We show that all axioms in the one set are 
derivable from the other set.
We first prove that
the axiom~\eqref{CPmem} is derivable
from $\CP_\stat$:
\begin{align*}
\CP_\stat\vdash~ &x\lef y\rig(z\lef u\rig(v\lef y\rig w))
\\
&=x\lef y\rig((v\lef y\rig w)\lef(\fa\lef u\rig\tr)\rig z)
&&\text{by \eqref{cp4}}\\
&=x\lef y~\rig\\
&
\phantom{=~}
((v\lef(\fa\lef u\rig\tr)\rig z)\lef y\rig
(w\lef(\fa\lef u\rig\tr)\rig z))
&&\text{by \eqref{CPstat}}\\
&= x\lef y\rig(w\lef(\fa\lef u\rig\tr)\rig z)
&&\text{by \eqref{eq:contr}}\\
&= x\lef y\rig(z\lef u\rig w),
&&\text{by \eqref{cp4}}
\end{align*}
where the contraction law~\eqref{eq:contr}, that is
$x\lef y\rig(v\lef y\rig w)=x\lef y\rig w$,
is derivable from $\CP_\stat$: replace $y$ by $\fa\lef y\rig\tr$ in 
the $\CP_\stat$-axiom~\eqref{eq:contr2}, that is,
\[(x\lef y\rig z)\lef y\rig u=x\lef y\rig u.\]
Hence $\CP_\stat\vdash \eqref{CPmem}$.
Furthermore, we showed in Section~\ref{subsec:CPstat}
(equation~\eqref{eq:Hoare})
that 
$\CP_\stat\vdash \fa\lef x\rig\fa=\fa$.

In order to show that $(\CP_\mem+\langle\, \fa\lef x\rig\fa=\fa\,\rangle)\vdash
\CP_\stat$ recall that the contraction law~\eqref{eq:contr2}
is derivable from $\CP_\mem$ (see Section~\ref{subsec:CPmem}). 
So, it remains to be proved that
$(\CP_\mem+\langle\, \fa\lef x\rig\fa=\fa\,\rangle)\vdash
\eqref{CPstat}$. 
In order to give a short proof, we use our completeness result
on \SSCL\ (Theorem~\ref{thm:SSCL}) and the
commutativity of ${\leftand}$ which is derivable in 
\SSCL\
(see Lemma~\ref{lem:comm}).
From these results it follows that
\begin{equation}
\label{eq:comm}
(\CP_\mem+\langle\, \fa\lef x\rig\fa=\fa\,\rangle)\vdash~
y\lef x\rig\fa=x\lef y\rig\fa,
\end{equation}
and with this identity we can easily derive the axiom~\eqref{CPstat}:
\begin{align*}
(\CP_\mem+~&\langle\, \fa\lef x\rig\fa=\fa\,\rangle)\vdash~(x\lef y\rig z)\lef u\rig v\\
&=(x\lef y\rig (z\lef u\rig v))\lef u\rig v
&&\text{by \eqref{CPmem'}}\\
&=(x\lef y\rig (z\lef u\rig v))\lef u\rig (z\lef u\rig v)
&&\text{by \eqref{eq:contr}}\\
&=x\lef (y\lef u\rig\fa)\rig (z\lef u\rig v)
&&\text{by \eqref{cp4}}\\
&=x\lef (u\lef y\rig\fa)\rig (z\lef u\rig v)
&&\text{by \eqref{eq:comm}}\\
&=(x\lef u\rig(z\lef u\rig v))\lef y\rig (z\lef u\rig v)
&&\text{by \eqref{cp4}}\\
&=(x\lef u\rig v)\lef y\rig (z\lef u\rig v).
&&\text{by \eqref{eq:contr}}
\end{align*}
\end{proof}

Summing up, this yields the following scheme on the
proposition algebra axiomatizations discussed:
\begin{align*}
\CP_\rp&=\CP + \eqref{CPrp1} + \eqref{CPrp2},\\
\CP_\con&\dashv\vdash\CP_\rp + \eqref{CPcr1} + \eqref{CPcr2},\\
\CP_\mem&\dashv\vdash\CP_\con + \eqref{CPmem},\\
\CP_\stat&\dashv\vdash\CP_\mem + \langle\,\fa\lef x\rig \fa=\fa\,\rangle.
\end{align*}
(In~\cite{BP10} we also define \emph{weakly memorizing 
valuation congruence} that has an axiomatization
with the same property in between $\CP_\con$ 
and $\CP_\mem$.) 

\subsection{An elegant equational basis for static valuation congruence}
\label{subsec:ele}

Another result on proposition algebra that arose from our focus
on left-sequential conjunction
concerns a concise and elegant axiomatization of static valuation
congruence (see Section~\ref{subsec:CPstat}).  
In Table~\ref{CPstat2} we introduce a set of axioms that we call
\[\CP_\stat^*\]
and we shall prove that $\CP_\stat^*$ is equivalent with $\CP_\stat$
and that its axioms are independent.
First observe that $\CP_\stat^*$ 
consists of only five axioms that are
more simple than those of $\CP_\stat$ and  
$\CP_\mem + \langle\,\fa\lef x\rig \fa=\fa\,\rangle$.\footnote{Of course, 
  ``simple and few in number'' is 
  not an easy qualification in the setting of proposition algebra: 
  with the axioms $x\lef\tr\rig y=x$ and $x\lef\fa\rig y
  =y$, each other pair of axioms $L_1 = R_1$
  and $L_2 = R_2$ can be combined using a
  fresh variabele, say $u$, to a single axiom
  $L_1\lef u\rig L_2 = R_1 \lef u\rig R_2$. So, each 
  extension of \CP\ can be defined by adding a
  single (and ugly) axiom to \CP\ (and \CP\ itself can be axiomatized
  with three axioms).}
Furthermore, it is immediately
clear that the axioms \eqref{cp3*} and \eqref{cp5}
are derivable from $\CP_\stat$: \eqref{cp3*} is the conditional
expression for $x\leftor y=y\leftor x$, and \eqref{cp5} 
is an instance of the contraction law~\eqref{eq:contr2} of $\CP_\stat$.
Hence, $\CP_\stat\vdash \CP_\stat^*$.

\begin{table}
\centering
\hrule
\begin{align*}
	\tag{\ref{cp1}}
	x \lef \tr \rig y &= x\\
	\tag{\ref{cp2}}
	x \lef \fa \rig y &= y\\
	\label{cp3*}\tag{CP3$^*$} 
	\tr \lef x \rig y  &= \tr\lef
	y\rig x\\
    \tag{\ref{cp4}}
    x \lef (y \lef z \rig u)\rig v &= 
	(x \lef y \rig v) \lef z \rig (x \lef u \rig v)\\
	\label{cp5}\tag{CP5}
\qquad	(x\lef y\rig z)\lef y\rig \fa
   &=x\lef y\rig\fa
\end{align*}
\hrule
\caption{$\CP_\stat^*$, a set of axioms for static \CP}
\label{CPstat2}
\end{table}

We now show that $\CP_\stat^*\vdash \CP_\stat$ and first discuss some
intermediate results.
Axioms~\eqref{cp3*} and \eqref{cp2} imply
axiom~\eqref{cp3}
(i.e., $\tr\lef x\rig\fa=x$), thus $\CP_\stat^*\vdash\CP$.
Furthermore, with~\eqref{cp1} we find
$\tr\lef x\rig\tr=\tr\lef\tr\rig x=\tr$.
In the extension of $\CP_\stat^*$ with the defining equations for
$\neg$ and $\leftand$ (and ${\leftor}$ as a derived connective) we 
immediately find $x\leftor y=y\leftor x$
and $x\leftor\tr=\tr$. Also, observe that
the duality principle is derivable
(because it is in \CP, cf.~Section~\ref{subsec:Hoare}),
so $x\leftand y =y\leftand x$ and 
$x\leftand \fa=\fa$. 
With
\eqref{cp5} it follows that
\begin{align}
\label{D.1}
x\leftand \neg x&=(\fa\lef x\rig\tr)\lef x\rig\fa=\fa\lef x\rig\fa=\fa,\\
\label{D.2}
x\leftand x&=(\tr\lef x\rig\fa)\lef x\rig\fa=\tr\lef x\rig\fa=x.
\end{align}
Finally, distributivity is derivable in $\CP_\stat^*$ extended 
with $\neg$, $\leftand$ and $\leftor$:
\begin{align*}
x\leftor (y\leftand z)&=\tr\lef x\rig (z\lef y\rig \fa)\\
&=\tr\lef (z\lef y\rig\fa)\rig x
&&\text{by~\eqref{cp3*}}\\
&=(\tr\lef z\rig x)\lef y\rig x
&&\text{by~\eqref{cp4} and~\eqref{cp2}}\\
&=(\tr\lef x\rig z)\lef y\rig (\tr\lef x\rig\fa)
&&\text{by~\eqref{cp3*}}\\
&=(\tr\lef x\rig z)\lef y\rig ((\tr\lef x\rig z)\lef x\rig\fa)
&&\text{by~\eqref{cp5}}\\
&=(\tr\lef x\rig z)\lef (\tr\lef y\rig x)\rig \fa
&&\text{by~\eqref{cp4} and~\eqref{cp1}}\\
&=(\tr\lef x\rig z)\lef (\tr\lef x\rig y)\rig \fa
&&\text{by~\eqref{cp3*}}\\
&=(x\leftor y)\leftand (x\leftor z).
\end{align*}

With the identity $x\lef y\rig z =
(y\leftand x)\leftor(\neg y\leftand z)$ it is now a 
simple exercise to derive the $\CP_\stat$-axiom 
\eqref{CPstat} and the contraction law~\eqref{eq:contr2}: we
implicitly use commutativity and associativity
of the binary connectives, and equations~\eqref{D.1} and
\eqref{D.2}, and  derive
\begin{align*}
(x\lef u\rig v)&\lef y\rig (z\lef u\rig v)\\
&=(y\leftand[(u\leftand x)\leftor(\neg u\leftand v)])\leftor
(\neg y\leftand[(u\leftand z)\leftor(\neg u\leftand v)])\\
&=(y\leftand u\leftand x)\leftor(y\leftand \neg u\leftand v)
\leftor (\neg y\leftand u\leftand z)\leftor (\neg y
\leftand \neg u\leftand v)\\
&=(y\leftand u\leftand x)\leftor((y\leftor\neg y)
\leftand \neg u\leftand v)
\leftor (\neg y\leftand u\leftand z)\\
&=(u\leftand [(y\leftand x)\leftor(\neg y\leftand z)])
\leftor(\neg u \leftand v)\\
&=(x\lef y\rig z)\lef u\rig v,
\end{align*}
and
\begin{align*}
(x\lef y\rig z)\lef y\rig u
&=(y\leftand [(y\leftand x)\leftor(\neg y \leftand z)])\leftor 
(\neg y\leftand u)\\
&=(y\leftand y\leftand  x)\leftor(y\leftand \neg y \leftand z)
\leftor (\neg y\leftand u)\\
&=(y\leftand  x)\leftor (\neg y\leftand u)\\
&=x\lef y\rig u.
\end{align*}
Hence, $\CP_\stat^*\vdash \CP_\stat$. In~\cite{Chris} it is proved that $\CP_\stat$ 
is $\omega$-complete, so $\CP_\stat^*$ is $\omega$-complete
as well.

Next, we show that the axioms of $\CP_\stat^*$
are independent.
Inspired by~\cite{Chris,Sioson} we define five different 
independence-models, 
where $P,Q$ and $R$
range over closed terms, 
$a$ and $b$ are two atoms,
and $\phi$ is the interpretation function. The domain of
the first two
independence-models is $\{\tr,\fa\}$ and the interpretation 
refers to propositional logic.

\begin{enumerate}
\item The model defined by
$\phi(\tr)=\fa$, $\phi(\fa)=\phi(a)=\phi(b)=\tr$ and
\[\phi(P\lef Q\rig R)=\phi(Q)\wedge\phi(R)\] satisfies
all axioms but~\eqref{cp1}: $\phi(a)=\tr$ and $\phi(a\lef\tr\rig\fa)=\fa$.
\item The model defined by
$\phi(\tr)=\phi(a)=\phi(b)=\tr$, $\phi(\fa)=\fa$ and
\[\phi(P\lef Q\rig R)=\phi(P)\] satisfies
all axioms but~\eqref{cp2}: $\phi(a)=\tr$, while $\phi(\fa\lef\fa\rig a)=\fa$.
\item Any model for $\CP_\mem$ that does
not satisfy $T\lef a\rig b=T\lef b\rig a$,
satisfies all axioms but~\eqref{cp3*}. Such models
exist and are discussed in~\cite{BP10,HMA}.
\item The model with the natural numbers as its domain, 
and the interpretation defined by
\begin{align*}
\phi(\tr)&=0,~
\phi(\fa)=1,~
\phi(a)=2,~
\phi(b)=3,\\
\phi(P\lef Q\rig R)&=
\begin{cases}
\phi(P)&\text{if $\phi(Q)=0$,}\\
\phi(R)&\text{if $\phi(Q)=1$,}\\
\phi(Q)\cdot\phi(R)&\text{otherwise,}
\end{cases} 
\end{align*}
satisfies
all axioms but~\eqref{cp4}: first observe that
$\phi(\fa\lef a\rig \tr)=\phi(a)\cdot\phi(\tr)=0$ and
$\phi(\fa\lef\tr\rig\tr)=\phi(\fa)=1$, so 
\[\phi(\fa\lef(\fa\lef a\rig \tr)\rig \tr)=\phi(\fa)=1,\] 
while 
$\phi((\fa\lef\fa\rig \tr)\lef a\rig
(\fa\lef\tr\rig \tr))=\phi(a)\cdot\phi(\fa)=2$.

\item The model with the integers numbers as its domain, 
and the interpretation defined by
\begin{align*}
\phi(\tr)&=0,~
\phi(\fa)=1,~
\phi(a)=2,~\phi(b)=3,\\
\phi(P\lef Q\rig R)&=(1-\phi(Q))\cdot\phi(P)+
\phi(Q)\cdot\phi(R),
\end{align*}
satisfies
all axioms but~\eqref{cp5}: $\phi(\tr\lef a\rig \fa)=2$, 
while
\[\phi((\tr\lef a\rig \fa)\lef a\rig\fa)
=(1-2)\cdot 2+2\cdot 1=0.\]
\end{enumerate}

Finally, we note that $\CP_\stat^*$ has an elegant, 
symmetric property: exchanging~\eqref{cp3*} with 
$x\lef y\rig\fa=y\lef x\rig\fa$ and \eqref{cp5} with 
$\tr\lef x\rig(y\lef x\rig z)=\tr\lef x\rig z$ 
yields an equally strong axiomatization. 
Moreover, the resulting axioms
are also independent (modifying the above independence 
proofs is not difficult).

\section*{Acknowledgement}
We wish to thank two anonymous referees for their careful reviews
and constructive comments.

\newpage

\appendix

\section{Memorizing evaluations}
\label{app:Mem}
We define \emph{memorizing evaluation trees} (me-trees) in order to
model evaluations in which the evaluation result of an atom that was 
evaluated before is memorized.

\begin{definition}
Let $B\in\{\tr,\fa\}$.
The unary 
\textbf{memorizing short-circuit evaluation} function $mse:\PS\to \T$
is defined by
\[mse(P)=m(se(P)),\]
where the function $m:\T\to\T$ is defined for all $a\in A$ by
\begin{align*}
m(B)&=B,\\
m(X\unlhd a\unrhd Y)&=(m(X\sub a{\tr}))\unlhd a\unrhd(m(Y\sub a{\fa})),
\end{align*}
and the node reduction function
$\sub{a}B: \T\to\T$, notation $X\sub{a}B$, is defined by
\begin{align*}
\tr\sub aB&=\tr,\\
\fa\sub aB&=\fa,\\
(X\unlhd a\unrhd Y)\sub a\tr&=X\sub a{\tr},\\
(X\unlhd a\unrhd Y)\sub a\fa&=Y\sub a{\fa},\\
(X\unlhd b\unrhd Y)\sub aB&=X\sub aB\unlhd b\unrhd Y\sub aB
&&\text{if }b\ne a.
\end{align*}
For each closed term in \PS, we shall refer to $mse(P)$ as an
\textbf{me-tree} (memorizing evaluation tree).
\end{definition}
So, $mse$ maps a propositional statement $P$ to $se(P)$ and then
reduces $se(P)$ to an me-tree in the following way:
if $a$ is the root of $se(P)$ then each 
proper subtree in the \tr-branch of $se(P)$ with root $a$ is 
reduced to its \tr-branch, and each 
proper subtree in the \fa-branch of $se(P)$ with root $a$ is 
reduced to its \fa-branch. Then the same procedure is repeated for all nodes 
at depth 1, and so forth.

A simple example (using the extension to $\leftand$):
$mse(a\leftand (b\leftand a))=m(X)$ with $X$ the following tree
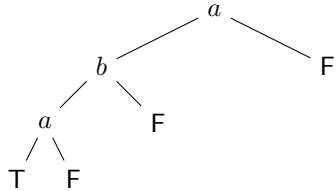
\[
\begin{tikzpicture}[%
level distance=7.5mm,
level 1/.style={sibling distance=30mm},
level 2/.style={sibling distance=15mm},
level 3/.style={sibling distance=7.5mm}
]
\node (a) {$a$}
  child {node (b1) {$b$}
    child {node (c1) {$a$}
      child {node (d1) {$\tr$}} 
      child {node (d2) {$\fa$}}
    }
    child {node (c2) {$\fa$}
    }
  }
  child {node (b2) {$\fa$}
  };
\end{tikzpicture}
\]
and thus $m(X)=(\tr\unlhd b\unrhd\fa)\unlhd a\unrhd\fa$. 
Note that $m(X)=mse(a\leftand b)=se(a\leftand b)$. 
Also, $mse(a\leftand (a\leftand a))=\tr\unlhd a\unrhd \fa$, thus
$mse(a)=mse(a\leftand (a\leftand a))=\tr\unlhd a\unrhd \fa$.
For $P,Q\in\PS$, we write
\[P=_\mem Q \iff mse(P)=mse(Q),\]
and the relation $=_\mem$ is called \emph{memorizing valuation congruence}.

\begin{theorem}
\label{thm:mem}
For all $P,Q\in\PS$, 
$\CP_\mem\vdash P=Q\iff P=_\mem Q$.
\end{theorem}

\begin{proof}
In~\cite{BP10} it is shown that memorizing valuation congruence $=_\mem$ 
as defined in that paper coincides with equality on a particular subset
of the basic forms (Definition~\ref{def:basic}), namely those basic forms
in which no atom occurs more than once. This subset of basic forms
coincides with the image of the function $mse$.
Using these basic forms, it easily follows that $=_\mem$ is a congruence relation 
and that all $\CP_\mem$-axioms are sound, which proves $\Longrightarrow$.

With respect to $\Longleftarrow$ we prove in~\cite{BP10} that 
$\CP_\mem$ axiomatizes $=_\mem$ by further reasoning on this particular 
subset of basic forms.
\end{proof}
We conjecture that those \SNF-terms (see Definition~\ref{def:snf})
in which no atom occurs more than once
constitute a set of \emph{normal forms}
for \MSCL, but we leave this matter open for future work.

\section{Correctness of the normalization function \nfs}
\label{app:nf}
In order to prove that $\nfs: \ST \to \SNF$ is indeed a normalization function
we need to prove that for all $\SCL$-terms $P$, $\nfs(P)$ terminates, $\nfs(P)
\in \SNF$ and $\EqFSCL \vdash \nfs(P) = P$. To arrive at this result, we prove
several intermediate results about the functions $\nfs^n$ and $\nfs^c$ in the
order in which their definitions were presented in Section~\ref{subsec:snf}. For
the sake of brevity we will not explicitly prove that these functions
terminate. To see that each function terminates consider that a termination
proof would closely mimic the proof structure of the lemmas dealing with the
grammatical categories of the images of these functions.

\begin{lemma}
\label{lem:ptpfs}
For all $P\in P^\fa$ and $Q\in P^\tr$, $\EqFSCL \vdash P = P \sleftand
x$ and $\EqFSCL \vdash Q = Q \sleftor x$.
\end{lemma}
\begin{proof}
We prove both claims simultaneously by induction. In the base case we have
$\fa = \fa \sleftand x$ by axiom~\eqref{SCL6}. The base case for the
second claim follows from that for the first claim by duality.

For the induction we have $(a \sleftor P_1) \sleftand P_2 = (a
\sleftor P_1) \sleftand (P_2 \sleftand x)$ by the induction
hypothesis and the result follows from \eqref{SCL7}. For the second claim we
again appeal to duality.
\end{proof}

The equality we showed as an example in Lemma~\ref{lem:seqs} 
(Section~\ref{subsec:SCLe}) will prove
useful in this appendix, as will the following equalities, which also deal
with terms of the form $x \sleftand \fa$ and $x \sleftor \tr$.

\begin{lemma}
\label{lem:seqs2}
The following equations can all be derived from $\EqFSCL$.
\begin{enumerate}
\item
$(x \sleftor (y \sleftand \fa)) \sleftand (z \sleftand \fa) = (\neg
  x \sleftor (z \sleftand \fa)) \sleftand (y \sleftand \fa),$
  \label{eq:b2}
\item
$(x \sleftand (y \sleftor \tr)) \sleftor (z \sleftand \fa) 
= (x
  \sleftor (z \sleftand \fa)) \sleftand (y \sleftor \tr),$
  \label{eq:b3}
\item $(x \sleftor \tr) \sleftand \neg y = \neg((x \sleftor \tr) \sleftand
  y),$
  \label{eq:b4}
\item $(x \sleftand (y \sleftand (z \sleftor \tr))) \sleftor
  (w \sleftand (z \sleftor \tr)) = ((x \sleftand y) \sleftor w) \sleftand
  (z \sleftor \tr),$
  \label{eq:b5}
\item $(x \sleftor ((y \sleftor \tr) \sleftand (z \sleftand \fa)))
  \sleftand ((w \sleftor \tr) \sleftand (z \sleftand \fa)) =
  ((x \sleftand (w \sleftor \tr)) \sleftor (y \sleftor \tr)) \sleftand (z
  \sleftand \fa),$
  \label{eq:b6}
\item $(x \sleftor ((y \sleftor \tr) \sleftand (z \sleftand \fa)))
  \sleftand (w \sleftand \fa) = ((\neg x \sleftand (y \sleftor \tr))
  \sleftor (w \sleftand \fa)) \sleftand (z \sleftand \fa).$
  \label{eq:b7}
\end{enumerate}
\end{lemma}

\begin{proof}
We derive the equations in order in Table~\ref{tab:proof}.
\end{proof}

\begin{table}
\centering
\hrule
\small\begin{align*}
(x&\sleftor (y \sleftand \fa)) \sleftand (z \sleftand
  \fa)\\
&= (\neg x \sleftor (z \sleftand \fa)) \sleftand ((y \sleftand \fa)
  \sleftand (z \sleftand \fa))
&&\textrm{by Lemma~\ref{lem:seqs}} \\
&= (\neg x \sleftor (z \sleftand \fa)) \sleftand (y \sleftand \fa),
&&\textrm{by \eqref{SCL6} and \eqref{SCL7}} 
\\[2mm]
(x&\sleftand (y \sleftor \tr)) \sleftor (z \sleftand
  \fa)\\
&= (x \sleftor (z \sleftand \fa)) \sleftand ((y \sleftor \tr) \sleftor (z
  \sleftand \fa))
&&\textrm{by \eqref{SCL10}} \\
&= (x \sleftor (z \sleftand \fa)) \sleftand (y \sleftor \tr),
&&\textrm{by $\eqref{SCL6}'$ and $\eqref{SCL7}'$}
\\[2mm]
(x &\sleftor \tr)\sleftand \neg y\\
&= \neg((\neg x \sleftand \fa) \sleftor y)
&&\textrm{by duality} \\
&= \neg((x \sleftand \fa) \sleftor y)
&&\textrm{by \eqref{SCL8}} \\
&= \neg((x \sleftor \tr) \sleftand y),
&&\textrm{by \eqref{SCL9}}
\\[2mm]
(x &\sleftand (y \sleftand (z \sleftor \tr))) \sleftor
  (w \sleftand (z \sleftor \tr)) \\
&= ((x \sleftand y) \sleftand (z \sleftor \tr)) \sleftor
  (w \sleftand (z \sleftor \tr))
&&\textrm{by \eqref{SCL7}} \\
&= ((x \sleftand y) \sleftor w) \sleftand (z \sleftor \tr),
&&\textrm{by $\eqref{SCL10}'$}
\\[2mm]
(x &\sleftor ((y \sleftor \tr) \sleftand (z \sleftand \fa)))
  \sleftand ((w \sleftor \tr) \sleftand (z \sleftand \fa)) \\
&= (x \sleftor ((y \sleftand \fa) \sleftor (z \sleftand \fa)))
\\
&\phantom{=~} \sleftand~((w \sleftor \tr) \sleftand (z \sleftand \fa)) 
&&\textrm{by \eqref{SCL9}} \\
&= ((x \sleftor (y \sleftand \fa)) \sleftor (z \sleftand \fa))
\\
&\phantom{=~}\sleftand~((w \sleftor \tr) \sleftand (z \sleftand \fa)) 
&&\textrm{by $\eqref{SCL7}'$} \\
&= (\neg(x \sleftor (y \sleftand \fa)) \sleftor (w \sleftor \tr))
  \sleftand (z \sleftand \fa)
&&\textrm{by Lemma~\ref{lem:seqs}} \\
&= ((\neg x \sleftand (\neg y \sleftor \tr)) \sleftor (w \sleftor \tr))
  \sleftand (z \sleftand \fa)
&&\textrm{by duality} \\
&= ((\neg x \sleftand (y \sleftor \tr)) \sleftor (w \sleftor \tr))
  \sleftand (z \sleftand \fa)
&&\textrm{by $\eqref{SCL8}'$} \\
&= ((\neg x \sleftand (y \sleftor \tr)) \sleftor (w \sleftor (\tr
  \sleftor (y \sleftor \tr))))\\
&\phantom{=~}\sleftand~ (z \sleftand \fa)
&&\textrm{by $\eqref{SCL6}'$} \\
&= ((\neg x \sleftand (y \sleftor \tr)) \sleftor ((w \sleftor \tr)
  \sleftor (y \sleftor \tr)))\\
&\phantom{=~}\sleftand~(z \sleftand \fa)
&&\textrm{by $\eqref{SCL7}'$} \\
&= ((x \sleftand (w \sleftor \tr)) \sleftor (y \sleftor \tr))
  \sleftand (z \sleftand \fa),
&&\textrm{by Lemma~$\ref{lem:seqs}'$} 
\\[2mm]
(x&\sleftor ((y \sleftor \tr) \sleftand (z \sleftand
  \fa))) \sleftand (w \sleftand \fa) \\
&= (\neg x \sleftor (w \sleftand \fa)) \sleftand (((y \sleftor \tr)
  \sleftand (z \sleftand \fa))\\
&\phantom{=~}\sleftand~ (w \sleftand \fa))
&&\textrm{by Lemma~\ref{lem:seqs}} \\
&= (\neg x \sleftor (w \sleftand \fa)) \sleftand ((y \sleftor \tr)
  \sleftand (z \sleftand \fa))
&&\textrm{by \eqref{SCL6} and \eqref{SCL7}} \\
&= ((\neg x \sleftor (w \sleftand \fa)) \sleftand (y \sleftor \tr))
  \sleftand (z \sleftand \fa)
&&\textrm{by \eqref{SCL7}} \\
&= ((\neg x \sleftand (y \sleftor \tr)) \sleftor ((w \sleftand \fa)
  \sleftand (y \sleftor \tr)))\\
&\phantom{=~}\sleftand~(z \sleftand \fa)
&&\textrm{by $\eqref{SCL10}'$} \\
&= ((\neg x \sleftand (y \sleftor \tr)) \sleftor (w \sleftand \fa))
  \sleftand (z \sleftand \fa).
&&\textrm{by \eqref{SCL6} and \eqref{SCL7}}
  &\qedhere
\end{align*}
\hrule
\caption{Derivations for the proof of Lemma~\ref{lem:seqs2}}
\label{tab:proof}
\end{table}

\begin{lemma}
\label{lem:nfsn}
For all $P \in \SNF$, if $P$ is a $\tr$-term then $\nfs^n(P)$ is an
$\fa$-term, if it is an $\fa$-term then $\nfs^n(P)$ is a $\tr$-term, if
it is a $\tr$-$*$-term then so is $\nfs^n(P)$, and
\begin{equation*}
\EqFSCL \vdash \nfs^n(P) = \neg P.
\end{equation*}
\end{lemma}
\begin{proof}
We first prove the claims for $\tr$-terms, by induction on $P^\tr$.  In the
base case $\nfs^n(\tr) = \fa$ by \eqref{eq:nfsn1}, so
$\nfs^n(\tr)$ is an
$\fa$-term. The claim that $\EqFSCL \vdash \nfs^n(\tr) = \neg \tr$ is
immediate by \eqref{SCL1}. For the inductive case we have that $\nfs^n((a
\sleftand P^\tr) \sleftor Q^\tr) = (a \sleftor \nfs^n(Q^\tr)) \sleftand
\nfs^n(P^\tr)$ by \eqref{eq:nfsn2}, where we assume that $\nfs^n(P^\tr)$ and 
$\nfs^n(Q^\tr)$
are $\fa$-terms and that $\EqFSCL \vdash \nfs^n(P^\tr) = \neg P^\tr$
and $\EqFSCL \vdash \nfs^n(Q^\tr) = \neg Q^\tr$. It follows from the
induction hypothesis that $\nfs^n((a \sleftand P^\tr) \sleftor Q^\tr)$ is
an $\fa$-term. Furthermore, noting that by the induction hypothesis we may
assume that $\nfs^n(P^\tr)$ and $\nfs^n(Q^\tr)$ are $\fa$-terms, we
have:
\begin{align*}
\nfs^n((a&\sleftand P^\tr) \sleftor Q^\tr)\\
&= (a \sleftor \nfs^n(Q^\tr)) \sleftand \nfs^n(P^\tr)
&&\textrm{by \eqref{eq:nfsn2}} \\
&= (a \sleftor (\nfs^n(Q^\tr) \sleftand \fa)) \sleftand (\nfs^n(P^\tr)
  \sleftand \fa)
&&\textrm{by Lemma~\ref{lem:ptpfs}} \\
&= (\neg a \sleftor (\nfs^n(P^\tr) \sleftand \fa)) \sleftand
  (\nfs^n(Q^\tr) \sleftand \fa )
&&\textrm{by Lemma~\ref{lem:seqs2}.\ref{eq:b2}} \\
&= (\neg a \sleftor \nfs^n(P^\tr)) \sleftand \nfs^n(Q^\tr)
&&\textrm{by Lemma~\ref{lem:ptpfs}} \\
&= (\neg a \sleftor \neg P^\tr) \sleftand \neg Q^\tr
&&\textrm{by induction hypothesis} \\
&= \neg((a \sleftand P^\tr) \sleftor Q^\tr).
&&\textrm{by \eqref{SCL2} and its dual}
\end{align*}

For $\fa$-terms we prove our claims by induction on $P^\fa$. In the base
case $\nfs^n(\fa) = \tr$ by \eqref{eq:nfsn3}, so $\nfs^n(\fa)$ is a
$\tr$-term. The claim that $\EqFSCL \vdash \nfs^n(\fa) = \neg \fa$ is
immediate by the dual of \eqref{SCL1}. For the inductive case we have that
$\nfs^n((a \sleftor P^\fa) \sleftand Q^\fa) = (a \sleftand
\nfs^n(Q^\fa)) \sleftor \nfs^n(P^\fa)$ by \eqref{eq:nfsn4}, where we assume that
$\nfs^n(P^\fa)$ and $\nfs^n(Q^\fa)$ are $\tr$-terms and that $\EqFSCL
\vdash \nfs^n(P^\fa) = \neg P^\fa$ and $\EqFSCL \vdash \nfs^n(Q^\fa) =
\neg Q^\fa$. It follows from the induction hypothesis that $\nfs^n((a
\sleftor P^\fa) \sleftand Q^\fa)$ is a $\tr$-term. Furthermore, noting
that by the induction hypothesis we may assume that $\nfs^n(P^\fa)$ and
$\nfs^n(Q^\fa)$ are $\tr$-terms, the proof of derivably equality is dual
to that for $\nfs^n((a \sleftand P^\tr) \sleftor Q^\tr)$.

To prove the lemma for $\tr$-$*$-terms we first verify that the auxiliary
function $\nfs^n_1$ returns a $*$-term and that for any $*$-term $P$, $\EqFSCL
\vdash \nfs^n_1(P) = \neg P$. We show this by induction on the number of
$\ell$-terms in $P$. For the base cases it is immediate by the above cases
for $\tr$-terms and $\fa$-terms that $\nfs^n_1(P)$ is a $*$-term.
Furthermore, if $P$ is an $\ell$-term with a positive first atom we
have:
\begin{align*}
\nfs^n_1((a&\sleftand P^\tr) \sleftor Q^\fa)\\
&= (\neg a \sleftand \nfs^n(Q^\fa)) \sleftor \nfs^n(P^\tr)
&&\textrm{by \eqref{eq:nfsn6}} \\
&= (\neg a \sleftand (\nfs^n(Q^\fa) \sleftor \tr)) \sleftor
  (\nfs^n(P^\tr) \sleftand \fa)
&&\textrm{by Lemma~\ref{lem:ptpfs}} \\
&= (\neg a \sleftor (\nfs^n(P^\tr) \sleftand \fa)) \sleftand
  (\nfs^n(Q^\fa) \sleftor \tr)
&&\textrm{by Lemma~\ref{lem:seqs2}.\ref{eq:b3}} \\
&= (\neg a \sleftor \nfs^n(P^\tr)) \sleftand \nfs^n(Q^\fa)
&&\textrm{by Lemma~\ref{lem:ptpfs}} \\
&= (\neg a \sleftor \neg P^\tr) \sleftand \neg Q^\fa
&&\textrm{by induction hypothesis} \\
&= \neg((a \sleftand P^\tr) \sleftor Q^\fa).
&&\textrm{by \eqref{SCL2} and its dual}
\end{align*}
If $P$ is an $\ell$-term with a negative first atom the proof proceeds
the same, substituting $\neg a$ for $a$ and applying \eqref{eq:nfsn7}
and \eqref{SCL3} where
needed. For the inductive step we assume that the result holds for all
$*$-terms with fewer $\ell$-terms than $P^* \sleftand Q^d$ and $P^* \sleftor
Q^c$. By \eqref{eq:nfsn8} and \eqref{eq:nfsn9},
each application of $\nfs^n_1$ changes the main connective
(not occurring inside an $\ell$-term) and hence the result is a $*$-term.
Derivable equality is, given the induction hypothesis, an instance of (the dual
of) \eqref{SCL2}.

With this result we can now see that $\nfs^n(P^\tr \sleftand Q^*)$ is indeed
a $\tr$-$*$-term. We note that, by the above, Lemma~\ref{lem:ptpfs}
implies that $\neg P^\tr = \neg P^\tr \sleftand \fa$. Now we find that:
\begin{align*}
\nfs^n(P^\tr \sleftand Q^*)
&= P^\tr \sleftand \nfs^n_1(Q^*)
&&\textrm{by \eqref{eq:nfsn5}} \\
&= P^\tr \sleftand \neg Q^*
&&\textrm{as shown above} \\
&= (P^\tr \sleftor \tr) \sleftand \neg Q^*
&&\textrm{by Lemma~\ref{lem:ptpfs}} \\
&= \neg((P^\tr \sleftor \tr) \sleftand Q^*)
&&\textrm{by Lemma~\ref{lem:seqs2}.\ref{eq:b4}} \\
&= \neg(P^\tr \sleftand Q^*).
&&\textrm{by Lemma~\ref{lem:ptpfs}}
\end{align*}
Hence for all $P \in \SNF$, $\EqFSCL \vdash \nfs^n(P) = \neg P$.
\end{proof}

\begin{lemma}
\label{lem:nfsc1}
For any $\tr$-term $P$ and $Q \in \SNF$, $\nfs^c(P, Q)$ has the same
grammatical category as $Q$ and
\begin{equation*}
\EqFSCL \vdash \nfs^c(P, Q) = P \sleftand Q.
\end{equation*}
\end{lemma}
\begin{proof}
By induction on the complexity of the $\tr$-term. In the base case we see
that $\nfs^c(\tr, P) = P$ by \eqref{eq:nfsc1}, which is clearly of the same grammatical category
as $P$. Derivable equality is an instance of \eqref{SCL4}.

For the inductive step we assume that the result holds for all $\tr$-terms of
lesser complexity than $(a \sleftand P^\tr)\leftor Q^\tr$. The claim about the grammatical
category follows immediately from the induction hypothesis. For the claim about
derivable equality we make a case distinction on the grammatical category of
the second argument. If the second argument is a $\tr$-term, we prove
derivable equality as follows:
\begin{align*}
\nfs^c((a&\sleftand P^\tr) \sleftor Q^\tr, R^\tr) \\
&= (a \sleftand \nfs^c(P^\tr, R^\tr)) \sleftor \nfs^c(Q^\tr, R^\tr)
&&\textrm{by \eqref{eq:nfsc2}} \\
&= (a \sleftand (P^\tr \sleftand R^\tr)) \sleftor (Q^\tr \sleftand
  R^\tr)
&&\textrm{by induction hypothesis} \\
&= (a \sleftand (P^\tr \sleftand (R^\tr \sleftor \tr))) \sleftor (Q^\tr
  \sleftand (R^\tr \sleftor \tr))
&&\textrm{by Lemma~\ref{lem:ptpfs}} \\
&= ((a \sleftand P^\tr) \sleftor Q^\tr) \sleftand (R^\tr \sleftor \tr)
&&\textrm{by Lemma~\ref{lem:seqs2}.\ref{eq:b5}} \\
&= ((a \sleftand P^\tr) \sleftor Q^\tr) \sleftand R^\tr.
&&\textrm{by Lemma~\ref{lem:ptpfs}}
\end{align*}
If the second argument is an $\fa$-term, we prove derivable equality as
follows:
\begin{align*}
\nfs^c((a &\sleftand P^\tr) \sleftor Q^\tr, R^\fa) \\
&= (a \sleftor \nfs^c(Q^\tr, R^\fa)) \sleftand \nfs^c(P^\tr, R^\fa)
&&\textrm{by \eqref{eq:nfsc3}} \\
&= (a \sleftor (Q^\tr \sleftand R^\fa)) \sleftand (P^\tr \sleftand
  R^\fa)
&&\textrm{by induction hypothesis} \\
&= (a \sleftor ((Q^\tr \sleftor \tr) \sleftand (R^\fa \sleftand \fa)))
  \sleftand~\\
  &\phantom{~=}((P^\tr \sleftor \tr) \sleftand (R^\fa \sleftand \fa))
&&\textrm{by Lemma~\ref{lem:ptpfs}} \\
&= ((a \sleftand (P^\tr \sleftor \tr)) \sleftor (Q^\tr \sleftor \tr))
  \sleftand (R^\fa \sleftand \fa)
&&\textrm{by Lemma~\ref{lem:seqs2}.\ref{eq:b6}} \\
&= ((a \sleftand P^\tr) \sleftor Q^\tr) \sleftand R^\fa.
&&\textrm{by Lemma~\ref{lem:ptpfs}}
\end{align*}

If the second argument is $\tr$-$*$-term, the result follows 
by \eqref{eq:nfsc4} from the case
where the second argument is a $\tr$-term, and \eqref{SCL7}.
\end{proof}

\begin{lemma}
\label{lem:nfsc2}
For any $\fa$-term $P$ and $Q \in \SNF$, $\nfs^c(P, Q)$ is an
$\fa$-term and
\begin{equation*}
\EqFSCL \vdash \nfs^c(P, Q) = P \sleftand Q.
\end{equation*}
\end{lemma}
\begin{proof}
The grammatical result is immediate by \eqref{eq:nfsc5}
and the claim about derivable equality
follows from Lemma~\ref{lem:ptpfs}, \eqref{SCL7} and \eqref{SCL6}.
\end{proof}

\begin{lemma}
\label{lem:nfsc3}
For any $\tr$-$*$-term $P$ and $\tr$-term $Q$, $\nfs^c(P, Q)$ has the same
grammatical category as $P$ and
\begin{equation*}
\EqFSCL \vdash \nfs^c(P, Q) = P \sleftand Q.
\end{equation*}
\end{lemma}
\begin{proof}
By \eqref{eq:nfsc6} and \eqref{SCL7}
it suffices to prove the claims for $\nfs^c_1$, i.e., that
$\nfs^c_1(P^*, Q^\tr)$ is a $*$-term and that $\EqFSCL \vdash \nfs^c_1(P^*,
Q^\tr) = P^* \sleftand Q^\tr$. We prove this by induction on the number of
$\ell$-terms in $P^*$. In the base case we deal with $\ell$-terms and the
grammatical claim follows from Lemma~\ref{lem:nfsc1}. We prove derivable
equality as follows, letting $\hat{a} \in \{a, \neg a\}$:
\begin{align*}
\nfs^c_1((\hat{a} \sleftand P^\tr) \sleftor Q^\fa, R^\tr)
&= (\hat{a} \sleftand \nfs^c(P^\tr, R^\tr)) \sleftor Q^\fa
&&\textrm{by \eqref{eq:nfsc7}, \eqref{eq:nfsc8}} \\
&= (\hat{a} \sleftand (P^\tr \sleftand R^\tr)) \sleftor Q^\fa
&&\textrm{by Lemma~\ref{lem:nfsc1}} \\
&= ((\hat{a} \sleftand P^\tr) \sleftand R^\tr) \sleftor Q^\fa
&&\textrm{by \eqref{SCL7}} \\
&= ((\hat{a} \sleftand P^\tr) \sleftand (R^\tr \sleftor \tr))
  \sleftor (Q^\fa \sleftand \fa)
&&\textrm{by Lemma~\ref{lem:ptpfs}} \\
&= ((\hat{a} \sleftand P^\tr) \sleftor (Q^\fa \sleftand \fa))
  \sleftand (R^\tr \sleftor \tr)
&&\textrm{by Lemma~\ref{lem:seqs2}.\ref{eq:b3}} \\
&= ((\hat{a} \sleftand P^\tr) \sleftor Q^\fa) \sleftand R^\tr.
&&\textrm{by Lemma~\ref{lem:ptpfs}}
\end{align*}

For the induction step we assume that the result holds for all $*$-terms with
fewer $\ell$-terms than $P^* \sleftand Q^d$ and $P^* \sleftor Q^c$.  In the
case of conjunctions the results follow from \eqref{eq:nfsc9}, the induction 
hypothesis, and \eqref{SCL7}. 
In the case of disjunctions the results follow immediately
from \eqref{eq:nfsc10}, the induction hypothesis, Lemma~\ref{lem:ptpfs}, 
and the dual of \eqref{SCL10}.
\end{proof}

\begin{lemma}
\label{lem:nfsc4}
For any $\tr$-$*$-term $P$ and $\fa$-term $Q$, $\nfs^c(P, Q)$ is an
$\fa$-term and
\begin{equation*}
\EqFSCL \vdash \nfs^c(P, Q) = P \sleftand Q.
\end{equation*}
\end{lemma}
\begin{proof}
By \eqref{eq:nfsc11}, Lemma~\ref{lem:nfsc1} and \eqref{SCL7} it suffices 
to prove that
$\nfs^c_2(P^*, Q^\fa)$ is an $\fa$-term and that $\EqFSCL \vdash
\nfs^c_2(P^*, Q^\fa) = P^* \sleftand Q^\fa$. We prove this by induction
on the number of $\ell$-terms in $P^*$. In the base case we deal with
$\ell$-terms and the grammatical claim follows from Lemma~\ref{lem:nfsc1}. We
derive the remaining claim for $\ell$-terms with positive first atoms
as:
\begin{align*}
\nfs^c_2((a \sleftand P^\tr) \sleftor Q^\fa, R^\fa)
&= (a \sleftor Q^\fa) \sleftand \nfs^c(P^\tr, R^\fa)
&&\textrm{by \eqref{eq:nfsc12}} \\
&= (a \sleftor Q^\fa) \sleftand (P^\tr \sleftand R^\fa)
&&\textrm{by Lemma~\ref{lem:nfsc1}} \\
&= ((a \sleftor Q^\fa) \sleftand P^\tr) \sleftand R^\fa
&&\textrm{by \eqref{SCL7}} \\
&= ((a \sleftor (Q^\fa \sleftand \fa)) \sleftand (P^\tr \sleftor
  \tr)) \sleftand R^\fa
&&\textrm{by Lemma~\ref{lem:ptpfs}} \\
&= ((a \sleftand (P^\tr \sleftor \tr)) \sleftor (Q^\fa \sleftand
  \fa)) \sleftand R^\fa
&&\textrm{by Lemma~\ref{lem:seqs2}.\ref{eq:b3}} \\
&= ((a \sleftand P^\tr) \sleftor Q^\fa) \sleftand R^\fa.
&&\textrm{by Lemma~\ref{lem:ptpfs}}
\end{align*}
For $\ell$-terms with negative first atoms we derive:
\begin{align*}
\nfs^c_2((&\neg a\sleftand P^\tr) \sleftor Q^\fa, R^\fa) \\
&= (a \sleftor \nfs^c(P^\tr, R^\fa)) \sleftand Q^\fa
&&\textrm{by \eqref{eq:nfsc13}} \\
&= (a \sleftor (P^\tr \sleftand R^\fa)) \sleftand Q^\fa
&&\textrm{by induction hypothesis} \\
&= (a \sleftor ((P^\tr \sleftor \tr) \sleftand (R^\fa \sleftand
  \fa))) \sleftand (Q^\fa \sleftand \fa)
&&\textrm{by Lemma~\ref{lem:ptpfs}} \\
&= ((\neg a \sleftand (P^\tr \sleftor \tr)) \sleftor (Q^\fa \sleftand
  \fa)) \sleftand (R^\fa \sleftand \fa)
&&\textrm{by Lemma~\ref{lem:seqs2}.\ref{eq:b7}} \\
&= ((\neg a \sleftand P^\tr) \sleftor Q^\fa) \sleftand R^\fa.
&&\textrm{by Lemma~\ref{lem:ptpfs}}
\end{align*}

For the induction step we assume that the result holds for all $*$-terms with
fewer $\ell$-terms than $P^* \sleftand Q^d$ and $P^* \sleftor Q^c$.  In the
case of conjunctions the results follow from \eqref{eq:nfsc14},
the induction hypothesis, and
\eqref{SCL7}. In the case of disjunctions note that by Lemma~\ref{lem:nfsn}
and the proof of Lemma~\ref{lem:nfsc3}, we have that $\nfs^n(\nfs^c_1(P^*,
\nfs^n(R^\fa)))$ is a $*$-term with same number of $\ell$-terms as $P^*$.
The grammatical result follows from this fact, \eqref{eq:nfsc15}, and the 
induction hypothesis.
Furthermore, noting that by the same argument $\nfs^n(\nfs^c_1(P^*,
\nfs^n(R^\fa))) = \neg(P^* \sleftand \neg R^\fa)$, we derive:
\begin{align*}
\nfs^c_2(P^* \sleftor &Q^c, R^\fa)\\
&= \nfs^c_2(\nfs^n(\nfs^c_1(P^*, \nfs^n(R^\fa))), \nfs^c_2(Q^c, R^\fa))
&&\textrm{by \eqref{eq:nfsc15}} \\
&= \nfs^n(\nfs^c_1(P^*, \nfs^n(R^\fa))) \sleftand (Q^c \sleftand R^\fa)
&&\textrm{by induction hypothesis} \\
&= \neg(P^* \sleftand \neg R^\fa) \sleftand (Q^c \sleftand R^\fa)
&&\textrm{as shown above} \\
&= (\neg P^* \sleftor R^\fa) \sleftand (Q^c \sleftand R^\fa)
&&\textrm{by \eqref{SCL3} and \eqref{SCL2}} \\
&= (\neg P^* \sleftor (R^\fa \sleftand \fa)) \sleftand (Q^c \sleftand
  (R^\fa \sleftand \fa))
&&\textrm{by Lemma~\ref{lem:ptpfs}} \\
&= (P^* \sleftor Q^c) \sleftand (R^\fa \sleftand \fa)
&&\textrm{by Lemma~\ref{lem:seqs}} \\
&= (P^* \sleftor Q^c) \sleftand R^\fa.
&&\textrm{by Lemma~\ref{lem:ptpfs}}
\end{align*}
This completes the proof.
\end{proof}

\begin{lemma}
\label{lem:nfsc5}
For any $P, Q \in \SNF$, $\nfs^c(P, Q)$ is in $\SNF$ and
\begin{equation*}
\EqFSCL \vdash \nfs^c(P, Q) = P \sleftand Q.
\end{equation*}
\end{lemma}
\begin{proof}
By the four preceding lemmas it suffices to show that 
\[\nfs^c(P^\tr \sleftand Q^*, R^\tr \sleftand S^*)\]
is in $\SNF$ and that $\EqFSCL \vdash
\nfs^c(P^\tr \sleftand Q^*, R^\tr \sleftand S^*) = (P^\tr \sleftand Q^*)
\sleftand (R^\tr \sleftand S^*)$. By \eqref{SCL7} and \eqref{eq:nfsc16}, in turn, it suffices to
prove that $\nfs^c_3(P^*, Q^\tr \sleftand R^*)$ is a $*$-term and that
$\EqFSCL \vdash \nfs^c_3(P^*, Q^\tr \sleftand R^*) = P^* \sleftand (Q^\tr
\sleftand R^*)$. We prove this by induction on the number of $\ell$-terms in
$R^*$. In the base case we have that $\nfs^c_3(P^*, Q^\tr \sleftand R^\ell)
= \nfs^c_1(P^*, Q^\tr) \sleftand R^\ell$ by \eqref{eq:nfsc17} and
the lemma's statement follows from Lemma~\ref{lem:nfsc3} and \eqref{SCL7}.

For conjunctions the lemma's statement follows from the induction hypothesis,
\eqref{SCL7} and \eqref{eq:nfsc18},
and for disjunctions it follows from Lemma~\ref{lem:nfsc3},
 \eqref{SCL7} and \eqref{eq:nfsc19}.
\end{proof}

We can now easily prove Theorem~\ref{thm:nfs}:
\textit{For any $P \in \ST$, $\nfs(P)$ terminates, $\nfs(P) \in \SNF$ and $\EqFSCL
\vdash \nfs(P) = P$.
}
\begin{proof}[Proof of Theorem~\ref{thm:nfs}]
By induction on the structure of $P$. If $P$ is an atom, the result follows from
\eqref{eq:nfs1} and axioms~\eqref{SCL4}, \eqref{SCL5} and its dual. 
If $P$ is $\tr$ or $\fa$
the result follows from by \eqref{eq:nfs2} or \eqref{eq:nfs3}. 
For the induction we get the result from definitions \eqref{eq:nfs4}-\eqref{eq:nfs6}, Lemma
\ref{lem:nfsn}, Lemma~\ref{lem:nfsc5}, and axiom~\eqref{SCL2}.
\end{proof}

\section{Correctness of the inverse function $\invs$}
\label{app:sclinvs}
In this appendix we prove Theorem~\ref{thm:sclinv}:
\textit{
For all $P \in \SNF$, $\invs(\SE(P)) \equiv P$.
}
\\
Recall that we use the symbol $\equiv$ to denote syntactic equivalence.

\begin{proof}[Proof of Theorem~\ref{thm:sclinv}]
We first prove that for all $\tr$-terms $P$, $\invs^\tr(\SE(P)) \equiv P$,
by induction on $P$. In the base case $P \equiv \tr$ and we have
by~\eqref{eq:g1} that
$\invs^\tr(\SE(P)) \equiv \invs^\tr(\tr) \equiv \tr \equiv P$. For the
inductive case we have $P \equiv (a \sleftand Q^\tr) \sleftor R^\tr$ and
\begin{align*}
\invs^\tr(\SE(P)) &\equiv \invs^\tr(\SE(Q^\tr) \tlef a \trig
  \SE(R^\tr))
&&\textrm{by definition of $\SE$} \\
&\equiv (a \sleftand \invs^\tr(\SE(Q^\tr))) \sleftor
  \invs^\tr(\SE(R^\tr))
&&\textrm{by \eqref{eq:g1}} \\
&\equiv (a \sleftand Q^\tr) \sleftor R^\tr
&&\textrm{by induction hypothesis} \\
&\equiv P.
\end{align*}

Similarly we see that for all $\fa$-terms $P$, $\invs^\fa(\SE(P)) \equiv
P$, by induction on $P$. In the base case $P \equiv \fa$ and we have
by~\eqref{eq:g2} that
$\invs^\fa(\SE(P)) \equiv \invs^\fa(\fa) \equiv \fa \equiv P$. For
the inductive case we have $P \equiv (a \sleftor Q^\fa) \sleftand R^\fa$
and
\begin{align*}
\invs^\fa(\SE(P)) &\equiv \invs^\fa(\SE(R^\fa) \tlef a \trig
  \SE(Q^\fa))
&&\textrm{by definition of $\SE$} \\
&\equiv (a \sleftor \invs^\fa(\SE(Q^\fa))) \sleftand
  \invs^\fa(\SE(R^\fa))
&&\textrm{by~\eqref{eq:g2}} \\
&\equiv (a \sleftor Q^\fa) \sleftand R^\fa
&&\textrm{by induction hypothesis} \\
&\equiv P.
\end{align*}

Next we check that for all $\ell$-terms $P$, $\invs^\ell(\SE(P)) \equiv P$.  We
observe that either $P \equiv (a \sleftand Q^\tr) \sleftor R^\fa$ or $P
\equiv (\neg a \sleftand Q^\tr) \sleftor R^\fa$. In the first case we have
\begin{align*}
\invs^\ell(\SE(P)) &\equiv \invs^\ell(\SE(Q^\tr) \tlef a \trig \SE(R^\fa))
&&\textrm{by definition of $\SE$} \\
&\equiv (a \sleftand \invs^\tr(\SE(Q^\tr))) \sleftor
  \invs^\fa(\SE(R^\fa))
&&\textrm{by~\eqref{eq:g3}, first case} \\
&\equiv (a \sleftand Q^\tr) \sleftor R^\fa
&&\textrm{as shown above} \\
&\equiv P.
\end{align*}
In the second case we have that
\begin{align*}
\invs^\ell(\SE(P)) &= \invs^\ell(\SE(R^\fa) \tlef a \trig \SE(Q^\tr))
&&\textrm{by definition of $\SE$} \\
&\equiv (\neg a \sleftand \invs^\tr(\SE(Q^\tr))) \sleftor
  \invs^\fa(\SE(R^\fa))
&&\textrm{by~\eqref{eq:g3}, second case} \\
&\equiv (\neg a \sleftand Q^\tr) \sleftor R^\fa
&&\textrm{as shown above} \\
&\equiv P.
\end{align*}

We now prove that for all $*$-terms $P$, $\invs^*(\SE(P)) \equiv P$, by
induction on the number of $\ell$-terms in $P$. In the base case we are
dealing with $\ell$-terms. Because an $\ell$-term has neither a cd nor a dd we
have $\invs^*(\SE(P)) \equiv \invs^\ell(\SE(P)) \equiv P$, 
where the first
identity is by~\eqref{eq:g4} and the second identity was shown above. For the
induction we have either $P \equiv Q \sleftand R$ or $P \equiv Q \sleftor R$.
In the first case note that by Theorem~\ref{thm:scddd}, $\SE(P)$ has a unique
cd and no dd.  So we have 
\begin{align*}
\invs^*(\SE(P)) &\equiv \invs^*(\cd_1(\SE(P))\sub{\Box}{\tr}) \sleftand
  \invs_*(\cd_2(\SE(P)))
&&\textrm{by~\eqref{eq:g4}} \\
&\equiv \invs^*(\SE(Q)) \sleftand \invs^*(\SE(R))
&&\textrm{by Theorem~\ref{thm:scddd}} \\
&\equiv Q \sleftand R
&&\textrm{by induction hypothesis} \\
&\equiv P.
\end{align*}
In the second case, again by Theorem~\ref{thm:scddd}, $P$ has a unique dd and 
no cd. So we have that
\begin{align*}
\invs^*(\SE(P)) &\equiv \invs^*(\dd_1(\SE(P))\sub{\Box}{\fa}) \sleftor
  \invs_*(\dd_2(\SE(P)))
&&\textrm{by~\eqref{eq:g4}} \\
&\equiv \invs^*(\SE(Q)) \sleftor \invs^*(\SE(R))
&&\textrm{by Theorem~\ref{thm:scddd}} \\
&\equiv Q \sleftor R
&&\textrm{by induction hypothesis} \\
&\equiv P.
\end{align*}

Finally, we prove the theorem's statement by making a case distinction on the
grammatical category of $P$. If $P$ is a $\tr$-term, then $\SE(P)$ has only
$\tr$-leaves and hence $\invs(\SE(P)) \equiv \invs^\tr(\SE(P)) \equiv P$,
where the first identity is by definition~\eqref{eq:g5}
of $\invs$ and the second identity was shown
above. If $P$ is a $\fa$-term, then $\SE(P)$ has only $\fa$-leaves and
hence $\invs(\SE(P)) \equiv \invs^\fa(\SE(P)) \equiv P$, where the first
identity is by definition~\eqref{eq:g5} of $\invs$ and the second one
was shown above. 
If $P$ is a $\tr$-$*$-term, then it has both $\tr$ and $\fa$-leaves and hence,
letting $P \equiv Q \sleftand R$,
\begin{align*}
\invs(\SE(P)) &\equiv \invs^\tr(\tsd_1(\SE(P))\sub{\Box}{\tr}) \sleftand
  \invs^*(\tsd_2(\SE(P)))
&&\textrm{by \eqref{eq:g5}} \\
&\equiv \invs^\tr(\SE(Q)) \sleftand \invs^*(\SE(R))
&&\textrm{by Theorem~\ref{thm:stsd}} \\
&\equiv Q \sleftand R
&&\textrm{as shown above} \\
&\equiv P,
\end{align*}
which completes the proof.
\end{proof}

\end{document}